\documentclass[a4paper,11pt,english]{article}
\usepackage{jheppub}

\usepackage[utf8]{inputenc}

\usepackage{tabularx,booktabs,multirow}			
\usepackage[dvipsnames]{xcolor}			
\usepackage{hyperref}					
\usepackage{empheq}
\usepackage{enumitem}
\makeatletter
\g@addto@macro\bfseries{\boldmath}
\makeatother

\usepackage{adjustbox}
\usepackage{tikz,tikz-cd}			
\usetikzlibrary{calc,decorations.markings,decorations.pathmorphing}

\usepackage[textsize=footnotesize]{todonotes}
\presetkeys%
		{todonotes}%
		{color=Apricot}{}%
\usepackage[noabbrev]{cleveref}
\crefname{lemma}{lemma}{lemmata}
\Crefname{lemma}{Lemma}{Lemmata}
\crefname{subsection}{subsection}{subsections}
\Crefname{subsection}{Subsection}{Subsections}
\crefname{conjecture}{conjecture}{conjectures}
\Crefname{conjecture}{Conjecture}{Conjectures}
\crefname{assumption}{assumption}{assumptions}
\Crefname{assumption}{Assumption}{Assumptions}
\crefname{introthm}{theorem}{theorems}
\Crefname{introthm}{Theorem}{Theorems}

\usepackage{amsmath,amssymb,amsthm,mathtools,thmtools}	
\usepackage{braket,stmaryrd,bm}

\theoremstyle{plain}
\newtheorem{theorem}{Theorem}[section]
\newtheorem{proposition}[theorem]{Proposition}

\newtheorem{lemma}[theorem]{Lemma}

\newtheorem{assumption}[theorem]{Assumption}
\theoremstyle{definition}

\theoremstyle{plain}
\newtheorem{introthm}{Theorem}

\newcommand{\de}{\partial}

\newcommand{\iu}{\mathrm{i}}
\DeclareMathOperator*{\Res}{Res}

\newcommand{\cc}{\mathrm{c}}

\DeclareMathOperator{\Tr}{Tr}
\newcommand{\normord}[1]{:\mathrel{#1}:}
\newcommand{\NTr}[1]{\normord{\Tr{#1}}}

\newcommand{\mc}[1]{\mathcal{#1}}
\newcommand{\ms}[1]{\mathsf{#1}}

\newcommand{\Z}{\mathbb{Z}}

\newcommand{\R}{\mathbb{R}}
\newcommand{\C}{\mathbb{C}}
\renewcommand{\P}{\mathbb{P}}

\newcommand{\M}{\mathcal{M}}
\newcommand{\Mbar}{\overline{\mathcal{M}}}

\title{Matrix Correlators as Discrete Volumes of Moduli Space I: \\
Recursion Relations, the BMN-limit and DSSYK}

\author[a]{A. Giacchetto,}
\author[b]{P. Maity,}
\author[c]{E. A. Mazenc}

\affiliation[a]{Departement Mathematik, ETH Zürich \\ CH-8006 Zürich, Switzerland}
\affiliation[b]{Laboratory for Theoretical Fundamental Physics, EPFL \\ CH-1015, Lausanne, Switzerland}
\affiliation[c]{Institut für Theoretische Physik, ETH Zürich \\ CH-8093 Zürich, Switzerland}

\emailAdd{alessandro.giacchetto@math.ethz.ch,  pronobesh.maity@epfl.ch, emazenc@ethz.ch}

\abstract{
    We show certain correlators in generic one-matrix models define a notion of ``discrete'' volumes of the moduli space of Riemann surfaces, generalizing the connection between random matrices and JT gravity. We prove they obey a discrete, Mirzakhani-like recursion relation. Their fundamental discreteness crucially relies upon studying these matrix integrals away from the usual double-scaling limit. In a BMN-like limit of large traces, this recursion universally goes over to a continuous one, and the correlators asymptote to the volumes of Kontsevich. Finally, we demonstrate that the ETH matrix integral for DSSYK furnishes a discrete, $q$-analog of the Weil--Petersson volumes, thereby proving a conjecture due to K.~Okuyama.
} 

\begin{document}

\maketitle
\flushbottom

\setlength\parindent{0em} 
\setlength{\parskip}{.25em} 

\section{Overview \& summary of results}
This first installment establishes three main results:
\begin{enumerate}[label=\Alph*)]
 \setlength\itemsep{0.05mm}
    \item a discrete, manifestly geometric recursion relation for correlators of \textit{pruned} traces in generic one-cut matrix models,
    
    \item this recursion becomes universal and continuous in a BMN-like limit of large powers of the matrices, and
    
    \item a proof that the DSSYK matrix model computes a discrete $q$-analog of the Weil--Petersson volumes, thereby proving a conjecture by K.~Okuyama \cite{Oku23}.
\end{enumerate}

\subsection{A discrete Mirzakhani recursion for pruned correlators}
Correlation functions of resolvents in large $\ms{N}$ matrix models obey topological recursion \cite{Eyn04,AMM05,CEO06,EO07}. Essentially, a clever $1/\ms{N}$ expansion of the Schwinger--Dyson equations for the matrix integral \cite{ginsparg1993lectures,DiFrancesco:1993cyw} shows that the planar one-point function suffices to determine all correlators to all orders in $1/\ms{N}$. 
From an entirely different perspective, certain volumes of the moduli space of Riemann surfaces, such as the Weil--Petersson volumes studied by Mirzakhani \cite{Mir07a,Mir07b}, were also found to follow from a recursion relation. This similarity was elucidated by \cite{EO07+} and provided the starting point of Saad--Shenker--Stanford (SSS) \cite{SSSJTmatrix}, recasting JT gravity and its supersymmetric extensions as a matrix integral \cite{SW20,SaadTasi,TuriaciLesHouches, CJohnsonReview,CliffMaciek,CliffWPsuper}. However, the geometric origins of the recursion kernels appearing in the work of Mirzakhani are somewhat obscure on the matrix model side. 

\begin{figure}[b]
\centering
\tikzset{every picture/.style=thick}
\begin{tikzpicture}[x=1pt,y=1pt,xscale=.68,yscale=.7]
    \draw[shift={(-0.604, 566.341)}, xscale=0.7935, yscale=0.7518](0, 0) .. controls (-0.1451, 4.2965) and (-2.7567, 8.6152) .. (-4.7824, 8.5038) .. controls (-6.8081, 8.3925) and (-8.2478, 3.8512) .. (-8.1027, -0.4453) .. controls (-7.9576, -4.7418) and (-6.2277, -8.7935) .. (-4.2021, -8.6822) .. controls (-2.1764, -8.5708) and (0.1451, -4.2965) .. cycle;
    \draw[shift={(112.621, 535.355)}, xscale=0.7935, yscale=0.7518](0, 0) .. controls (-0.3877, 3.8657) and (-3.075, 8.012) .. (-5.187, 7.4605) .. controls (-7.299, 6.909) and (-8.8357, 1.6597) .. (-8.448, -2.206) .. controls (-8.0603, -6.0717) and (-5.7483, -8.5537) .. (-3.6363, -8.0022) .. controls (-1.5243, -7.4507) and (0.3877, -3.8657) .. cycle;
    \draw[shift={(109.163, 563.699)}, xscale=0.7935, yscale=0.7518](0, 0) .. controls (2.259, -0.032) and (4.7837, 4.053) .. (4.767, 7.7985) .. controls (4.7503, 11.544) and (2.1923, 14.95) .. (-0.0667, 14.982) .. controls (-2.3257, 15.014) and (-4.2857, 11.672) .. (-4.269, 7.9265) .. controls (-4.2523, 4.181) and (-2.259, 0.032) .. cycle;
    \draw[shift={(110.23, 598.997)}, xscale=0.7935, yscale=0.7518](0, 0) .. controls (2.1735, 0.9637) and (3.7122, 5.896) .. (2.6215, 9.3515) .. controls (1.5308, 12.807) and (-2.1892, 14.7857) .. (-4.3627, 13.822) .. controls (-6.5362, 12.8583) and (-7.1632, 8.9523) .. (-6.0725, 5.4968) .. controls (-4.9818, 2.0413) and (-2.1735, -0.9637) .. cycle;
    \draw[shift={(-3.581, 559.881)}, xscale=0.7935, yscale=0.7518](0, 0) .. controls (16.6717, 2.95) and (23.0172, -3.5195) .. (28.7709, -11.5327) .. controls (34.5245, -19.546) and (39.6863, -29.103) .. (46.3058, -34.6228) .. controls (52.9253, -40.1427) and (61.0026, -41.6253) .. (68.45, -40.328) .. controls (75.8975, -39.0307) and (82.7151, -34.9533) .. (91.0177, -34.6642) .. controls (99.3203, -34.375) and (109.1078, -37.874) .. (116.3951, -40.1822) .. controls (123.6823, -42.4905) and (128.4693, -43.608) .. (142.8073, -40.625);
    \draw[shift={(108.863, 541.001)}, xscale=0.7935, yscale=0.7518](0, 0) .. controls (-18.564, 0.469) and (-23.1115, 7.181) .. (-23.8052, 13.463) .. controls (-24.499, 19.745) and (-21.339, 25.597) .. (0.378, 30.191);
    \draw[shift={(109.219, 574.958)}, xscale=0.7935, yscale=0.7518](0, 0) .. controls (-18.846, 1.655) and (-25.906, 5.044) .. (-29.5542, 8.6797) .. controls (-33.2023, 12.3153) and (-33.4387, 16.1977) .. (-29.5698, 20.1448) .. controls (-25.701, 24.092) and (-17.727, 28.104) .. (1.274, 31.973);
    \draw[shift={(106.768, 609.388)}, xscale=0.7935, yscale=0.7518](0, 0) .. controls (-21.88, -5.97) and (-38.8606, -3.3755) .. (-51.1255, -2.3059) .. controls (-63.3905, -1.2363) and (-70.9397, -1.6917) .. (-77.4322, -3.6668) .. controls (-83.9246, -5.642) and (-89.3603, -9.137) .. (-94.0658, -14.4012) .. controls (-98.7714, -19.6653) and (-102.747, -26.6987) .. (-107.3767, -32.548) .. controls (-112.0064, -38.3973) and (-117.2902, -43.0627) .. (-122.3721, -45.6188) .. controls (-127.4539, -48.175) and (-132.3338, -48.622) .. (-140.0627, -48.752);
    \draw[shift={(26.512, 571.466)}, xscale=0.7935, yscale=0.7518](0, 0) .. controls (5.1605, -10.224) and (10.4697, -11.218) .. (14.4718, -10.8655) .. controls (18.4738, -10.513) and (21.1687, -8.814) .. (25.1425, -3.604);
    \draw[shift={(50.089, 588.924)}, xscale=0.7935, yscale=0.7518](0, 0) .. controls (8.5741, -5.307) and (14.8107, -3.712) .. (19.4125, -1.5103) .. controls (24.0143, 0.6915) and (26.9813, 3.5) .. (28.4753, 8.025);
    \draw[shift={(56.913, 550.468)}, xscale=0.7935, yscale=0.7518](0, 0) .. controls (4.7424, -6.899) and (8.7385, -7.643) .. (12.9232, -7.4638) .. controls (17.1079, -7.2845) and (21.4813, -6.182) .. (26.4903, -3.055);
    \draw[shift={(28.657, 568.048)}, xscale=0.7935, yscale=0.7518](0, 0) .. controls (7.2705, 5.299) and (14.3878, 3.332) .. (19.3267, -2.674);
    \draw[shift={(60.319, 546.819)}, xscale=0.7935, yscale=0.7518](0, 0) .. controls (2.639, 5.594) and (16.016, 8.462) .. (18.853, -0.045);
    \draw[shift={(53.572, 587.251)}, xscale=0.7935, yscale=0.7518](0, 0) .. controls (2.946, 9.172) and (12.581, 10.542) .. (22.375, 6.666);
    \draw(171.8853, 568.755) .. controls (174.1823, 568.8717) and (177.0347, 572.1063) .. (177.0275, 575.8002) .. controls (177.0203, 579.494) and (174.1537, 583.647) .. (171.8567, 583.5303) .. controls (169.5597, 583.4137) and (167.8323, 579.0273) .. (167.8395, 575.3335) .. controls (167.8467, 571.6397) and (169.5883, 568.6383) .. cycle;
    \draw(201.886, 588.5763) .. controls (203.4853, 589.7727) and (204.7127, 593.5803) .. (203.6448, 595.6337) .. controls (202.577, 597.687) and (199.214, 597.986) .. (197.6147, 596.7897) .. controls (196.0153, 595.5933) and (196.1797, 592.9017) .. (197.2475, 590.8483) .. controls (198.3153, 588.795) and (200.2867, 587.38) .. cycle;
    \draw(197.404, 559.1175) .. controls (199.3973, 558.7833) and (203.0547, 561.3197) .. (203.7047, 564) .. controls (204.3547, 566.6803) and (201.9973, 569.5047) .. (200.004, 569.8388) .. controls (198.0107, 570.173) and (196.3813, 568.017) .. (195.7313, 565.3367) .. controls (195.0813, 562.6563) and (195.4107, 559.4517) .. cycle;
    \draw(171.266, 568.8) .. controls (180.8433, 567.216) and (189.3387, 564.069) .. (196.752, 559.359);
    \draw(201.002, 569.484) .. controls (193.397, 573.218) and (193.16, 580.328) .. (201.886, 588.576);
    \draw(197.453, 596.659) .. controls (188.8543, 589.8343) and (180.1773, 585.4337) .. (171.422, 583.457);
    \draw(287.6224, 576.7666) .. controls (289.2348, 576.6753) and (291.4234, 579.5939) .. (291.3308, 582.7861) .. controls (291.2381, 585.9783) and (288.8641, 589.4439) .. (287.2518, 589.5353) .. controls (285.6394, 589.6266) and (284.7888, 586.3436) .. (284.8814, 583.1514) .. controls (284.9741, 579.9593) and (286.0101, 576.8579) .. cycle;
    \draw(208.4407, 555.5724) .. controls (210.434, 555.2382) and (214.0914, 557.7746) .. (214.7414, 560.4549) .. controls (215.3914, 563.1352) and (213.034, 565.9596) .. (211.0407, 566.2937) .. controls (209.0474, 566.6279) and (207.418, 564.4719) .. (206.768, 561.7916) .. controls (206.118, 559.1112) and (206.4474, 555.9066) .. cycle;
    \draw(276.4045, 518.9438) .. controls (278.7577, 517.8127) and (283.5773, 519.2833) .. (285.3018, 522.3388) .. controls (287.0263, 525.3943) and (285.6557, 530.0347) .. (283.3025, 531.1658) .. controls (280.9493, 532.297) and (277.6137, 529.919) .. (275.8892, 526.8635) .. controls (274.1647, 523.808) and (274.0513, 520.075) .. cycle;
    \draw(208.107, 555.66) .. controls (229.61, 547.082) and (235.793, 535.5495) .. (241.9072, 529.9609) .. controls (248.0213, 524.3723) and (254.0667, 524.7277) .. (259.3486, 524.3478) .. controls (264.6305, 523.968) and (269.149, 522.853) .. (276.404, 518.944);
    \draw(284.613, 530.065) .. controls (272.258, 547.071) and (275.979, 564.478) .. (288.741, 577.115);
    \draw(287.252, 589.535) .. controls (277.154, 587.151) and (270.4195, 591.123) .. (263.6609, 593.0767) .. controls (256.9023, 595.0303) and (250.1197, 594.9657) .. (244.4412, 593.038) .. controls (238.7627, 591.1103) and (234.1883, 587.3197) .. (231.2033, 583.4157) .. controls (228.2183, 579.5117) and (226.8227, 575.4943) .. (224.1483, 572.7037) .. controls (221.474, 569.913) and (217.521, 568.349) .. (210.074, 566.262);
    \draw(237.714, 567.225) .. controls (242.583, 557.27) and (250.561, 557.064) .. (255.484, 559.28);
    \draw(256.171, 579.705) .. controls (260.109, 570.823) and (268.618, 571.578) .. (274.668, 576.56);
    \draw(247.583, 545.153) .. controls (250.645, 536.6) and (254.9885, 535.608) .. (258.8275, 535.2357) .. controls (262.6665, 534.8635) and (266.001, 535.111) .. (268.694, 536.446);
    \draw(240.047, 563.438) .. controls (245.149, 565.824) and (251.867, 563.462) .. (252.603, 558.369);
    \draw(258.897, 575.615) .. controls (261.69, 580.842) and (269.799, 580.7) .. (271.578, 574.502);
    \draw(250.175, 539.929) .. controls (252.028, 545.008) and (257.0045, 544.4385) .. (260.4093, 542.7088) .. controls (263.814, 540.979) and (265.647, 538.089) .. (266.599, 535.653);
    \draw[shift={(315.439, 559.238)}, xscale=1.2159, yscale=1.1218](0, 0) .. controls (1.98, 0.4132) and (3.982, 4.4522) .. (3.9435, 7.9083) .. controls (3.905, 11.3645) and (1.826, 14.2378) .. (-0.154, 13.8247) .. controls (-2.134, 13.4115) and (-4.015, 9.7118) .. (-3.9765, 6.2557) .. controls (-3.938, 2.7995) and (-1.98, -0.4132) .. cycle;
    \draw[shift={(341.023, 574.128)}, xscale=1.2159, yscale=1.1218](0, 0) .. controls (1.6665, 0.4743) and (3.3588, 3.9727) .. (2.8643, 6.7988) .. controls (2.3698, 9.625) and (-0.3115, 11.779) .. (-1.978, 11.3047) .. controls (-3.6445, 10.8303) and (-4.2962, 7.7277) .. (-3.8017, 4.9015) .. controls (-3.3072, 2.0753) and (-1.6665, -0.4743) .. cycle;
    \draw[shift={(341.078, 550.627)}, xscale=1.2159, yscale=1.1218](0, 0) .. controls (1.9443, -0.0173) and (4.195, 3.0327) .. (4.3018, 5.905) .. controls (4.4087, 8.7773) and (2.3717, 11.472) .. (0.4273, 11.4893) .. controls (-1.517, 11.5067) and (-3.3687, 8.8467) .. (-3.4755, 5.9743) .. controls (-3.5823, 3.102) and (-1.9443, 0.0173) .. cycle;
    \draw[shift={(314.61, 559.234)}, xscale=1.2159, yscale=1.1218](0, 0) .. controls (7.5307, -0.696) and (14.3193, -3.0687) .. (20.366, -7.118);
    \draw[shift={(342.226, 563.441)}, xscale=1.2159, yscale=1.1218](0, 0) .. controls (-5.437, 2.1248) and (-7.9239, 4.9363) .. (-5.1613, 9.4105);
    \draw[shift={(314.732, 574.589)}, xscale=1.2159, yscale=1.1218](0, 0) .. controls (9.88, 4.227) and (15.294, 12.278) .. (20.68, 10.894);
    \draw[shift={(334.715, 571.36)}, xscale=1.2159, yscale=1.1218](0, 0) .. controls (1.0724, 1.623) and (2.7056, 2.4282) .. (4.8995, 2.4156);
    \draw(363.5038, 578.7191) .. controls (365.1703, 579.1934) and (366.8627, 582.6918) .. (366.3682, 585.5179) .. controls (365.8737, 588.3441) and (363.1923, 590.4981) .. (361.5258, 590.0238) .. controls (359.8593, 589.5494) and (359.2077, 586.4468) .. (359.7022, 583.6206) .. controls (360.1967, 580.7944) and (361.8373, 578.2448) .. cycle;
    \draw(363.7993, 546.9999) .. controls (365.7437, 546.9826) and (367.9943, 550.0326) .. (368.1012, 552.9049) .. controls (368.208, 555.7773) and (366.171, 558.4719) .. (364.2267, 558.4893) .. controls (362.2823, 558.5066) and (360.4307, 555.8466) .. (360.3238, 552.9743) .. controls (360.217, 550.1019) and (361.855, 547.0173) .. cycle;
    \draw(364.0972, 558.4865) .. controls (377.3657, 559.4955) and (384, 565.3333) .. (384, 576);
    \draw(363.2704, 578.6721) .. controls (371.7568, 579.5574) and (378.2192, 576.2221) .. (382.6575, 568.6662);
    \draw(361.8945, 590.0857) .. controls (376, 592) and (382, 596) .. (389.6667, 598.6667) .. controls (397.3333, 601.3333) and (406.6667, 602.6667) .. (416.3333, 601.3333) .. controls (426, 600) and (436, 596) .. (448.4864, 600.3585);
    \draw(363.7993, 546.9999) .. controls (384, 548) and (396, 536) .. (409.8352, 532.3554) .. controls (423.6704, 528.7108) and (439.3407, 533.4216) .. (449.0714, 530.5666);
    \draw(450.2647, 557.7688) .. controls (452.28, 558.6503) and (454.007, 563.0877) .. (453.1807, 566.0277) .. controls (452.3543, 568.9677) and (448.9747, 570.4103) .. (446.9593, 569.5288) .. controls (444.944, 568.6473) and (444.293, 565.4417) .. (445.1193, 562.5017) .. controls (445.9457, 559.5617) and (448.2493, 556.8873) .. cycle;
    \draw(451.0047, 588.6) .. controls (453.157, 589.1447) and (455.015, 593.2273) .. (454.3455, 596.158) .. controls (453.676, 599.0887) and (450.479, 600.8673) .. (448.3267, 600.3227) .. controls (446.1743, 599.778) and (445.0667, 596.91) .. (445.7362, 593.9793) .. controls (446.4057, 591.0487) and (448.8523, 588.0553) .. cycle;
    \draw(449.1768, 530.5665) .. controls (451.5053, 530.6297) and (454.5697, 535.0933) .. (454.6577, 538.8647) .. controls (454.7457, 542.636) and (451.8573, 545.715) .. (449.5288, 545.6518) .. controls (447.2003, 545.5887) and (445.4317, 542.3833) .. (445.3437, 538.612) .. controls (445.2557, 534.8407) and (446.8483, 530.5033) .. cycle;
    \draw(449.8868, 545.6375) .. controls (438.1103, 547.5884) and (438.3818, 555.5231) .. (449.485, 557.5983);
    \draw(448.8875, 569.7407) .. controls (432.4397, 574.0518) and (436.7916, 584.2531) .. (451.0778, 588.6199);
    \draw(403.9729, 557.9209) .. controls (406.0039, 550.2199) and (412.3739, 547.1989) .. (420.2019, 550.5909);
    \draw(406.0009, 553.4089) .. controls (408.3569, 557.1839) and (411.0044, 557.2559) .. (413.5694, 556.3269) .. controls (416.1344, 555.3979) and (418.6169, 553.4679) .. (418.3269, 549.8909);
    \draw(400.0543, 583.9956) .. controls (403.1193, 575.4236) and (413.5173, 577.9006) .. (419.6883, 582.7376);
    \draw(402.0833, 580.6506) .. controls (406.7523, 585.4276) and (412.9443, 585.5386) .. (417.1593, 581.0286);
    \draw(518.0105, 579.1541) .. controls (519.677, 579.6284) and (521.3694, 583.1268) .. (520.8749, 585.9529) .. controls (520.3804, 588.7791) and (517.699, 590.9331) .. (516.0325, 590.4588) .. controls (514.366, 589.9844) and (513.7144, 586.8818) .. (514.2089, 584.0556) .. controls (514.7034, 581.2294) and (516.344, 578.6798) .. cycle;
    \draw(518.3067, 547.4349) .. controls (520.2511, 547.4176) and (522.5017, 550.4676) .. (522.6086, 553.3399) .. controls (522.7154, 556.2123) and (520.6784, 558.9069) .. (518.7341, 558.9243) .. controls (516.7897, 558.9416) and (514.9381, 556.2816) .. (514.8312, 553.4093) .. controls (514.7244, 550.5369) and (516.3624, 547.4523) .. cycle;
    \draw(597.1768, 522.5665) .. controls (599.5053, 522.6297) and (602.5697, 527.0933) .. (602.6577, 530.8647) .. controls (602.7457, 534.636) and (599.8573, 537.715) .. (597.5288, 537.6518) .. controls (595.2003, 537.5887) and (593.4317, 534.3833) .. (593.3437, 530.612) .. controls (593.2557, 526.8407) and (594.8483, 522.5033) .. cycle;
    \draw(582.2647, 593.7688) .. controls (584.28, 594.6503) and (586.007, 599.0877) .. (585.1807, 602.0277) .. controls (584.3543, 604.9677) and (580.9747, 606.4103) .. (578.9593, 605.5288) .. controls (576.944, 604.6473) and (576.293, 601.4417) .. (577.1193, 598.5017) .. controls (577.9457, 595.5617) and (580.2493, 592.8873) .. cycle;
    \draw(575.0047, 620.6) .. controls (577.157, 621.1447) and (579.015, 625.2273) .. (578.3455, 628.158) .. controls (577.676, 631.0887) and (574.479, 632.8673) .. (572.3267, 632.3227) .. controls (570.1743, 631.778) and (569.0667, 628.91) .. (569.7362, 625.9793) .. controls (570.4057, 623.0487) and (572.8523, 620.0553) .. cycle;
    \draw(518.062, 579.17) .. controls (533.282, 584.407) and (545.7265, 581.608) .. (556.7675, 582.0988) .. controls (567.8085, 582.5895) and (577.446, 586.37) .. (582.754, 594.055);
    \draw(515.676, 590.309) .. controls (530.103, 601.369) and (532.871, 609.8175) .. (539.9533, 616.7865) .. controls (547.0355, 623.7555) and (558.432, 629.245) .. (572.503, 632.362);
    \draw(517.23, 547.762) .. controls (539.808, 533.732) and (541.798, 526.495) .. (544.7, 522.1308) .. controls (547.602, 517.7667) and (551.416, 516.2753) .. (556.919, 516.927) .. controls (562.422, 517.5787) and (569.614, 520.3733) .. (575.886, 522.5804) .. controls (582.158, 524.7875) and (587.51, 526.407) .. (596.605, 522.641);
    \draw(598.819, 537.396) .. controls (584.434, 544.538) and (580.501, 554.7315) .. (575.0778, 560.5951) .. controls (569.6547, 566.4587) and (562.7413, 567.9923) .. (556.2702, 567.2203) .. controls (549.799, 566.4483) and (543.77, 563.3707) .. (538.0385, 561.4356) .. controls (532.307, 559.5005) and (526.873, 558.708) .. (518.734, 558.924);
    \draw(542.323, 603.983) .. controls (545.388, 595.411) and (555.786, 597.888) .. (561.957, 602.725);
    \draw(550.916, 537.951) .. controls (552.947, 530.25) and (559.317, 527.229) .. (567.145, 530.621);
    \draw(555.856, 554.313) .. controls (559.485, 547.254) and (566.439, 546.117) .. (573.639, 550.008);
    \draw(579.77, 605.752) .. controls (571.168, 606.112) and (568.8015, 608.4225) .. (568.2763, 611.5123) .. controls (567.751, 614.602) and (569.067, 618.471) .. (575.179, 620.652);
    \draw(544.352, 600.638) .. controls (549.021, 605.415) and (555.213, 605.526) .. (559.428, 601.016);
    \draw(558.619, 550.58) .. controls (561.289, 555.18) and (564.603, 554.696) .. (566.8698, 553.7312) .. controls (569.1365, 552.7665) and (570.356, 551.321) .. (571.675, 549.069);
    \draw(552.944, 533.439) .. controls (555.3, 537.214) and (557.9475, 537.286) .. (560.5125, 536.357) .. controls (563.0775, 535.428) and (565.56, 533.498) .. (565.27, 529.921);
    \node at (56, 620) {\small$N_{g,n}$};
    \node at (141, 570) {$\displaystyle= \sum_{m=2}^n$};
    \node at (186, 577) {\footnotesize$B$};
    \node at (252, 620) {\small$N_{g,n-1}$};
    \node at (298, 570) {$+$};
    \node at (327, 568) {\footnotesize$C$};
    \node at (352, 570) {$\Biggl($};
    \node at (408, 620) {\small$N_{g-1,n+1}$};
    \node at (552, 645) {\small$N_{h,1+J}$};
    \node at (552, 505) {\small$N_{h',1+J'}$};
    \node at (612, 570) {$\Biggr)$};
    \node at (484, 555) {$\displaystyle + \sum_{\substack{h,h' \\ J,J'}}$};
\end{tikzpicture}
\caption{
    \textbf{A Pictorial Representation of the Discrete Recursion:} 
    The pruned correlators 
    $\langle \prod_{i=1}^n \frac{1}{b_i} \! \NTr{M^{b_i}} \rangle_{g,\cc}$ 
    define a discrete notion of volume of the moduli space of Riemann surfaces, denoted 
    $N_{g,n}(b_1,\ldots,b_n)$. 
    They satisfy a discrete recursion relation that parallels Mirzakhani's formula for the Weil--Petersson volumes. 
    The recursion kernels $B$ and $C$ can be computed directly from the matrix model potential.
}
\label{fig:disc:rec}
\end{figure}
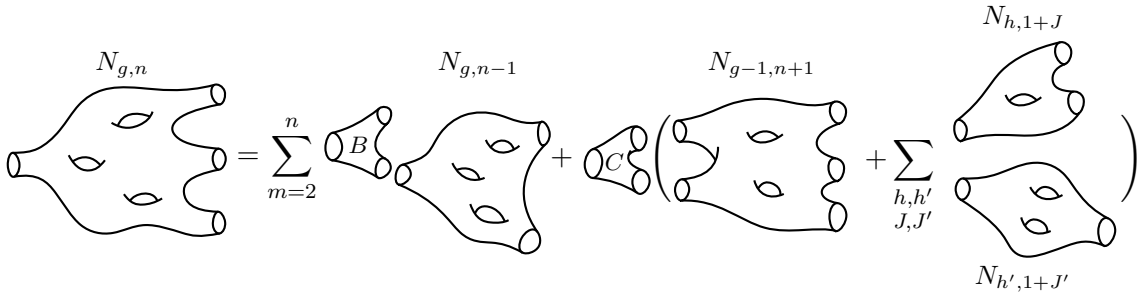

In this work, we derive a recursion relation directly for the connected correlators of \emph{pruned traces} in a generic one-cut matrix model:
\begin{equation}
    N_{g,n}(b_1,\ldots,b_n)
    \coloneqq
    \left\langle \prod_{i=1}^n \frac{1}{b_i} \NTr{M^{b_i}} \right\rangle_{\!\!g,\cc},
    \qquad\qquad
    b_i \in \Z_+ .
\end{equation}
Pruning can be viewed as a planar analog of normal ordering: all planar one-point functions vanish, though not necessarily the higher-genus ones. In a Feynman diagram expansion, this normal ordering corresponds to deleting all petals from the diagrams. Pruned correlators were first introduced by Norbury and Scott \cite{NS13} in the abstract setting of topological recursion, independently of any matrix model realization.

It is the correlator of pruned traces, and \textit{not} the standard ones, which satisfies a Mirzakhani-like recursion relation. In the specific context of the double-scaled JT gravity matrix integral, they play the same role as the Weil--Petersson volumes. Our work makes a more general statement: it goes beyond the specific choice of matrix model and, most importantly, breaks away from the usual double-scaling limit.

While our recursion ultimately follows from the Eynard--Orantin topological recursion for matrix model correlators \cite{EO07}, it possesses two fundamental features that distinguish it. First, it is inherently discrete, replacing the traditional residue calculus with sums over the powers of the matrices appearing in the traces. Second, it makes manifest the geometric content of the recursion, in direct parallel with the Mirzakhani recursion satisfied by the Weil--Petersson volumes, as illustrated in \cref{fig:disc:rec}. The first main result of this paper is:

\begin{introthm}\label{thm:A}
    For $2g-2+n > 1$, the pruned correlators in a generic one-cut matrix model satisfy the recursion relation 
    \begin{multline} \label{eq:disc:rec}
    	N_{g,n}(b_1,\dots,b_n) 
    	=
    	\sum_{m=2}^n \sum_{\beta > 0}
    		\beta \, B(b_1,b_m,\beta)
    		N_{g,n-1}(\beta,b_2,\dots, \widehat{b_m},\dots, b_n) \\
    	+ \frac{1}{2} \sum_{\beta,\beta' > 0}
    		\beta \beta' \, C(b_1,\beta,\beta') \Bigg(
    			N_{g-1,n+1}(\beta,\beta',b_2,\dots,b_n) \\
    			+
    			\sum_{\substack{ h + h' = g \\ J \sqcup J' = \{2,\dots,n\} }}^{\textup{stable}}
    				N_{h,1+|J|}(\beta,b_{J}) N_{h',1+|J'|}(\beta',b_{J'})
    	\Bigg),
    \end{multline}
    where a caret as in $\widehat{b_m}$ denotes omission. The recursion kernels $B$ and $C$ can be expressed in terms of a single building-block function $H$:
    \begin{equation}
    \begin{aligned}
    	B(b,b',\beta)
    	& \coloneqq
    	\frac{1}{2b} \Bigl(
    		H(b + b' - \beta)
    		-
    		H(-b - b' - \beta) \\
    		& \qquad\qquad\qquad
    		+
    		H(b - b' - \beta)
    		-
    		H(-b + b' - \beta)
    	\Bigr), \\
    	C(b,\beta,\beta')
    	& \coloneqq
    	\frac{1}{b} \Bigl(
    		H(b - \beta - \beta')
    		-
    		H(-b - \beta - \beta')
    	\Bigr).
    \end{aligned}
    \end{equation}
    The function $H(\ell)$ is \emph{explicitly} determined from the matrix-model potential $V$, the planar resolvent $W_{0,1}$, and the uniformization of the spectral curve $x(z) = \gamma(z + 1/z) + \delta$, see \cref{eq:resolvent,eq:SC}, via contour integration around cut in the non-physical sheet:
    \begin{equation}\label{eq:H}
        H(\ell)
        \coloneqq
        \frac{1}{2\pi\iu} \oint_{\Gamma} \frac{z^{-2-\ell}}{y(z) \, x'(z)} \, dz,
        \qquad\quad
        y = \frac{1}{2} V' - W_{0,1}.
        \qquad\quad
        \begin{tikzpicture}[baseline={(0, 0cm)}]
            \draw[thick,decoration={markings,mark=at position 0.33 with {\arrow{>}}},postaction={decorate}] (1,0) arc (0:180:1cm);
            \draw[thick,decoration={markings,mark=at position 0.33 with {\arrow{>}}},postaction={decorate}] (-1,0) arc (-180:0:1cm);
            \fill[white] (1,0) circle (.2cm);
            \fill[white] (-1,0) circle (.2cm);
            \draw[->] (-1.5,0) -- (1.5,0);
            \draw[->] (0,-1.3) -- (0,1.3);
            \node at (1,0) {\tiny$\bullet$};
            \node at (1,0) [below right] {\tiny$+1$};
            \node at (-1,0) {\tiny$\bullet$};
            \node at (-1,0) [below left] {\tiny$-1$};
            \node at (1,1) {$\Gamma$};
            \clip (0,0) circle (1cm);
            \draw[thick] (1,0) circle (.2cm);
            \draw[thick] (-1,0) circle (.2cm);
        \end{tikzpicture}
    \end{equation}
    Together with the genus-$0$, $3$-point correlator $N_{0,3}$ and the genus-$1$, $1$-point correlator $N_{1,1}$, the recursion uniquely determines all $N_{g,n}$.
\end{introthm}

This geometric version of topological recursion therefore suggests that the pruned correlators can be viewed as providing a discrete notion of volumes of the moduli space of Riemann surfaces. In this picture, they would compute a weighted count of Riemann surfaces with integer-length boundaries. The discrete boundary lengths correspond to the powers appearing in the traces of the dual matrix model correlators.

We will make this picture precise using a well-established bijection between metrized ribbon graphs and points on the decorated moduli space of Riemann surfaces \cite{Str84,Kon92}. In a nutshell, we expand the correlators in terms of Feynman diagrams, and map each diagram to a point on the moduli space. This notion of discreteness exists at each order in $1/\ms{N}$, and we argue it persists to any finite order in perturbation theory in the interaction coupling---see \cref{sec:matrix}.

\subsection{The BMN-like limit and its Airy universality}
Our recursion relation also reveals the existence of a particularly interesting and universal limit of pruned correlators. It reflects the well-known Airy universality that governs the square-root vanishing of the matrix eigenvalue distribution near its endpoints \cite{TW94,Eyn16,BrezinHikamiAiry,BrezinZee}.

In essence, this limit consists in taking the powers of the matrices inside each trace to be very large. This closely parallels the Berenstein--Maldacena--Nastase (BMN) limit in AdS/CFT \cite{BMN02}. Geometrically, it corresponds to sending the boundaries of the dual Riemann surfaces to infinity. Using the recursion relation, we will demonstrate that for any one-cut matrix model, the pruned correlators converge in this limit to the Kontsevich volumes, independently of the potential. 

\begin{introthm}\label{thm:B}
    For $2g-2+n > 0$, the pruned correlators in a generic one-cut matrix model satisfy
    \begin{equation}\label{eq:BMN}
        \lim_{\substack{t \rightarrow 0^+ \\ t^{-1}(L_1 + \cdots + L_n) \\ \textup{ even/odd}}} \!\!
            t^{2(3g-3+n)}
            \left\langle \prod_{i=1}^n \frac{1}{L_i/t} \NTr{M^{L_i/t}} \right\rangle_{\!\!g,\cc}
        =
        \bigl( \sigma_+^{2g-2+n} \pm \sigma_-^{2g-2+n} \bigr)
        V^{\textup{Kon}}_{g,n}(L_1,\ldots,L_n) .
    \end{equation}
    Here $V^{\textup{Kon}}_{g,n}$ are the Kontsevich volumes of the moduli space of Riemann surfaces and $\sigma_{\pm}$ are scaling constants that depend on the matrix model spectral curve near endpoints of the eigenvalue distribution.
\end{introthm}

In the BMN-like limit, the discrete recursion universally goes over to a continuous one, known to govern the celebrated Kontsevich volumes \cite{BCSW12,ABCGLW}. For the precise definition of these volumes and their expression as integrals over $\Mbar_{g,n}$, see \cref{sec:volumes}. The scaling constants originate from the two endpoints of the large $\ms{N}$ eigenvalue distribution, which characterize the one-cut phase of the underlying matrix model.

To heuristically motivate this result from a diagrammatic point of view, note that as the powers $b_i$ of the matrices grow, the contributing Feynman diagrams become dominated by contractions between external legs, thus washing away the details of the underlying potential. Taking a moduli space vantage point, the large number of Wick-contractions effectively translates to filling the moduli space with more and more discrete points, see \cref{fig:BMN}. In this limit, the dependence on the matrix model potential disappears, and the pruned correlators converge to the continuous Kontsevich volumes.

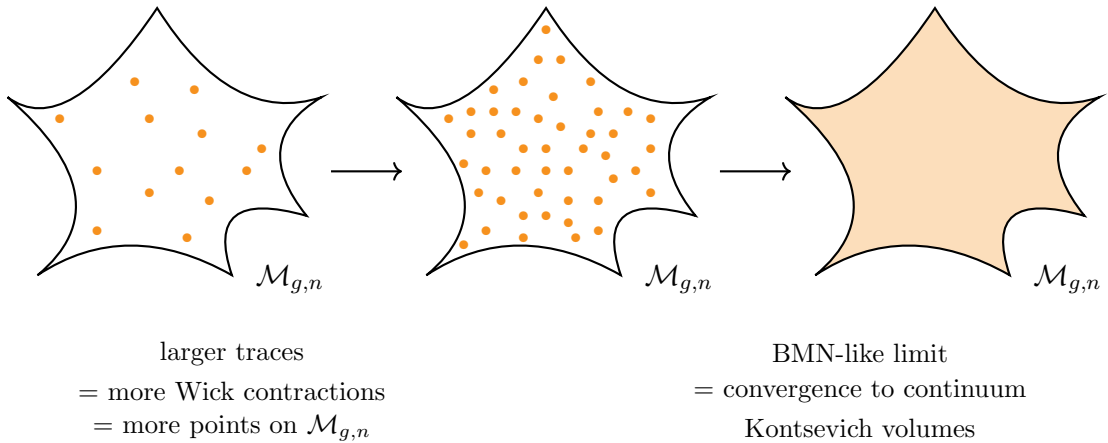
\begin{figure}[t]
    \centering
    \tikzset{every picture/.style=thick}
    \begin{tikzpicture}[x=1pt,y=1pt,scale=.7]
        \draw(24, 392) .. controls (56, 413.3333) and (90.6667, 413.3333) .. (128, 392) .. controls (117.3333, 424) and (130.6667, 434.6667) .. (168, 424) .. controls (146.6667, 450.6667) and (149.3333, 472) .. (176, 488) .. controls (144, 477.3333) and (114.6667, 493.3333) .. (88, 536) .. controls (56, 488) and (29.3333, 472) .. (8, 488) .. controls (50.6667, 450.6667) and (56, 418.6667) .. (24, 392) -- cycle;
        \draw(232, 392) .. controls (264, 413.3333) and (298.6667, 413.3333) .. (336, 392) .. controls (325.3333, 424) and (338.6667, 434.6667) .. (376, 424) .. controls (354.6667, 450.6667) and (357.3333, 472) .. (384, 488) .. controls (352, 477.3333) and (322.6667, 493.3333) .. (296, 536) .. controls (264, 488) and (237.3333, 472) .. (216, 488) .. controls (258.6667, 450.6667) and (264, 418.6667) .. (232, 392) -- cycle;
        \filldraw[fill=BurntOrange,fill opacity=.3](440, 392) .. controls (472, 413.3333) and (506.6667, 413.3333) .. (544, 392) .. controls (533.3333, 424) and (546.6667, 434.6667) .. (584, 424) .. controls (562.6667, 450.6667) and (565.3333, 472) .. (592, 488) .. controls (560, 477.3333) and (530.6667, 493.3333) .. (504, 536) .. controls (472, 488) and (445.3333, 472) .. (424, 488) .. controls (466.6667, 450.6667) and (472, 418.6667) .. (440, 392) -- cycle;
        \node[BurntOrange] at (56, 416) {\scriptsize$\bullet$};
        \node[BurntOrange] at (56, 448) {\scriptsize$\bullet$};
        \node[BurntOrange] at (36, 476) {\scriptsize$\bullet$};
        \node[BurntOrange] at (76, 496) {\scriptsize$\bullet$};
        \node[BurntOrange] at (84, 476) {\scriptsize$\bullet$};
        \node[BurntOrange] at (112, 468) {\scriptsize$\bullet$};
        \node[BurntOrange] at (108, 492) {\scriptsize$\bullet$};
        \node[BurntOrange] at (144, 460) {\scriptsize$\bullet$};
        \node[BurntOrange] at (136, 448) {\scriptsize$\bullet$};
        \node[BurntOrange] at (116, 432) {\scriptsize$\bullet$};
        \node[BurntOrange] at (104, 412) {\scriptsize$\bullet$};
        \node[BurntOrange] at (100, 448) {\scriptsize$\bullet$};
        \node[BurntOrange] at (84, 436) {\scriptsize$\bullet$};
        \node[BurntOrange] at (264, 416) {\scriptsize$\bullet$};
        \node[BurntOrange] at (264, 448) {\scriptsize$\bullet$};
        \node[BurntOrange] at (244, 476) {\scriptsize$\bullet$};
        \node[BurntOrange] at (284, 496) {\scriptsize$\bullet$};
        \node[BurntOrange] at (292, 476) {\scriptsize$\bullet$};
        \node[BurntOrange] at (320, 468) {\scriptsize$\bullet$};
        \node[BurntOrange] at (320, 496) {\scriptsize$\bullet$};
        \node[BurntOrange] at (352, 460) {\scriptsize$\bullet$};
        \node[BurntOrange] at (344, 448) {\scriptsize$\bullet$};
        \node[BurntOrange] at (324, 432) {\scriptsize$\bullet$};
        \node[BurntOrange] at (312, 412) {\scriptsize$\bullet$};
        \node[BurntOrange] at (308, 448) {\scriptsize$\bullet$};
        \node[BurntOrange] at (292, 436) {\scriptsize$\bullet$};
        \node[BurntOrange] at (284, 412) {\scriptsize$\bullet$};
        \node[BurntOrange] at (284, 424) {\scriptsize$\bullet$};
        \node[BurntOrange] at (272, 432) {\scriptsize$\bullet$};
        \node[BurntOrange] at (260, 436) {\scriptsize$\bullet$};
        \node[BurntOrange] at (252, 452) {\scriptsize$\bullet$};
        \node[BurntOrange] at (252, 408) {\scriptsize$\bullet$};
        \node[BurntOrange] at (256, 468) {\scriptsize$\bullet$};
        \node[BurntOrange] at (256, 480) {\scriptsize$\bullet$};
        \node[BurntOrange] at (268, 480) {\scriptsize$\bullet$};
        \node[BurntOrange] at (268, 492) {\scriptsize$\bullet$};
        \node[BurntOrange] at (280, 480) {\scriptsize$\bullet$};
        \node[BurntOrange] at (292, 508) {\scriptsize$\bullet$};
        \node[BurntOrange] at (304, 508) {\scriptsize$\bullet$};
        \node[BurntOrange] at (300, 488) {\scriptsize$\bullet$};
        \node[BurntOrange] at (296, 524) {\scriptsize$\bullet$};
        \node[BurntOrange] at (304, 472) {\scriptsize$\bullet$};
        \node[BurntOrange] at (296, 460) {\scriptsize$\bullet$};
        \node[BurntOrange] at (296, 448) {\scriptsize$\bullet$};
        \node[BurntOrange] at (280, 448) {\scriptsize$\bullet$};
        \node[BurntOrange] at (284, 460) {\scriptsize$\bullet$};
        \node[BurntOrange] at (272, 468) {\scriptsize$\bullet$};
        \node[BurntOrange] at (308, 432) {\scriptsize$\bullet$};
        \node[BurntOrange] at (296, 424) {\scriptsize$\bullet$};
        \node[BurntOrange] at (308, 420) {\scriptsize$\bullet$};
        \node[BurntOrange] at (332, 444) {\scriptsize$\bullet$};
        \node[BurntOrange] at (328, 456) {\scriptsize$\bullet$};
        \node[BurntOrange] at (316, 460) {\scriptsize$\bullet$};
        \node[BurntOrange] at (332, 468) {\scriptsize$\bullet$};
        \node[BurntOrange] at (340, 480) {\scriptsize$\bullet$};
        \node[BurntOrange] at (324, 480) {\scriptsize$\bullet$};
        \node[BurntOrange] at (352, 476) {\scriptsize$\bullet$};
        \node[BurntOrange] at (352, 436) {\scriptsize$\bullet$};
        \node[BurntOrange] at (324, 416) {\scriptsize$\bullet$};
        \node at (159, 390) {$\mathcal{M}_{g,n}$};
        \node at (367, 390) {$\mathcal{M}_{g,n}$};
        \node at (575, 390) {$\mathcal{M}_{g,n}$};

        \draw[->] (181, 448) -- (219, 448);
        \draw[->] (389, 448) -- (427, 448);
        
        \node at (128, 350) {\small{larger traces}};
        \node at (128, 330) {\small{= more Wick contractions}};
        \node at (128, 310) {\small{= more points on $\M_{g,n}$}};
        
        \node at (464, 350) {\small{BMN-like limit}};
        \node at (464, 330) {\small{= convergence to continuum}};
        \node at (464, 310) {\small{Kontsevich volumes}};
    \end{tikzpicture}
    \caption{\textbf{A Moduli Space Perspective on the BMN-like Limit:}  In the limit of large traces, many more Wick contractions are possible. Since each such Feynman diagram maps onto one point on $\M_{g,n}$, more and more points populate the moduli space. Our recursion relation proves that, generically, the $N_{g,n}$ converge to the well-known continuum Kontsevich volumes, $V^{\textup{Kon}}_{g,n}$.}
    \label{fig:BMN}
\end{figure}

This result was previously established by Norbury and Scott within the framework of spectral curve topological recursion \cite{NS13}, where the Airy correlators naturally appear in this limit. Our approach provides an independent derivation based entirely on the discrete recursion of \cref{thm:A}.

\subsection{DSSYK matrix correlators as discrete \texorpdfstring{$q$}{q}-Weil--Petersson volumes}
The third main result of this paper pertains to the pruned matrix correlators in the ETH-matrix description \cite{SonnerETHMatrix} of double-scaled SYK (DSSYK), introduced in \cite{JKMS23}. We will refer to this model as the DSSYK matrix integral, which is not double-scaled and is studied in the conventional 't~Hooft limit. It plays the same role to DSSYK as the matrix integral of SSS \cite{SSSJTmatrix} plays to the usual SYK model. 

In \cite{Oku23}, Okuyama explicitly computed several low-order correlators and noticed their striking similarity to Weil--Petersson volumes. He conjectured that a particular \textit{combined} limit, where the power of the matrices is sent to infinity while the DSSYK double-scaling parameter $\lambda$ is sent to zero, would precisely recover these Weil--Petersson volumes. This limit is more involved than the BMN-like limit of the previous section, since the underlying matrix model potential is tuned simultaneously. We prove his conjecture using the discrete recursion outlined above. Since the underlying potential is even, only the cases where the matrix powers sum to an even integer are non-trivial.

\begin{introthm}[Okuyama's conjecture]\label{thm:C}
    For $2g-2+n > 0$, the DSSYK correlators satisfy
    \begin{equation} \label{eq:OkConj}
        \lim_{\lambda \rightarrow 0^+} \
        (2 \, (q)_{\infty}^3)^{2g-2+n}
        \lambda^{2(3g-3+n)} 
        \left\langle 
            \prod_{i=1}^n 
            \frac{1}{L_i/\lambda} 
            \NTr{M^{L_i/\lambda}} 
        \right\rangle_{\!\!g,\cc}^{\!\!\textup{DSSYK}}
        =
        2 \cdot V^{\textup{WP}}_{g,n}(L_1,\ldots,L_n)
    \end{equation}
    whenever the sum of $L_i/\lambda \in \mathbb{Z}_+$ is even. Here $q = e^{-\lambda}$ is the double-scaling parameter of the underlying DSSYK model, $(q)_{\infty} = \prod_{k \ge 1} (1-q^k)$, and $V^{\textup{WP}}_{g,n}$ are the Weil--Petersson volumes of the moduli space of Riemann surfaces.
\end{introthm}

This is the precise sense in which the DSSYK matrix integral correlators furnish a discrete, $q$-analog of the Weil--Petersson volumes. 
For the definition of the Weil--Petersson volumes and their expression as integrals over $\Mbar_{g,n}$, see \cref{sec:volumes}; For more details on the DSSYK matrix model, see \cref{sec:DSSYK:rec}.

Our recursion relation is particularly suited to this combined limit. While the spectral curve of the DSSYK matrix integral is relatively complicated, being expressed in terms of a Jacobi theta function, the building-block function of the recursion kernel for the DSSYK model takes a remarkably simple form:
\begin{equation}
	H_q(b)
	=
	\frac{2}{(q)_{\infty}^3} \sum_{k \ge 1}
		(-1)^{k+1} q^{\frac{k(k+1)}{2}}
		\frac{q^{-\frac{kb}{2}}}{1 - q^k},
    \qquad
    b \in \Z.
\end{equation}
This series furnishes a $q$-analog of 
\begin{equation}
    H(\ell) = 2 \log(1 + e^{\ell/2}),
    \qquad
    \ell \in \R,
\end{equation}
the building-block function appearing in the recursion relation for the continuum Weil--Petersson volumes. From the explicit expression, we immediately see that $H_q$ computed from the DSSYK matrix integral reduces to the continuum version in the combined limit:
\begin{equation}
	\lim_{\lambda \to 0^+} \,
	   (q)_{\infty}^3 \, \lambda \,
       H_q\Bigl( \frac{\ell}{\lambda} \Bigr)
	=
	2 \sum_{k \ge 1} \frac{(-1)^{k+1}}{k} e^{\ell k/2}
	=
	2 \log(1 + e^{\ell/2}) 
	= H(\ell) .
\end{equation}
Together with the fact that the base cases of our recursion relation (corresponding to the topologies of a pair of pants and a one-holed torus) also flow to their continuum Weil--Petersson counterparts, this establishes Okuyama's conjecture for all $g$ and $n$.

\paragraph{Note added.} The posting of this work was coordinated with the authors of \cite{DN}, who received an early draft of our paper in late September 2025. Unlike their work, the boundary lengths appearing in our volumes are discrete and coincide with the correlators computed by the DSSYK matrix integral.

\section{Moduli space of Riemann surfaces: a tale of three volumes}
\label{sec:volumes}
In this section, we review the Weil--Petersson volumes, the Kontsevich volumes, and the discrete Norbury volumes of the moduli space of Riemann surfaces, as well as their recursive computation and its geometric origin. These compute, respectively, the volumes of the moduli space of hyperbolic metrics, the volumes of the moduli space of flat Strebel metrics, and the number of lattice points on the moduli space of flat Strebel metrics. All three are connected to the moduli space of Riemann surfaces, i.e.~the moduli space of complex structures. Here we draw a parallel following \cite{Gia21}.

From a physical perspective, the Weil--Petersson volumes play a central role in JT gravity, where they describe the geometry of the moduli space of hyperbolic surfaces contributing to the gravitational path integral \cite{SSSJTmatrix}. The Kontsevich volumes were introduced to prove Witten's conjecture relating intersection theory on the moduli space of curves to topological quantum gravity \cite{Wit91,Kon92}. Finally, the discrete Norbury volumes provide a discretization of the latter \cite{Nor10}, and are closely related to the Gaussian Unitary Ensemble (GUE), as will be discussed in the next section.

\subsection{A warm-up analogy}
Before reviewing these volumes and their geometric origins, let us illustrate an analogy. Consider the topological $2$-sphere. There are two different, yet equally meaningful, models for this space:
\begin{itemize}
    \item The \emph{smooth model}: the round sphere of radius $L$, defined as
    \begin{equation}
        S_L \coloneqq \Set{ (x,y,z) \in \R^3 | x^2 + y^2 + z^2 = L^2 }.
    \end{equation}
    
    \item The \emph{combinatorial model}: the surface of the cube of side $L$, defined as
    \begin{equation}
        C_L \coloneqq \Set{ (x,y,z) \in \R^3 | \max\{ |x|,|y|,|z| \} = \frac{L}{2} }.
    \end{equation}
\end{itemize}

The two models are topologically equivalent, yet each carries its own intrinsic geometry. 
The smooth model has a natural notion of symplectic area, obtained by integrating the canonical $2$-form on $S_L$:
\begin{equation}
    \mathrm{Area}(S_L)
    \coloneqq
    \int_{0}^{2\pi} \int_{0}^{\pi} L^2 \sin\theta \, d\theta \wedge d\varphi
    = 4\pi L^2.
\end{equation}
The combinatorial model, on the other hand, admits a different notion of area, obtained by summing the areas of its six faces:
\begin{equation}
    \mathrm{Area}(C_L)
    \coloneqq
    6
    \int_{-\frac{L}{2}}^{\frac{L}{2}} \int_{-\frac{L}{2}}^{\frac{L}{2}}
        dx \wedge dy
    =
    6 L^2.
\end{equation}
This is again a symplectic volume, where the symplectic form is obtained by gluing the Darboux form $dx \wedge dy$ on each face of the cube's surface.

When the side of the cube is an integer, which to avoid confusion we denote as $b$, the combinatorial model also has a notion of integral points: the points in $C_b^{\Z} \coloneqq C_b \cap \Z^3$, i.e.~points on the surface of the cube with integer coordinates. In this case, it makes sense to define the discrete area of the cube's surface as the number of such integral points:
\begin{equation}
    \# C_b^{\Z}
    =
    \frac{1 + (-1)^b}{2}
    \bigl(
        6 b^2 + 2
    \bigr).
\end{equation}
The factor $\frac{1 + (-1)^b}{2}$ enforces that $b$ is an even integer: otherwise, $C_b^{\Z}$ is empty, since at least one coordinate would be a half-integer.  
Notice that this discrete area already encodes information about the continuous one: the leading term of $6b^2 + 2$ is precisely $6b^2$, that is,
\begin{equation}
    \lim_{t \to 0^+} \
        t^{2} \cdot \# C_{L/t}^{\Z}
    =
    \mathrm{Area}(C_L) .
\end{equation}
This is not a coincidence but an instance of the general correspondence between lattice point counting and integration, known as Ehrhart theory. Geometrically, it expresses the fact that by counting lattice points on an increasingly finer mesh, one recovers the continuous volume in the limit. It can also be seen as a convergence of the Dirac delta measure on the rescaled lattice points on the cube's surfaces to the Euclidean measure on the top-dimensional faces.

To summarize, the same topological space, namely the topological $2$-sphere, admits two natural geometries, smooth and combinatorial, the first with a natural notion of area, and the second with both an area and a discrete area.

\begin{center}
    \renewcommand{\arraystretch}{1.2}
    \begin{tabular}{>{\centering\arraybackslash}p{.3\textwidth}|>{\centering\arraybackslash}p{.3\textwidth}|>{\centering\arraybackslash}p{.3\textwidth}}
        \toprule
         smooth model & \multicolumn{2}{c}{combinatorial model} \\
         \midrule
         round sphere $S_L$ &
         cube $C_L$ &
         integer points on cube $C_b^{\Z}$ \\
        \begin{tikzpicture}[x=1pt,y=1pt,scale=.4]
            \fill[opacity=.05] (410.3177, 584.4762)  arc[start angle=-115.3539, end angle=-25.3772, radius=88]  .. controls (495.8362, 598.0951) and (456.7726, 584.1587) .. (410.3177, 584.4762)  -- cycle;
            \draw[opacity=.5] (360.3662, 672.0196)  .. controls (360, 664) and (376, 656) .. (393.3333, 650.6667)  .. controls (410.6667, 645.3333) and (429.3333, 642.6667) .. (448, 642.6667)  .. controls (466.6667, 642.6667) and (485.3333, 645.3333) .. (502.6667, 650.6667)  .. controls (520, 656) and (536, 664) .. (535.6345, 672.0119);
        \draw[thick] (448, 664) circle[radius=88];
        \end{tikzpicture}
        &
        \begin{tikzpicture}[x=1pt,y=1pt,scale=.4]
            \draw[thick,opacity=.5](128, 576) -- (164, 640) -- (284, 640);
            \draw[thick,opacity=.5](164, 640) -- (164, 752);
            \draw[thick] (128, 576) rectangle (256, 704);
            \draw[thick] (256, 576) -- (284, 640) -- (284, 752) -- (164, 752) -- (128, 704);
            \draw[thick] (256, 704) -- (284, 752);
        \end{tikzpicture}
        &
        \begin{tikzpicture}[x=1pt,y=1pt,scale=.4]
            \draw[opacity=.5](128, 576) -- (164, 640) -- (284, 640);
            \draw[opacity=.5](164, 640) -- (164, 752);
            \node[opacity=.5] at (224, 640) {\tiny$\bullet$};
            \node[opacity=.5] at (164, 696) {\tiny$\bullet$};
            \node[opacity=.5] at (144, 672) {\tiny$\bullet$};
            \node[opacity=.5] at (224, 696) {\tiny$\bullet$};
            \node[opacity=.5] at (208, 608) {\tiny$\bullet$};
            \node[opacity=.5] at (164, 640) {\tiny$\bullet$};
            \node[opacity=.5] at (145.9942, 607.9897) {\tiny$\bullet$};
            \node at (128, 576) {\tiny$\bullet$};
            \node at (192, 576) {\tiny$\bullet$};
            \node at (256, 576) {\tiny$\bullet$};
            \node at (128, 640) {\tiny$\bullet$};
            \node at (192, 640) {\tiny$\bullet$};
            \node at (256, 640) {\tiny$\bullet$};
            \node at (256, 704) {\tiny$\bullet$};
            \node at (192, 704) {\tiny$\bullet$};
            \node at (128, 704) {\tiny$\bullet$};
            \node at (164, 752) {\tiny$\bullet$};
            \node at (284, 752) {\tiny$\bullet$};
            \node at (224, 752) {\tiny$\bullet$};
            \node at (284, 696) {\tiny$\bullet$};
            \node at (272, 672) {\tiny$\bullet$};
            \node at (269.9998, 607.9996) {\tiny$\bullet$};
            \node at (284, 640) {\tiny$\bullet$};
            \node at (270.0026, 728.0045) {\tiny$\bullet$};
            \node at (146.0057, 728.0076) {\tiny$\bullet$};
            \node at (208, 728) {\tiny$\bullet$};
            \draw(128, 576) rectangle (256, 704);
            \draw(256, 576) -- (284, 640) -- (284, 752) -- (164, 752) -- (128, 704);
            \draw(256, 704) -- (284, 752);
        \end{tikzpicture}
         \\
         \midrule
         sphere's area & cube's area & cube's lattice point count \\
         $\mathrm{Area}(S_L)=4\pi L^2$ & $\mathrm{Area}(C_L)=6L^2$ & $\#C_b^{\Z} = \frac{1+(-1)^b}{2}(6b^2 + 2)$\\
         \bottomrule
    \end{tabular}
\end{center}

It is also worth mentioning that there exists a third model of the $2$-sphere: the \emph{complex-geometric} one, namely the projective line $\P^1 = \C \cup \set{\infty}$. On $\P^1$, it is natural to consider cohomology classes defined by complex vector bundles. A natural question then arises: can we express the cohomology classes representing the symplectic forms from the smooth and combinatorial models under the appropriate isomorphisms? The answer is yes. For instance, via the stereographic projection, one can identify the cohomology class of the symplectic form $L^2 \sin\theta \, d\theta \wedge d\varphi$ on the sphere with the first Chern class of the hyperplane line bundle $\mc{O}(1) \to \P^1$:
\begin{equation}
    \mathrm{Area}(S_L)
    = 
    4\pi L^2 \int_{\P^1} c_1\bigl( \mc{O}(1) \bigr).
\end{equation}
This equality is rather striking: the left-hand side is intrinsic to the differential-geometric nature of the round sphere $S_L$, while the right-hand side is intrinsic to the complex geometry of $\P^1$ and, a priori, it has nothing to do with a notion of area. The comparison is made possible only through the isomorphism between $\P^1$ and $S_L$. A similar correspondence can be established with the area form arising from the combinatorial model.

\subsection{Riemann surfaces and their volumes}
We can now move on to a more intricate example that exhibits all the features discussed above: the moduli space of Riemann surfaces, its different models, and the corresponding notions of volume.

The ``complex-geometric model'' is the moduli space parameterizing complex structures up to biholomorphism, denoted simply by $\M_{g,n}$ and called the moduli space of complex curves (see \cite{GL26} for a physics oriented account on the subject):
\begin{equation}
    \M_{g,n}
    \coloneqq
    \left.
    \Set{\substack{
        \displaystyle\text{complex structures on a surface of genus $g$} \\[1ex]
        \displaystyle\text{with $n$ marked points}
    }}
    \right/ \sim.
\end{equation}
As in the toy example above, there are two different but meaningful models one can consider of the same moduli space, which depend on the additional data of boundary lengths $L_1,\ldots,L_n \in \R_+$:
\begin{itemize}
    \item The ``smooth model'': the moduli space of \emph{hyperbolic metrics} with geodesic boundaries up to isometry:
    \begin{equation}
        \M_{g,n}^{\textup{hyp}}(L_1,\ldots,L_n)
        \coloneqq
        \left.
        \Set{\substack{
            \displaystyle\text{hyperbolic metrics on a surface of genus $g$} \\[1ex]
            \displaystyle\text{with $n$ geodesic boundaries of lengths $(L_1,\ldots,L_n)$}
        }}
        \right/ \sim.
    \end{equation}

    \item The ``combinatorial model'': the moduli space of \emph{Strebel graphs} (also known as metrized ribbon graphs) up to isometry:
    \begin{equation}
        \M_{g,n}^{\textup{comb}}(L_1,\ldots,L_n)
        \coloneqq
        \left.
        \Set{\substack{
            \displaystyle\text{Strebel graphs on a surface of genus $g$} \\[1ex]
            \displaystyle\text{with $n$ boundaries of lengths $(L_1,\ldots,L_n)$}
        }}
        \right/ \sim.
    \end{equation}
\end{itemize}
The reason why these spaces are isomorphic follows from two classical theorems due to Riemann and Strebel. The first, the \emph{uniformization theorem}, asserts that for every complex structure there exists a unique hyperbolic metric with prescribed geodesic boundary lengths. The second, \emph{Strebel's theorem}, guarantees that for each complex structure there exists a unique Strebel differential with prescribed residues.

In contrast to the analytic description given by uniformization, Strebel's theorem leads to a purely combinatorial one. It identifies each Riemann surface with a metrized ribbon graph, a \emph{Strebel graph}, thereby providing a combinatorial model of the moduli space that is perhaps less familiar in the physics literature. For each point of the moduli space $\mathcal{M}_{g,n}$ and each vector of positive real numbers $(L_1, \ldots, L_n)$, there exists a unique meromorphic quadratic differential $\phi(z)\,(dz)^2$, called the \emph{Strebel differential}, satisfying certain properties; see \cite{Str84, MP98, Muk03}. Its only singularities are double poles whose residues are the prescribed positive real numbers~$L_i$. The Strebel differential foliates the Riemann surface into a family of curves known as horizontal trajectories, shown in red in the central panel of \cref{fig:StrebelCurves}. Along these curves, the square root of the differential is purely real. In general, the horizontal trajectories form closed concentric loops, whose limit set defines a canonical graph on the surface, the \emph{Strebel graph}, depicted in orange in \cref{fig:StrebelCurves}. This graph can be embedded in the surface by replacing each vertex with a small disk and each edge with a thin ribbon, hence the name \emph{ribbon graph}. The vertices of this graph correspond to the zeros of the differential, and the valence of each vertex equals the order of the zero plus two; in particular, all vertices are at least trivalent, a property that will play an important role in \cref{subsec:GUE:correl}.

The differential induces a natural metric on the surface,
\begin{equation}
    ds^2_{\textup{Strebel}} \coloneqq |\phi| \, dz\, d\bar{z},
\end{equation}
which is flat almost everywhere, except at curvature singularities at the zeros and poles of $\phi$. Each edge of the Strebel graph acquires a length $\ell_e$ by integrating the line element along a horizontal trajectory between two zeros, hence the name \emph{metrized ribbon graph}. The continuous moduli of the Riemann surface are encoded in these edge lengths $\ell_e$, which provide a combinatorial parametrization of the moduli space $\M_{g,n}$, denoted $\mathcal{M}_{g,n}^{\textup{comb}}(L_1,\ldots,L_n)$.

Geometrically, one can view each surface as being composed of semi-infinite cylinders glued along the Strebel graph. The circumferences of these cylinders correspond to the boundary lengths $L_i$; see the right panel of \cref{fig:StrebelCurves}. This construction makes it manifest that the horizontal trajectories are geodesics with respect to the Strebel metric. For a physicist-friendly introduction to Strebel's construction, see sections~2.2--2.4 of \cite{GKKMS}.

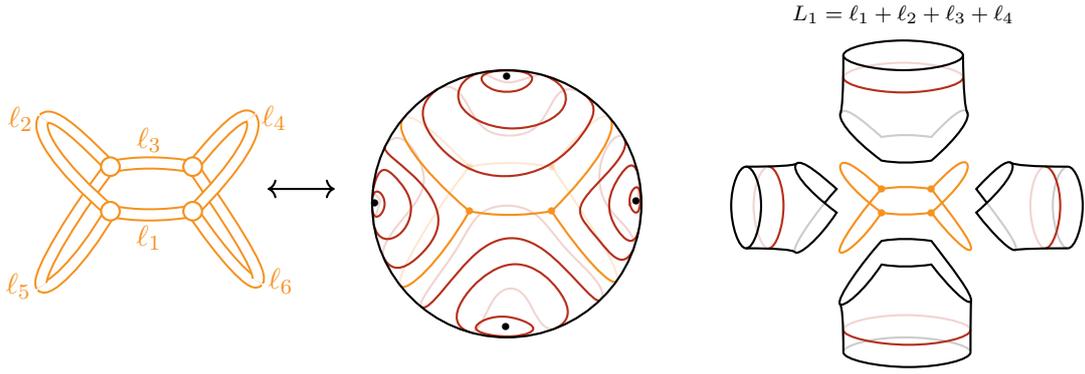
\begin{figure}[t]
    \centering
    \tikzset{every picture/.style=thick}
    \begin{tikzpicture}[x=1pt,y=1pt,scale=.7]
        \draw[BurntOrange,line width=5pt](68, 404) .. controls (81.3333, 401.3333) and (96, 401.3333) .. (112, 404);
        \draw[white,line width=3.5pt](68, 404) .. controls (81.3333, 401.3333) and (96, 401.3333) .. (112, 404);
        
        \draw[BurntOrange,line width=5pt](68, 404) .. controls (52, 376) and (34.431, 359.893) .. (31.238, 363.705);
        \draw[white,line width=3.5pt](68, 404) .. controls (52, 376) and (34.431, 359.893) .. (31.238, 363.705);
        
        \draw[BurntOrange,line width=5pt](112, 404) .. controls (128, 380) and (143.937, 362.667) .. (146.735, 366.356);
        \draw[white,line width=3.5pt](112, 404) .. controls (128, 380) and (143.937, 362.667) .. (146.735, 366.356);
        
        \draw[BurntOrange,line width=5pt](32.318, 453.645) .. controls (37.519, 459.339) and (56, 440) .. (68, 428);
        \draw[white,line width=3.5pt](32.318, 453.645) .. controls (37.519, 459.339) and (56, 440) .. (68, 428);
        
        \draw[BurntOrange,line width=5pt](68, 428) .. controls (81.3333, 430.6667) and (96, 430.6667) .. (112, 428);
        \draw[white,line width=3.5pt](68, 428) .. controls (81.3333, 430.6667) and (96, 430.6667) .. (112, 428);
        
        \draw[BurntOrange,line width=5pt](112, 428) .. controls (128, 448) and (139.238, 458.584) .. (143.117, 454.326);
        \draw[white,line width=3.5pt](112, 428) .. controls (128, 448) and (139.238, 458.584) .. (143.117, 454.326);
        
        \draw[BurntOrange,line width=5pt](146.735, 366.356) .. controls (150.179, 371.699) and (128, 404) .. (112, 428);
        \draw[white,line width=3.5pt](146.735, 366.356) .. controls (150.179, 371.699) and (128, 404) .. (112, 428);
        
        \draw[BurntOrange,line width=5pt](31.238, 363.705) .. controls (27.422, 369.086) and (52, 408) .. (68, 428);
        \draw[white,line width=3.5pt](31.238, 363.705) .. controls (27.422, 369.086) and (52, 408) .. (68, 428);

        \draw[BurntOrange,line width=5pt](68, 404) .. controls (40, 428) and (27.786, 447.476) .. (32.318, 453.645);
        \draw[white,line width=3.5pt](68, 404) .. controls (40, 428) and (27.786, 447.476) .. (32.318, 453.645);

        \draw[BurntOrange,line width=5pt](112, 404) .. controls (136, 428) and (148, 448) .. (143.117, 454.326);
        \draw[white,line width=3.5pt](112, 404) .. controls (136, 428) and (148, 448) .. (143.117, 454.326);
  
        \filldraw[draw=BurntOrange, fill=white](68, 428) circle[radius=5];
        \filldraw[draw=BurntOrange, fill=white](112, 428) circle[radius=5];
        \filldraw[draw=BurntOrange, fill=white](112, 404) circle[radius=5];
        \filldraw[draw=BurntOrange, fill=white](68, 404) circle[radius=5];

        \node[BurntOrange] at (88.615, 402.001) [below] {$\ell_1$};
        \node[BurntOrange] at (32.318, 453.645) [left]{$\ell_2$};
        \node[BurntOrange] at (88.757, 430) [above]{$\ell_3$};
        \node[BurntOrange] at (143.117, 454.326) [right]{$\ell_4$};
        \node[BurntOrange] at (31.238, 363.705) [left]{$\ell_5$};
        \node[BurntOrange] at (146.735, 366.356) [right]{$\ell_6$};

        \draw[<->] (152, 416) -- (188, 416);
        
        \draw[BrickRed, opacity=.2] (350.5277, 393.5139)  .. controls (351.5842, 400.273) and (350.3863, 404.9605) .. (350.2408, 409.8785)  .. controls (350.0953, 414.7965) and (351.0022, 419.9451) .. (350.1627, 424.1617);
        \draw[BrickRed, opacity=.2] (347.6978, 383.4845)  .. controls (350.0052, 391.1695) and (341.0026, 397.5848) .. (335.8346, 404.1257)  .. controls (330.6667, 410.6667) and (329.3333, 417.3333) .. (329.5349, 423.4731)  .. controls (329.7365, 429.6129) and (331.473, 435.2258) .. (334.2492, 439.4479)  .. controls (337.0254, 443.67) and (340.8412, 446.5012) .. (342.825, 443.1713);
        \draw[BrickRed, opacity=.2] (264.7396, 478.3642)  .. controls (268.586, 479.0895) and (274.293, 473.5448) .. (280.7457, 473.4089)  .. controls (287.1984, 473.273) and (294.3969, 478.5459) .. (300.2874, 477.0827);
        \draw[BrickRed, opacity=.2] (321.2489, 467.0129)  .. controls (311.3814, 472.8013) and (307.6907, 468.4006) .. (304.512, 462.867)  .. controls (301.3333, 457.3333) and (298.6667, 450.6667) .. (292.6667, 447.3333)  .. controls (286.6667, 444) and (277.3333, 444) .. (271.3333, 446.6667)  .. controls (265.3333, 449.3333) and (262.6667, 454.6667) .. (259.2044, 460.8615)  .. controls (255.7422, 467.0563) and (251.4844, 474.1125) .. (242.9554, 469.7389);
        \draw[BrickRed, opacity=.2] (214.1076, 378.9794)  .. controls (211.7803, 384.9766) and (220.2582, 394.6942) .. (226.2543, 402.1455)  .. controls (232.2504, 409.5967) and (235.7646, 414.7815) .. (235.6419, 421.3795)  .. controls (235.5193, 427.9775) and (231.7596, 435.9888) .. (227.4603, 440.2808)  .. controls (223.1609, 444.5729) and (218.3219, 445.1458) .. (214.8346, 438.6181);
        \draw[BrickRed, opacity=.2] (208.9047, 396.6222)  .. controls (208.3798, 400.6144) and (212.1899, 406.3072) .. (212.1323, 409.312)  .. controls (212.0746, 412.3169) and (208.1493, 412.6337) .. (208.3315, 414.9007);
        \draw[BrickRed, opacity=.2] (234.2927, 352.3687)  .. controls (231.646, 354.6531) and (243.823, 369.3265) .. (251.2448, 380.6633)  .. controls (258.6667, 392) and (261.3333, 400) .. (268, 404)  .. controls (274.6667, 408) and (285.3333, 408) .. (292.6667, 404)  .. controls (300, 400) and (304, 392) .. (312.8266, 382.4751)  .. controls (321.6533, 372.9501) and (335.3065, 361.9002) .. (329.7967, 355.9972);
        \draw[BrickRed, opacity=.2] (259.9735, 338.8412)  .. controls (265.3717, 337.5017) and (268.6859, 350.7508) .. (273.6763, 357.3754)  .. controls (278.6667, 364) and (285.3333, 364) .. (290.1352, 357.5635)  .. controls (294.937, 351.1269) and (297.874, 338.2539) .. (306.1367, 340.9115);
        \node[BurntOrange, opacity=.2] at (304, 428) {\tiny$\bullet$};
        \node[BurntOrange, opacity=.2] at (260, 428) {\tiny$\bullet$};
        \draw[BurntOrange, opacity=.2] (223.2376, 363.7052)  .. controls (219.4222, 369.0856) and (244, 408) .. (260, 428);
        \draw[BurntOrange, opacity=.2] (338.7348, 366.356)  .. controls (342.1793, 371.6993) and (320, 404) .. (304, 428);
        \draw[BurntOrange, opacity=.2] (304, 428)  .. controls (320, 448) and (331.2376, 458.5837) .. (335.1174, 454.3257);
        \draw[BurntOrange, opacity=.2] (224.3178, 453.6453)  .. controls (229.5193, 459.339) and (248, 440) .. (260, 428);
        \draw[BurntOrange, opacity=.2] (260, 428)  .. controls (273.3333, 430.6667) and (288, 430.6667) .. (304, 428);
        \draw[BurntOrange] (260, 404)  .. controls (273.3333, 401.3333) and (288, 401.3333) .. (304, 404);
        \draw[BrickRed] (271.3239, 478.9537)  .. controls (268.6421, 478.2871) and (266.3211, 477.1435) .. (266.4939, 475.2384)  .. controls (266.6667, 473.3333) and (269.3333, 470.6667) .. (273.3333, 469.3333)  .. controls (277.3333, 468) and (282.6667, 468) .. (286.6667, 469.3333)  .. controls (290.6667, 470.6667) and (293.3333, 473.3333) .. (293.5352, 475.2347)  .. controls (293.737, 477.1361) and (291.4741, 478.2722) .. (288.8225, 478.9389)  .. controls (286.1709, 479.6056) and (283.1307, 479.8028) .. (280.0894, 479.8065)  .. controls (277.0481, 479.8102) and (274.0056, 479.6204) .. cycle;
        \draw[BrickRed] (264.7396, 478.3642)  .. controls (255.6692, 475.7644) and (251.8346, 467.8822) .. (255.2506, 461.2744)  .. controls (258.6667, 454.6667) and (269.3333, 449.3333) .. (280.6667, 448.6667)  .. controls (292, 448) and (304, 452) .. (309.0241, 458.1126)  .. controls (314.0481, 464.2251) and (312.0962, 472.4502) .. (300.2874, 477.0827);
        \draw[BrickRed] (242.9554, 469.7389)  .. controls (232.0274, 461.6902) and (234.0137, 452.8451) .. (239.6735, 444.4225)  .. controls (245.3333, 436) and (254.6667, 428) .. (265.939, 424.6695)  .. controls (277.2113, 421.3389) and (290.4225, 422.6779) .. (300.9876, 428.1559)  .. controls (311.5527, 433.634) and (319.4716, 443.2511) .. (323.5328, 450.3042)  .. controls (327.5939, 457.3572) and (327.7973, 461.8462) .. (321.2489, 467.0129);
        \draw[BrickRed] (214.1076, 378.9794)  .. controls (219.4553, 369.0342) and (227.7276, 374.5171) .. (233.8638, 380.5919)  .. controls (240, 386.6667) and (244, 393.3333) .. (243.3333, 400.6667)  .. controls (242.6667, 408) and (237.3333, 416) .. (230.0767, 422.9839)  .. controls (222.8202, 429.9678) and (213.6403, 435.9355) .. (214.8346, 438.6181);
        \draw[BrickRed] (208.9047, 396.6222)  .. controls (210.9878, 387.474) and (217.4939, 387.737) .. (222.0803, 391.2018)  .. controls (226.6667, 394.6667) and (229.3333, 401.3333) .. (228.6667, 407.3333)  .. controls (228, 413.3333) and (224, 418.6667) .. (219.8651, 422.8464)  .. controls (215.7301, 427.0262) and (211.4603, 430.0523) .. (208.3315, 414.9007);
        \draw[BrickRed] (208.0391, 405.6276)  .. controls (208.3553, 400.8561) and (210.1776, 400.4281) .. (211.7555, 401.5474)  .. controls (213.3333, 402.6667) and (214.6667, 405.3333) .. (214.6667, 408)  .. controls (214.6667, 410.6667) and (213.3333, 413.3333) .. (211.7366, 414.2514)  .. controls (210.1398, 415.1695) and (208.2796, 414.3391) .. (208.0259, 409.9324);
        \draw[BrickRed] (234.2927, 352.3687)  .. controls (238.5939, 349.0972) and (251.297, 370.5486) .. (263.6485, 381.2743)  .. controls (276, 392) and (288, 392) .. (300.7528, 382.3665)  .. controls (313.5057, 372.7331) and (327.0113, 353.4662) .. (329.7967, 355.9972);
        \draw[BrickRed] (259.9735, 338.8412)  .. controls (249.7548, 342.6607) and (252.8774, 351.3303) .. (257.772, 358.3318)  .. controls (262.6667, 365.3333) and (269.3333, 370.6667) .. (275.3333, 373.3333)  .. controls (281.3333, 376) and (286.6667, 376) .. (292, 374)  .. controls (297.3333, 372) and (302.6667, 368) .. (307.4955, 361.5063)  .. controls (312.3244, 355.0126) and (316.6488, 346.0253) .. (306.1367, 340.9115);
        \draw[BrickRed] (286.8383, 336.9936)  .. controls (290.5347, 337.616) and (293.2674, 338.808) .. (293.967, 340.0707)  .. controls (294.6667, 341.3333) and (293.3333, 342.6667) .. (290.6667, 344)  .. controls (288, 345.3333) and (284, 346.6667) .. (279.9985, 347.0623)  .. controls (275.9969, 347.458) and (271.9938, 346.916) .. (268.6605, 345.5826)  .. controls (265.3272, 344.2493) and (262.6636, 342.1247) .. (263.6354, 340.4399)  .. controls (264.6071, 338.7552) and (269.2143, 337.5104) .. (273.8479, 336.9144)  .. controls (278.4816, 336.3184) and (283.1418, 336.3712) .. cycle;
        \draw[BrickRed] (347.6978, 383.4845)  .. controls (342.0496, 371.4781) and (335.0807, 381.28) .. (329.2315, 388.4342)  .. controls (323.3823, 395.5883) and (318.653, 400.0946) .. (319.222, 405.1699)  .. controls (319.7911, 410.2451) and (325.6587, 415.8891) .. (332.0848, 422.8043)  .. controls (338.5109, 429.7194) and (345.4955, 437.9055) .. (342.825, 443.1713);
        \draw[BrickRed] (350.5277, 393.5139)  .. controls (347.794, 383.7518) and (342.5304, 388.9548) .. (339.1735, 394.0056)  .. controls (335.8166, 399.0564) and (334.3664, 403.955) .. (334.8219, 408.9287)  .. controls (335.2774, 413.9024) and (337.6387, 418.9512) .. (340.7812, 423.5002)  .. controls (343.9237, 428.0492) and (347.8474, 432.0984) .. (350.1627, 424.1617);
        \draw[BrickRed] (345.823, 407.881)  .. controls (345.3679, 410.2676) and (345.9077, 412.9517) .. (346.9695, 414.8059)  .. controls (348.0313, 416.6601) and (349.6152, 417.6844) .. (350.5385, 416.1819)  .. controls (351.4617, 414.6794) and (351.7242, 410.6499) .. (351.3873, 407.7511)  .. controls (351.0504, 404.8522) and (350.114, 403.0839) .. (348.9209, 403.2444)  .. controls (347.7278, 403.405) and (346.278, 405.4943) .. cycle;
        \node at (349.1885, 409.6479) {\tiny$\bullet$};
        \node at (209.4486, 408.5412) {\tiny$\bullet$};
        \node at (279.9554, 476.7277) {\tiny$\bullet$};
        \node at (279.4477, 341.7047) {\tiny$\bullet$};
        \draw[BurntOrange] (260, 404)  .. controls (244, 376) and (226.4307, 359.8925) .. (223.2376, 363.7052);
        \draw[BurntOrange] (260, 404)  .. controls (232, 428) and (219.7864, 447.4756) .. (224.3178, 453.6453);
        \draw[BurntOrange] (304, 404)  .. controls (328, 428) and (340, 448) .. (335.1174, 454.3257);
        \draw[BurntOrange] (304, 404)  .. controls (320, 380) and (335.937, 362.6673) .. (338.7348, 366.356);
        \draw (280, 408) circle[radius=72];
        \node[BurntOrange] at (304, 404) {\tiny$\bullet$};
        \node[BurntOrange] at (260, 404) {\tiny$\bullet$};

        \draw[opacity=.2](525.9041, 457.9217) .. controls (523.5559, 460.1775) and (516.7527, 454.5703) .. (507.0667, 443.9747) .. controls (497.3807, 445.3874) and (488.5019, 445.3874) .. (480.4303, 443.9747) .. controls (473.1658, 450.3321) and (461.9779, 460.5775) .. (458.8293, 457.561);
        
        \draw[BrickRed, opacity=.2](524, 476) .. controls (524, 480) and (508, 484) .. (492, 484) .. controls (476, 484) and (460, 480) .. (460, 476) -- (460, 476);
        \draw[BrickRed](460, 476) .. controls (460, 472) and (476, 468) .. (492, 468) .. controls (508, 468) and (524, 472) .. (524, 476);
        \draw[opacity=.2](460, 328) .. controls (460, 332) and (476, 336) .. (493, 336) .. controls (510, 336) and (528, 332) .. (528, 328);
        \draw[BrickRed, opacity=.2](460, 340) .. controls (460, 344) and (476, 348) .. (493, 348) .. controls (510, 348) and (528, 344) .. (528, 340);
        \draw[opacity=.2](552.0944, 383.3168) .. controls (553.7869, 385.6147) and (545.2575, 397.324) .. (536.8052, 408.4523);
        \draw[opacity=.2](580, 384) .. controls (576, 384) and (572, 396) .. (572, 407) .. controls (572, 418) and (576, 428) .. (580, 428);
        \draw[BrickRed, opacity=.2](568, 384) .. controls (564, 384) and (560, 396) .. (560, 407) .. controls (560, 418) and (564, 428) .. (568, 428);
        \draw[BrickRed](568, 428) .. controls (572, 428) and (576, 418) .. (576, 407) .. controls (576, 396) and (572, 384) .. (568, 384);
        \draw[BrickRed, opacity=.2](420, 428) .. controls (416, 428) and (412, 418) .. (412, 407) .. controls (412, 396) and (416, 384) .. (420, 384) -- (420, 384);
        \draw[BrickRed](420, 384) .. controls (424, 384) and (428, 396) .. (428, 407) .. controls (428, 418) and (424, 428) .. (420, 428);
        \draw[BrickRed](528, 340) .. controls (528, 336) and (512, 332) .. (495, 332) .. controls (478, 332) and (460, 336) .. (460, 340);
        \node[BurntOrange] at (507.0664, 415.9748) {\tiny$\bullet$};
        \node[BurntOrange] at (480.4299, 415.9748) {\tiny$\bullet$};
        \draw[shift={(458.175, 381.913)}, xscale=0.6054, yscale=0.5298, BurntOrange](0, 0) .. controls (-3.816, 5.381) and (20.762, 44.295) .. (36.762, 64.295);
        \draw[shift={(528.094, 383.317)}, xscale=0.6054, yscale=0.5298, BurntOrange](0, 0) .. controls (3.444, 5.343) and (-18.735, 37.644) .. (-34.735, 61.644);
        \draw[shift={(507.066, 415.975)}, xscale=0.6054, yscale=0.5298, BurntOrange](0, 0) .. controls (16, 20) and (27.238, 30.584) .. (31.117, 26.326);
        \draw[shift={(458.829, 429.561)}, xscale=0.6054, yscale=0.5298, BurntOrange](0, 0) .. controls (5.201, 5.694) and (23.682, -13.645) .. (35.682, -25.645);
        \draw[shift={(480.43, 415.975)}, xscale=0.6054, yscale=0.5298, BurntOrange](0, 0) .. controls (13.3333, 2.6667) and (28, 2.6667) .. (44, 0);
        \draw[shift={(480.43, 403.26)}, xscale=0.6054, yscale=0.5298, BurntOrange](0, 0) .. controls (13.3333, -2.6667) and (28, -2.6667) .. (44, 0);
        \draw[shift={(480.43, 403.26)}, xscale=0.6054, yscale=0.5298, BurntOrange](0, 0) .. controls (-16, -28) and (-33.569, -44.107) .. (-36.762, -40.295);
        \draw[shift={(480.43, 403.26)}, xscale=0.6054, yscale=0.5298, BurntOrange](0, 0) .. controls (-28, 24) and (-40.214, 43.476) .. (-35.682, 49.645);
        \draw[shift={(507.066, 403.26)}, xscale=0.6054, yscale=0.5298, BurntOrange](0, 0) .. controls (24, 24) and (36, 44) .. (31.117, 50.326);
        \draw[shift={(507.066, 403.26)}, xscale=0.6054, yscale=0.5298, BurntOrange](0, 0) .. controls (16, -24) and (31.937, -41.333) .. (34.735, -37.644);
        \node[BurntOrange] at (507.0664, 403.2601) {\tiny$\bullet$};
        \node[BurntOrange] at (480.4299, 403.2601) {\tiny$\bullet$};
        \draw(458.3038, 456.3336) .. controls (459.4346, 458.7779) and (460, 469.3333) .. (460, 488);
        \draw(526.4547, 456.8502) .. controls (525.0007, 458.9592) and (524.1825, 469.3424) .. (524, 488);
        \draw(525.9041, 457.9217) .. controls (528.8602, 454.5703) and (521.5957, 443.9747) .. (507.0667, 431.2599) .. controls (497.3807, 429.8472) and (488.5019, 429.8472) .. (480.4303, 431.2599) .. controls (463.4798, 443.9747) and (456.0857, 454.2927) .. (458.8293, 457.561);
        \draw(460, 488) .. controls (460, 484) and (476, 480) .. (492, 480) .. controls (508, 480) and (524, 484) .. (524, 488) .. controls (524, 492) and (508, 496) .. (492, 496) .. controls (476, 496) and (460, 492) .. (460, 488) -- cycle;
        \draw(434.1755, 381.9124) .. controls (432.7252, 383.3041) and (424, 384) .. (408, 384);
        \draw(434.8293, 429.561) .. controls (432.9431, 428.5203) and (424, 428) .. (408, 428);
        \draw[opacity=.2](434.1755, 381.9124) .. controls (432.3428, 384.1739) and (441.3283, 397.6163) .. (449.9569, 408.3664);
        \draw(449.9569, 408.3664) .. controls (452.2047, 411.167) and (454.4284, 413.7848) .. (456.4303, 415.9747) .. controls (449.1658, 422.3321) and (437.9779, 432.5775) .. (434.8293, 429.561) .. controls (432.0857, 426.2927) and (439.4798, 415.9747) .. (456.4303, 403.2599) .. controls (446.7443, 388.426) and (436.1085, 379.8928) .. (434.1755, 381.9124);
        \draw(408, 384) .. controls (412, 384) and (416, 396) .. (416, 407) .. controls (416, 418) and (412, 428) .. (408, 428) .. controls (404, 428) and (400, 418) .. (400, 407) .. controls (400, 396) and (404, 384) .. (408, 384) -- cycle;
        \draw(528.3406, 356.1604) .. controls (528.1135, 353.3868) and (528, 344) .. (528, 328);
        \draw(458.1755, 353.9124) .. controls (459.3918, 352.6375) and (460, 344) .. (460, 328);
        \draw(528, 328) .. controls (528, 324) and (512, 320) .. (495, 320) .. controls (478, 320) and (460, 324) .. (460, 328);
        \draw(458.1755, 353.9124) .. controls (455.8654, 356.7631) and (470.7443, 377.3791) .. (480.4303, 387.9747) .. controls (488.5019, 389.3874) and (497.3807, 389.3874) .. (507.0667, 387.9747) .. controls (516.7527, 375.2599) and (530.1793, 358.1474) .. (528.0944, 355.3168) .. controls (526.4005, 353.3624) and (516.7527, 362.5452) .. (507.0667, 375.2599) .. controls (497.3807, 373.8472) and (488.5019, 373.8472) .. (480.4303, 375.2599) .. controls (470.7443, 360.426) and (460.1085, 351.8928) .. (458.1755, 353.9124) -- cycle;
        \draw(549.2437, 430.3176) .. controls (551.0812, 428.7725) and (561.3333, 428) .. (580, 428);
        \draw(551.751, 383.0962) .. controls (554.5837, 383.6987) and (564, 384) .. (580, 384);
        \draw(580, 428) .. controls (584, 428) and (588, 418) .. (588, 407) .. controls (588, 396) and (584, 384) .. (580, 384);
        \draw(536.8052, 408.4523) .. controls (534.8454, 411.0325) and (532.8898, 413.5815) .. (531.0667, 415.9747) .. controls (540.7527, 426.5703) and (547.5559, 432.1775) .. (549.9041, 429.9217) .. controls (552.8602, 426.5703) and (545.5957, 415.9747) .. (531.0667, 403.2599) .. controls (540.7527, 390.5452) and (550.4005, 381.3624) .. (552.0944, 383.3168);
        \node at (492, 510) {\scriptsize$L_1 = \ell_1+\ell_2+\ell_3+\ell_4$};
    \end{tikzpicture}
    \caption{\textbf{Riemann Surfaces as Metrized Ribbon Graphs:} The Strebel differential foliates any Riemann surface by a unique set of curves, called horizontal trajectories (in red). A measure zero subset, the critical trajectories, assign a unique Strebel graph to the surface (left panel). The moduli are encoded as edge lengths, $\ell_{e}$, providing the basis for the combinatorial description of the moduli space, $\M_{g,n}^{\textup{comb}}$. The sum of these lengths around a face of the Strebel graph must equal the length of the boundary. Geometrically, this decomposes the surface as a collection of semi-infinite flat cylinders, glued to the Strebel graph.}
    \label{fig:StrebelCurves}
\end{figure}

The hyperbolic and combinatorial models, unlike the complex-geometric one, both carry a natural symplectic form, which in turn defines a natural notion of volume: the Weil--Petersson volume and the Kontsevich volume:
\begin{equation}
\begin{aligned}
    V^{\textup{WP}}_{g,n}(L_1,\ldots,L_n)
    &\coloneqq
    \mathrm{Vol}\left(
        \M_{g,n}^{\textup{hyp}}(L_1,\ldots,L_n)
    \right), \\
    V^{\textup{Kon}}_{g,n}(L_1,\ldots,L_n)
    &\coloneqq
    \mathrm{Vol}\left(
        \M_{g,n}^{\textup{comb}}(L_1,\ldots,L_n)
    \right).
\end{aligned}
\end{equation}
We will not review here the specific symplectic forms that must be integrated to define these volumes. However, as in the toy example, it is worth noting that both forms arise naturally from the intrinsic geometry of their respective moduli spaces: hyperbolic in the first case and combinatorial in the second.

Another important point, again following the previous analogy, is that the combinatorial model admits a notion of lattice points. If the boundary lengths are integers, say $(b_1,\ldots,b_n) \in \Z_+^n$, one can count integer Strebel graphs, i.e. Strebel graphs with integral edge-lengths:
\begin{equation}
    N^{\textup{Nor}}_{g,n}(b_1,\ldots,b_n)
    \coloneqq
    \#\,\M_{g,n}^{\textup{comb},\Z}(b_1,\ldots,b_n).
\end{equation}
Here $\M_{g,n}^{\textup{comb},\Z}(b_1,\ldots,b_n)$ denotes the discrete space of Strebel graphs with integral edge-lengths and fixed boundary lengths. As before, the number of lattice points encodes the continuous volume as its leading coefficient (cf. \cref{fig:BMN}):
\begin{equation}
    \lim_{t \to 0^+} \
        2^{2g-2+n} \,
        t^{2(3g-3+n)} \,
        N^{\textup{Nor}}_{g,n}(L_1/t,\ldots,L_n/t)
    =
    2 \cdot V^{\textup{Kon}}_{g,n}(L_1,\ldots,L_n).
\end{equation}
From the purely geometric point of view, the factor of $2$ on the right-hand side is due to the fact that the set of lattice points $\M_{g,n}^{\textup{comb},\Z}(b_1,\ldots,b_n)$ is empty whenever $b_1 + \cdots + b_n$ is odd, since the sum of all edge-lengths is twice the sum of the boundary lengths.

\begin{center}
    \renewcommand{\arraystretch}{1.2}
    \begin{tabular}{>{\centering\arraybackslash}p{.3\textwidth}|>{\centering\arraybackslash}p{.3\textwidth}|>{\centering\arraybackslash}p{.3\textwidth}}
        \toprule
        smooth model & \multicolumn{2}{c}{combinatorial model} \\
        \midrule
        hyperbolic metrics &
        Strebel graphs &
        integer Strebel graphs \\
        $\M_{g,n}^{\textup{hyp}}(L_1,\ldots,L_n)$ &
        $\M_{g,n}^{\textup{comb}}(L_1,\ldots,L_n)$ &
        $\M_{g,n}^{\textup{comb},\Z}(b_1,\ldots,b_n)$ \\
        \midrule
        Weil--Petersson volumes &
        Kontsevich volumes &
        Norbury discrete volumes \\
        $V^{\textup{WP}}_{g,n}(L_1,\ldots,L_n)$ &
        $V^{\textup{Kon}}_{g,n}(L_1,\ldots,L_n)$ &
        $N^{\textup{Nor}}_{g,n}(b_1,\ldots,b_n)$ \\
        \bottomrule
    \end{tabular}
\end{center}

As both models are isomorphic to the moduli of curves, one can also try to express the Weil--Petersson and Kontsevich symplectic forms in terms of complex-algebraic objects on $\Mbar_{g,n}$, the Deligne--Mumford compactification of the moduli space of curves. Under the respective identifications, one finds \cite{Wol85,Mir07b,Kon92}
\begin{equation}\label{eq:vol:kappa:psi}
\begin{aligned}
    V^{\textup{WP}}_{g,n}(L_1,\ldots,L_n)
    &=
    \int_{\Mbar_{g,n}}
        \exp\left(
            2 \pi^2 \kappa_1
            +
            \frac{1}{2} \sum_{i=1}^n L_i^2 \psi_i
        \right),
    \\
    V^{\textup{Kon}}_{g,n}(L_1,\ldots,L_n)
    &=
    \int_{\Mbar_{g,n}}
        \exp\left(
            \frac{1}{2} \sum_{i=1}^n L_i^2 \psi_i
        \right).
\end{aligned}
\end{equation}
Here $\kappa_1$ and $\psi_i$ are natural cohomology classes on the moduli space of curves, whose definition is omitted. It is worth stressing that these formulas are highly non-trivial: the left-hand sides are defined through the intrinsic geometry of the hyperbolic and combinatorial models, while the right-hand sides are purely complex-geometric and, a priori, have nothing to do with the notion of volume. The equalities rely on the uniformization and Strebel theorems, which bridge the hyperbolic and combinatorial worlds to the complex-algebraic one.

\subsection{The recursions}
A natural question is therefore: how can one compute these three notions of volumes? In all three cases, the answer is provided by a \emph{topological recursion formula}, that is, a recursion on the Euler characteristic $2g - 2 + n$. The structure of the recursions for these volumes is entirely parallel: the continuous volumes satisfy integral recursions, while the discrete ones satisfy a discrete recursion; the Weil--Petersson volumes involve recursion kernels built out of hyperbolic functions, whereas the Kontsevich and Norbury volumes involve kernels built out of piecewise linear functions.

We start with the Weil--Petersson volumes \cite{Mir07a}.

\begin{theorem}[Mirzakhani]\label{thm:Mir}
    For $2g-2+n > 1$, the Weil--Petersson volumes satisfy the recursion relation
	\begin{equation}
    \begin{split}
    	V^{\textup{WP}}_{g,n}(L_1,\ldots,L_n)
    	&=
    	\int_{0}^{+\infty}
    		d\ell \,
    		  \ell \, B^{\textup{hyp}}(L_1,L_m,\ell)
    		  V^{\textup{WP}}_{g,n-1}(\ell,L_2,\dots, \widehat{L_m},\dots, L_n) \\
    	& \quad
    	+ \frac{1}{2} \int_{0}^{+\infty} \int_{0}^{+\infty}
    		d\ell \, d\ell' \,
    		\ell \ell' \, C^{\textup{hyp}}(L_1,\ell,\ell') \Bigg(
    			V^{\textup{WP}}_{g-1,n+1}(\ell,\ell',L_2,\dots,L_n) \\
    	& \qquad\qquad\qquad\qquad\qquad
    			+
    			\sum_{\substack{ h + h' = g \\ J \sqcup J' = \{2,\dots,n\} }}^{\textup{stable}}
    				V^{\textup{WP}}_{h,1+|J|}(\ell,L_{J})
    				V^{\textup{WP}}_{h',1+|J'|}(\ell',L_{J'})
    	\Bigg).
    \end{split}
    \end{equation}
    The label ``stable'' means that both $2h-2+(1+|J|)>0$ and $2h'-2+(1+|J'|)>0$. The $B$ and $C$ kernels are defined in terms of $H^{\textup{hyp}}(\ell) \coloneqq 2 \log(1 + e^{\ell/2})$:
    \begin{equation}
    \begin{aligned}
    	B^{\textup{hyp}}(L,L',\ell)
    	& \coloneqq
    	\frac{1}{2L} \Bigl(
    		H^{\textup{hyp}}(L + L' - \ell)
    		-
    		H^{\textup{hyp}}(-L - L' - \ell) \\
    		& \qquad\qquad\qquad
    		+
    		H^{\textup{hyp}}(L - L' - \ell)
    		-
    		H^{\textup{hyp}}(-L + L' - \ell)
    	\Bigr), \\
    	C^{\textup{hyp}}(L,\ell,\ell')
    	& \coloneqq
    	\frac{1}{L} \Bigl(
    		H^{\textup{hyp}}(L - \ell - \ell')
    		-
    		H^{\textup{hyp}}(-L - \ell - \ell')
    	\Bigr).
    \end{aligned}
    \end{equation}
    Together with the initial data $V^{\textup{WP}}_{0,3}(L_1,L_2,L_3) = 1$ and $V^{\textup{WP}}_{1,1}(L_1) = \frac{L_1^2}{48} + \frac{\pi^2}{12}$, this recursion uniquely determines the volumes.
\end{theorem}

The Kontsevich volumes satisfy exactly the same recursion, but with different kernels and initial data. To the best of our knowledge, the recursion relation for the Kontsevich volumes first appeared in this form in \cite{BCSW12}; its proof paralleling Mirzakhani's argument was later given in \cite{ABCGLW}. It is equivalent to the Virasoro constraints of Dijkgraaf--Verlinde--Verlinde \cite{DVV91}, which in turn follow from the fact that the associated partition function is a solution of the KdV hierarchy, as conjectured by Witten and proved by Kontsevich \cite{Wit91,Kon92}.

\begin{theorem}[Kontsevich et al.]\label{thm:Kon}
    For $2g-2+n > 1$, the Kontsevich volumes satisfy the recursion relation
	\begin{equation}
    \begin{split}
    	V^{\textup{Kon}}_{g,n}(L_1,\ldots,L_n)
    	&=
    	\int_{0}^{+\infty}
    		d\ell \,
    		  \ell \, B^{\textup{comb}}(L_1,L_m,\ell)
    		  V^{\textup{Kon}}_{g,n-1}(\ell,L_2,\dots, \widehat{L_m},\dots, L_n) \\
    	& \quad
    	+ \frac{1}{2} \int_{0}^{+\infty} \int_{0}^{+\infty}
    		d\ell \, d\ell' \,
    		\ell \ell' \, C^{\textup{comb}}(L_1,\ell,\ell') \Bigg(
    			V^{\textup{Kon}}_{g-1,n+1}(\ell,\ell',L_2,\dots,L_n) \\
    	& \qquad\qquad\qquad\qquad\qquad
    			+
    			\sum_{\substack{ h + h' = g \\ J \sqcup J' = \{2,\dots,n\} }}^{\textup{stable}}
    				V^{\textup{Kon}}_{h,1+|J|}(\ell,L_{J})
    				V^{\textup{Kon}}_{h',1+|J'|}(\ell',L_{J'})
    	\Bigg),
    \end{split}
    \end{equation}
	where the $B$ and $C$ kernels are defined in terms of $H^{\textup{comb}}(\ell) \coloneqq \ell \theta(\ell)$, the ramp function:
    \begin{equation}\label{eq:comb:kernels}
    \begin{aligned}
    	B^{\textup{comb}}(L,L',\ell)
    	& \coloneqq
    	\frac{1}{2L} \Bigl(
    		H^{\textup{comb}}(L + L' - \ell)
    		-
    		H^{\textup{comb}}(-L - L' - \ell) \\
    		& \qquad\qquad\qquad
    		+
    		H^{\textup{comb}}(L - L' - \ell)
    		-
    		H^{\textup{comb}}(-L + L' - \ell)
    	\Bigr), \\
    	C^{\textup{comb}}(L,\ell,\ell')
    	& \coloneqq
    	\frac{1}{L} \Bigl(
    		H^{\textup{comb}}(L - \ell - \ell')
    		-
    		H^{\textup{comb}}(-L - \ell - \ell')
    	\Bigr).
    \end{aligned}
    \end{equation}
    Together with the initial data $V^{\textup{Kon}}_{0,3}(L_1,L_2,L_3) = 1$ and $V^{\textup{Kon}}_{1,1}(L_1) = \frac{L_1^2}{48}$, this recursion uniquely determines the volumes.
\end{theorem}

Finally, Norbury's discrete volumes satisfy an identical recursion, with the only difference that integrals are replaced by sums \cite{Nor10}.

\begin{theorem}[Norbury]\label{thm:Nor}
    For $2g-2+n > 1$ and $b_1 + \cdots + b_n$ even, the discrete Norbury volumes satisfy the recursion relation
	\begin{equation}
    \begin{split}
    	N^{\textup{Nor}}_{g,n}(b_1,\ldots,b_n)
    	&=
    	\sum_{\beta > 0}
    		  \beta \, B^{\textup{comb}}(b_1,b_m,\beta)
    		  N^{\textup{Nor}}_{g,n-1}(\beta,b_2,\dots, \widehat{b_m},\dots, b_n) \\
    	& \quad
    	+ \frac{1}{2} \sum_{\beta,\beta' > 0}
    		\beta \beta' \, C^{\textup{comb}}(b_1,\beta,\beta') \Bigg(
    			N^{\textup{Nor}}_{g-1,n+1}(\beta,\beta',b_2,\dots,b_n) \\
    	& \qquad\qquad\qquad\qquad\qquad
    			+
    			\sum_{\substack{ h + h' = g \\ J \sqcup J' = \{2,\dots,n\} }}^{\textup{stable}}
    				N^{\textup{Nor}}_{h,1+|J|}(\beta,b_{J})
    				N^{\textup{Nor}}_{h',1+|J'|}(\beta',b_{J'})
    	\Bigg),
    \end{split}
    \end{equation}
	where the $B$ and $C$ kernels are as in \cref{eq:comb:kernels}.
    Together with the initial data $N^{\textup{Nor}}_{0,3}(b_1,b_2,b_3) = \frac{1 + (-1)^{b_1+b_2+b_3}}{2}$ and $N^{\textup{Nor}}_{1,1}(b_1) = \frac{1 + (-1)^{b_1}}{2}\frac{b_1^2 - 4}{48}$, this recursion uniquely determines the discrete volumes.
\end{theorem}

\subsection{The geometric origin}
\label{sec:geomrorigins}
We conclude with a review of the geometric origin of the recursion relations above, following \cite{Mir07a} and \cite{ABCGLW}. In both hyperbolic and combinatorial context, the recursions arise from a recursive computation of the constant function $1$ on the respective models of the moduli space, which is then integrated against the Weil--Petersson volume form, the Kontsevich volume form, or the Dirac delta measure supported on the lattice points, respectively. Because the recursion is independent of the chosen measure, this also explains why the kernels for the Kontsevich and Norbury volumes coincide: the only difference lies in the measure, which merely converts integrals into sums.

A key ingredient in the integration process is the compatibility of all three measures with respect to cutting and gluing operations. In the continuous recursions, the integral $\int_0^{+\infty} d\ell \, \ell (\cdot) = \int_{0}^{+\infty} d\ell \int_{0}^{\ell} d\tau (\cdot)$ is interpreted as an integration over all possible hyperbolic or combinatorial Fenchel--Nielsen length and twist coordinates of the internal curve. An analogous geometric interpretation holds in the discrete setting, where only integer lengths and twists are allowed.
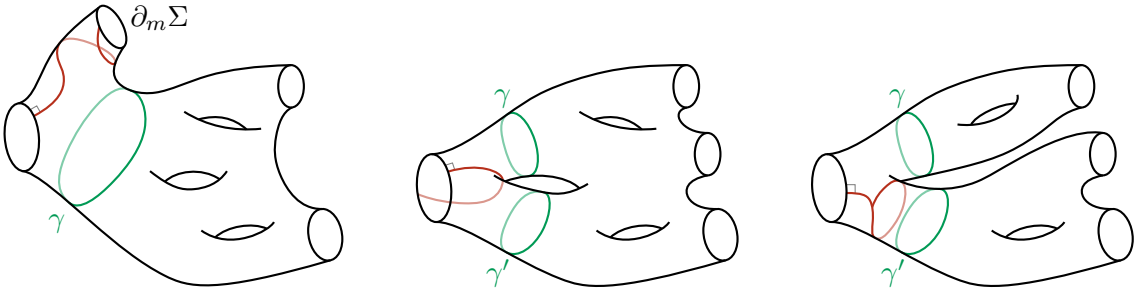
\begin{figure}[b]
  \centering
  {
  \tikzset{every picture/.style=thick}
  \begin{tikzpicture}[x=1pt,y=1pt,scale=.4]
    \draw[thin,opacity=.5] (183.7593, 729.9248) -- (189.9928, 733.5251) -- (193.3305, 727.5382);
    \draw[BrickRed, opacity=.5](211.7511, 793.1645) .. controls (216.627, 798.7579) and (226.5504, 797.5348) .. (238.0775, 793.4391) .. controls (249.6046, 789.3434) and (262.7355, 782.375) .. (263.8057, 774.7405);
    \draw[BrickRed](187.257, 723.9633) .. controls (211.8857, 737.2462) and (215.3571, 749.7916) .. (215.348, 758.2338) .. controls (215.3389, 766.676) and (211.8493, 771.0151) .. (210.1373, 776.5968) .. controls (208.4252, 782.1785) and (208.4908, 789.0028) .. (211.7511, 793.1645);
    \draw[BrickRed](263.8057, 774.7405) .. controls (262.57, 770.2811) and (256.0757, 774.9102) .. (251.7821, 780.9872) .. controls (247.4885, 787.0642) and (245.3954, 794.589) .. (248.5761, 808.2784);
    \draw[ForestGreen](220.133, 642.413) .. controls (230.051, 633.704) and (247.0255, 646.852) .. (261.5128, 662.0927) .. controls (276, 677.3333) and (288, 694.6667) .. (290.9638, 711.7898) .. controls (293.9275, 728.913) and (287.855, 745.826) .. (277.369, 748.285);
    \draw[ForestGreen, opacity=.5](277.369, 748.285) .. controls (267.726, 753.964) and (255.863, 744.982) .. (245.2648, 733.1577) .. controls (234.6667, 721.3333) and (225.3333, 706.6667) .. (218.0524, 690.0133) .. controls (210.7715, 673.36) and (205.543, 654.72) .. (220.133, 642.413);
    \draw(176, 672) .. controls (188, 672) and (216, 646) .. (240.6667, 624.3333) .. controls (265.3333, 602.6667) and (286.6667, 585.3333) .. (309.3333, 573.6667) .. controls (332, 562) and (356, 556) .. (464.699, 587.884);
    \draw(460, 612) circle[radius=1, cm={16,-4,0,24,(0,0)}];
    \draw(176, 704) circle[radius=1, cm={16,0,-2,32,(0,0)}];
    \draw(428, 752) ellipse[x radius=12, y radius=20];
    \draw(328, 728) .. controls (349.3333, 714.6667) and (373.3333, 709.3333) .. (400, 712);
    \draw(341.458, 720.808) .. controls (359.1527, 725.6027) and (374.1273, 722.4447) .. (386.382, 711.334);
    \draw(344, 616) .. controls (372, 596) and (408, 616) .. (412, 624);
    \draw(357.035, 609.515) .. controls (368, 620) and (392, 624) .. (405.362, 617.681);
    \draw(260, 808) circle[radius=1, cm={11.3137,0,-8,20,(0,0)}];
    \draw(172.845, 735.916) .. controls (188, 740) and (192, 764) .. (205, 784) .. controls (218, 804) and (240, 820) .. (247.936, 826.087);
    \draw(271.457, 789.268) .. controls (260, 772) and (260, 762) .. (266.6667, 755) .. controls (273.3333, 748) and (286.6667, 744) .. (300.9975, 746.049) .. controls (315.3283, 748.098) and (330.6567, 756.196) .. (350.3601, 761.9667) .. controls (370.0635, 767.7375) and (394.142, 771.181) .. (428.052, 772);
    \draw(456.951, 636.322) .. controls (436, 636) and (424, 654) .. (418, 668.3333) .. controls (412, 682.6667) and (412, 693.3333) .. (414, 704.6667) .. controls (416, 716) and (420, 728) .. (428.934, 732.061);
    \draw(296, 672) .. controls (316, 656) and (326, 654) .. (337, 655) .. controls (348, 656) and (360, 660) .. (368, 672);
    \draw(309.972, 662.04) .. controls (320, 672) and (344, 676) .. (361.221, 664.33);
    \node at (305, 816.5) {$\de_m \Sigma$};
    \node[ForestGreen] at (208.708, 622.152) {$\gamma$};

    \begin{scope}[xshift=13cm]
      \draw[thin,opacity=.5](202.4667, 677.5777) -- (209.1964, 679.7214) -- (211.3373, 673.2011);
      \draw[BrickRed, opacity=.5](256.336, 663.317) .. controls (256, 656) and (252, 650) .. (246, 646.3333) .. controls (240, 642.6667) and (232, 641.3333) .. (220, 641.6667) .. controls (208, 642) and (192, 644) .. (176.61, 650.314);
      \draw[BrickRed](205.017, 671.535) .. controls (220, 676) and (232, 676) .. (241, 674) .. controls (250, 672) and (256, 668) .. (256.336, 663.317);
      \draw[ForestGreen](279.374, 667.219) .. controls (290.525, 667.722) and (291.2625, 685.861) .. (285.9773, 701.866) .. controls (280.692, 717.871) and (269.384, 731.742) .. (260.42, 725.438);
      \draw[ForestGreen, opacity=.5](260.42, 725.438) .. controls (252.539, 719.814) and (254.2695, 707.907) .. (258.108, 694.265) .. controls (261.9465, 680.623) and (267.893, 665.246) .. (279.374, 667.219);
      \draw[ForestGreen](294.233, 652.127) .. controls (302.983, 650.866) and (301.4915, 631.433) .. (293.465, 616.0387) .. controls (285.4385, 600.6445) and (270.877, 589.289) .. (259.472, 595.569);
      \draw[ForestGreen, opacity=.5](259.472, 595.569) .. controls (250.614, 600.51) and (255.307, 616.255) .. (262.4798, 629.9933) .. controls (269.6525, 643.7315) and (279.305, 655.463) .. (294.233, 652.127);
      \draw(189.688, 687.994) .. controls (208, 688) and (236, 708) .. (254.6667, 721.3333) .. controls (273.3333, 734.6667) and (282.6667, 741.3333) .. (297.3333, 748) .. controls (312, 754.6667) and (332, 761.3333) .. (352, 765.6667) .. controls (372, 770) and (392, 772) .. (428.052, 772);
      \draw(194.593, 624.022) .. controls (208, 624) and (230, 612) .. (250.3333, 600.6667) .. controls (270.6667, 589.3333) and (289.3333, 578.6667) .. (310.6667, 570.3333) .. controls (332, 562) and (356, 556) .. (464.699, 587.884);
      \draw(460, 612) circle[radius=1, cm={16,-4,0,24,(0,0)}];
      \draw(192, 656) circle[radius=1, cm={16,0,-2,32,(0,0)}];
      \draw(448, 688) ellipse[x radius=12, y radius=20];
      \draw(428, 752) ellipse[x radius=12, y radius=20];
      \draw(428.493, 732.017) .. controls (412, 728) and (424, 708) .. (447.774, 707.996);
      \draw(447.984, 668) .. controls (428, 664) and (426, 654) .. (430, 647) .. controls (434, 640) and (444, 636) .. (456.329, 636.278);
      \draw(248, 668) .. controls (260, 660) and (276, 656) .. (290, 653) .. controls (304, 650) and (316, 648) .. (336, 660);
      \draw(258.448, 662.348) .. controls (284, 672) and (316, 668) .. (328.024, 655.739);
      \draw(328, 728) .. controls (349.3333, 714.6667) and (373.3333, 709.3333) .. (400, 712);
      \draw(341.458, 720.808) .. controls (359.1527, 725.6027) and (374.1273, 722.4447) .. (386.382, 711.334);
      \draw(344, 616) .. controls (372, 596) and (408, 616) .. (412, 624);
      \draw(357.035, 609.515) .. controls (368, 620) and (392, 624) .. (405.362, 617.681);
      \node[ForestGreen] at (256, 740) {$\gamma$};
      \node[ForestGreen] at (252, 576) {$\gamma'$};
    \end{scope}

    \begin{scope}[xshift=26cm]
      \draw[thin,opacity=.5](207.7156, 659.2374) -- (216.1697, 659.2375) -- (216.2145, 651.1117);;
      \draw[BrickRed, opacity=.5](234.863, 609.1302) .. controls (239.9432, 606.4022) and (244.9454, 611.3676) .. (249.6765, 617.4074) .. controls (254.4076, 623.4472) and (258.8675, 630.5614) .. (261.4032, 639.7206) .. controls (263.9389, 648.8798) and (264.5504, 660.0841) .. (258.448, 662.348);
      \draw[BrickRed](208.1211, 651.3745) .. controls (222.1133, 651.443) and (227.1733, 649.4002) .. (229.9307, 644.5661) .. controls (232.688, 639.7321) and (233.1427, 632.1069) .. (232.6948, 624.9749) .. controls (232.247, 617.8428) and (230.8965, 611.2038) .. (234.863, 609.1302);
      \draw[BrickRed](258.448, 662.348) .. controls (253.6513, 664.6681) and (249.0122, 661.1917) .. (244.4677, 656.9486) .. controls (239.9231, 652.7056) and (235.4732, 647.6958) .. (232.16, 637.9908);
      \draw[ForestGreen](279.374, 667.219) .. controls (290.525, 667.722) and (291.2625, 685.861) .. (285.9773, 701.866) .. controls (280.692, 717.871) and (269.384, 731.742) .. (260.42, 725.438);
      \draw[ForestGreen, opacity=.5](260.42, 725.438) .. controls (252.539, 719.814) and (254.2695, 707.907) .. (258.108, 694.265) .. controls (261.9465, 680.623) and (267.893, 665.246) .. (279.374, 667.219);
      \draw[ForestGreen](294.245, 655.236) .. controls (306.214, 655.877) and (305.107, 633.9385) .. (296.2727, 617.2915) .. controls (287.4385, 600.6445) and (270.877, 589.289) .. (259.472, 595.569);
      \draw[ForestGreen, opacity=.5](259.472, 595.569) .. controls (250.614, 600.51) and (255.307, 616.255) .. (262.9025, 630.1885) .. controls (270.498, 644.122) and (280.996, 656.244) .. (294.245, 655.236);
      \draw(189.688, 687.994) .. controls (208, 688) and (236, 708) .. (254.6667, 721.3333) .. controls (273.3333, 734.6667) and (282.6667, 741.3333) .. (297.3333, 748) .. controls (312, 754.6667) and (332, 761.3333) .. (352, 765.6667) .. controls (372, 770) and (392, 772) .. (428.052, 772);
      \draw(460, 612) circle[radius=1, cm={16,-4,0,24,(0,0)}];
      \draw(448, 688) ellipse[x radius=12, y radius=20];
      \draw(428, 752) ellipse[x radius=12, y radius=20];
      \draw(447.984, 668) .. controls (428, 664) and (426, 654) .. (430, 647) .. controls (434, 640) and (444, 636) .. (456.329, 636.278);
      \draw(258.448, 662.348) .. controls (284, 668) and (302, 672) .. (317.6667, 676) .. controls (333.3333, 680) and (346.6667, 684) .. (359.3333, 690) .. controls (372, 696) and (384, 704) .. (394, 712) .. controls (404, 720) and (412, 728) .. (428, 732);
      \draw(344, 616) .. controls (372, 596) and (408, 616) .. (412, 624);
      \draw(357.035, 609.515) .. controls (368, 620) and (392, 624) .. (405.362, 617.681);
      \draw(312, 716) .. controls (344, 712) and (372, 732) .. (368, 744);
      \draw(324.731, 715.617) .. controls (328, 724) and (352, 736) .. (366.568, 735.717);
      \draw(248, 668) .. controls (260, 660) and (276, 656) .. (290.6667, 655.3333) .. controls (305.3333, 654.6667) and (318.6667, 657.3333) .. (330, 662) .. controls (341.3333, 666.6667) and (350.6667, 673.3333) .. (371.3333, 684.6667) .. controls (392, 696) and (424, 712) .. (448, 708);
      \draw(194.593, 624.022) .. controls (208, 624) and (230, 612) .. (250.3333, 600.6667) .. controls (270.6667, 589.3333) and (289.3333, 578.6667) .. (310.6667, 570.3333) .. controls (332, 562) and (356, 556) .. (464.699, 587.884);
      \draw(192, 656) circle[radius=1, cm={16,0,-2,32,(0,0)}];
      \node[ForestGreen] at (256, 740) {$\gamma$};
      \node[ForestGreen] at (252, 576) {$\gamma'$};
    \end{scope}
  \end{tikzpicture}
  }
  \caption{
    \textbf{Geometric Origin of the Kernels:} By shooting an orthogeodesic (in red) from the first boundary component of the surface $\Sigma$, one determines one or two simple closed curves (in green). Different behaviors can arise: on the left, the orthogeodesic intersects the boundary component $\partial_m \Sigma$ ($B_m$-type), determining a single internal geodesic $\gamma$. In the two other cases, the orthogeodesic intersects $\partial_1 \Sigma$ or itself ($C$-type), determining two internal geodesics $\gamma$ and $\gamma'$. The kernels compute the probability of these different behaviors occurring.
    }
  \label{fig:orthogeodesic:behaviour}
\end{figure}

The geometric origin of the recursion kernels is also parallel in the two models. One picks a random point on the first boundary component $\partial_1 \Sigma$ of the underlying surface $\Sigma$, where “random” means distributed according to the probability measure induced by the hyperbolic or Strebel metric. From this random point, one shoots an orthogeodesic. This orthogeodesic determines a unique pair of pants, and topologically there are only two possible configurations (see \cref{fig:orthogeodesic:behaviour}):
\begin{description}
    \item[\boldmath$B_m$-type:] The pair of pants bounds two external boundary components $\partial_1 \Sigma$ and $\partial_m \Sigma$, together with an internal geodesic $\gamma$.
    \item[\boldmath$C$-type:] The pair of pants bounds the first external boundary component $\partial_1 \Sigma$ together with two internal geodesics $\gamma$ and $\gamma'$.
\end{description}

The hyperbolic and combinatorial $B$ and $C$ kernels are thus the probabilities, with respect to the hyperbolic or Strebel metric, that the pair of pants associated with a random point on the first boundary component is of $B$- or $C$-type:
\begin{equation}
\begin{aligned}
    B^{\textup{hyp}}(L_1,L_m,\ell)
    &=
    \mathrm{Prob}^{\textup{hyp}}\!\left(
        \substack{
        \text{point in $\partial_1\Sigma$} \\
        \text{determines a $B_m$-type pair of pants} \\
        \text{with boundary lengths $(L_1,L_m,\ell)$}
        }
    \right), \\
    C^{\textup{hyp}}(L_1,\ell,\ell')
    &=
    \mathrm{Prob}^{\textup{hyp}}\!\left(
        \substack{
        \text{point in $\partial_1\Sigma$} \\
        \text{determines a $C$-type pair of pants} \\
        \text{with boundary lengths $(L_1,\ell,\ell')$}
        }
    \right).
\end{aligned}
\end{equation}
The exact same interpretation applies to the combinatorial kernels, with the notion of probability defined using the Strebel metric instead of the hyperbolic one. This explains why the kernels take value in $[0,1]$. The hyperbolic probabilities are computed by Mirzakhani in \cite{Mir07a}, while the combinatorial ones are computed in \cite{ABCGLW}. A general theory producing topological recursion relations from functions on moduli spaces, called geometric recursion, was developed in \cite{ABO17}. Applications of geometric recursion to other volumes on moduli space include Masur--Veech volumes \cite{Andersen:MV}, whose JT gravity interpretation was found in \cite{Fuji:MV}.

\section{Pruned matrix correlators as discrete volumes}
\label{sec:matrix}
We now move to random matrix theory. The main point of this section is that certain matrix model correlators, called \emph{pruned traces}, define in a precise sense some \emph{discrete volumes} of moduli space, which we denote as
\begin{equation} \label{eq:mdiscvols}
    N_{g,n}(b_1,\ldots,b_n)
    \coloneqq
    \left\langle \prod_{i=1}^n \frac{1}{b_{i}} \NTr{M^{b_i}} \right\rangle_{\!\!g,\cc} .
\end{equation}
Such pruned traces are defined from the matrix integral, either diagrammatically or via topological recursion on the associated spectral curve. In the special case where the matrix integral is purely Gaussian, they admit an independent definition through the combinatorial description of the moduli space: they coincide with Norbury's lattice point counts $N^{\textup{Nor}}_{g,n}$ on $\M_{g,n}$, which enumerate integer Strebel graphs as reviewed in the previous section. The discreteness of the volumes is fundamentally tied to the fact that we study matrix integrals in a standard 't~Hooft limit rather than the double-scaling limit. A discrete analog of the Kontsevich model \cite{Kon92} had been presciently discussed by Chekhov \cite{ChekhovDisc} using a matrix integral introduced in \cite{ChekhovMakeenkoKP}.

In what follows, we first explain how to define pruned correlators in a generic one-cut matrix model. We then show how, in the GUE case, they reproduce Norbury's discrete counting of lattice points on the moduli space of curves. This construction follows the approach of \cite{Gop04a,Gop04b,Gop05} developed in the context of gauge/string duality. In essence, in the Feynman diagram expansion of the matrix correlator, each graph can be naturally identified with a point on the moduli space. Finally, we extend the discussion to interacting matrix models and demonstrate that a similar notion of discreteness persists to all orders in perturbation theory. Even non-perturbatively, a remnant of this discreteness remains: the parameters $b_i$, representing the boundary lengths of the dual Riemann surfaces, take integer values, consistent with their origin as matrix powers.

\subsection{Traces: standard vs. pruned}
Consider the following large $\ms{N}$ Hermitian matrix model with a single-trace potential:
\begin{equation} \label{eq:MM}
    Z_{\ms{N}}
    \coloneqq
    \int_{\mathcal{H}_{\ms{N} \times \ms{N}}} dM \ e^{-\ms{N} \Tr(V(M))},
    \qquad
    dM = \frac{2^{\frac{N(N-1)}{2}}}{\mathrm{Vol}( \mathrm{U}(\ms{N}))}
    \prod_i dM_{ii} \prod_{i<j} d\Re M_{ij}\, d\Im M_{ij}
\end{equation}
with \smash{$\mathrm{Vol}(\mathrm{U}(\ms{N})) = (2\pi)^{\frac{N(N+1)}{2}}/\prod_{k=1}^{N-1}k!$} and $V(M)$ is an arbitrary potential. We assume that the eigenvalue distribution of $M$ is supported on a single interval $[a_-, a_+]$, in which case the model is said to be in the \emph{one-cut phase}.

The $n$-point functions of \emph{standard traces} are defined by
\begin{equation}
    \left\langle \prod_{i=1}^n \Tr{M^{b_i}} \right\rangle
    \coloneqq
    \frac{1}{Z_{\ms{N},0}}
    \int_{\mathcal{H}_{\ms{N} \times \ms{N}}}
        dM \ e^{-\ms{N} \Tr(V(M))} \prod_{i=1}^n \Tr{M^{b_i}},
\end{equation}
where $Z_{\ms{N},0}$ is the partition function with purely Gaussian potential, i.e. $V(M) = \frac{M^2}{2}$. Their connected version, denoted by the subscript $\cc$, admits a natural $1/\ms{N}$ expansion:
\begin{equation}
     \left\langle \prod_{i=1}^n \Tr{M^{b_i}} \right\rangle_{\!\!\cc}
     =
     \sum_{g \ge 0} \ms{N}^{2-2g-n}
     \left\langle \prod_{i=1}^n \Tr{M^{b_i}} \right\rangle_{\!\!g,\cc}.
\end{equation}
These standard traces are conveniently encoded in a genus-$g$, $n$-point differential:
\begin{equation}\label{eq:Wgn}
\begin{split}
    \omega_{g,n}(z_1,\ldots,z_n)
    &\coloneqq
    \left\langle \prod_{i=1}^n \Tr{ \frac{dx_i(z)}{x_i(z) - M} } \right\rangle_{\!\!g,\cc} \\
    &=
    \sum_{b_{1},\ldots,b_n = 1}^{\infty}
        \left\langle \prod_{i=1}^n \Tr{M^{b_i}} \right\rangle_{\!\!g,\cc} \,
        \prod_{i=1}^n x_i(z)^{b_i-1} \, dx_{i}(z),
\end{split}
\end{equation}
where
\begin{equation}\label{eq:Joukowsky}
    x(z) = \gamma \left(z + \frac{1}{z}\right) + \delta,
    \qquad
    \gamma = \frac{a_+ - a_-}{4},
    \qquad
    \delta = \frac{a_+ + a_-}{2},
\end{equation}
is the Joukowsky variable (cf.~\cref{sec:discrete:rec}).

The matrix correlators relevant to the discrete volumes are not those of standard traces, but rather those of \emph{pruned traces}, denoted by $\NTr{M^b}$. As mentioned in the introduction, pruning can be viewed as a genus-zero analog of normal ordering, hence the notation, in the sense that the planar one-point function vanishes, $\langle \NTr{ M^{b} } \rangle_{g=0} = 0$, though higher-genus contributions may not. Diagrammatically, pruning corresponds to removing all petals from Feynman diagrams, where petals represent planar Wick contractions between neighboring edges attached to the same vertex. This interpretation is encoded in the fact that $x(z)$ is essentially the generating function of the Catalan numbers counting such petals.

Concretely, connected correlators of pruned traces are neatly related to those of standard traces as
\begin{equation}\label{eq:corr:TR}
     \omega_{g,n}(z_1,\ldots,z_n)
     =
     \sum_{b_{1},\ldots,b_n=1}^{\infty}
        \underbrace{\left\langle \prod_{i=1}^n \frac{1}{b_i} \NTr{ M^{b_i} } \right\rangle_{\!\!g,\cc}}_{\eqqcolon N_{g,n}(b_1,\ldots,b_n)} \
        \prod_{i=1}^n b_i z_i^{b_i-1} \, dz_{i} .
\end{equation}
The above equation defines the connected correlators of pruned traces. Such quantities have been considered in the mathematical literature by Norbury and Scott in \cite{NS13}, purely from the perspective of abstract topological recursion. For instance, they prove the quasi-polynomiality\footnote{
	A function $N(b_1,\dots,b_n)$ is called a \emph{quasi-polynomial} if it restricts to an honest polynomial on each coset of the sublattice $2\Z^n \subset \Z^n$. Equivalently, $N$ can be expressed as a polynomial in the variables $b_1,\dots,b_n$ and in the parity indicators $(-1)^{b_1},\dots,(-1)^{b_n}$.
} of these quantities, a property that is far from transparent from the matrix model perspective. We also mention that the above relation between the pruned correlators and the correlation differentials is nothing but a \emph{discrete Laplace transform} (also known as the Z-transform in signal processing theory). Connections between the Eynard--Orantin topological recursion and the Laplace transform have been extensively studied in the literature, especially in the context of mirror symmetry. One novelty here, in accordance with the motto of the paper, is its discrete flavor.

As pointed out to us by A.~Levine, one can use the Joukowsky map to summarize the relation between pruned and standard traces succinctly in terms of Chebyshev polynomials of the first kind, defined by $T_b(\cos\theta) = \cos(b\theta)$:
\begin{equation} \label{eq:prune:Cheby}
    \frac{1}{b}\NTr{ M^{b} }
    \;\longleftrightarrow\;
    \Tr 2\,T_{b}\Bigl( \frac{M - \delta}{2\gamma} \Bigr).
\end{equation}
The correspondence should be understood as an identity holding inside any correlator.

\subsection{From GUE to lattice points on \texorpdfstring{$\M_{g,n}$}{the moduli of curves}} 
\label{subsec:GUE:correl}
In this subsection, we explain how the correlators of pruned traces in the purely Gaussian case are connected with the lattice point count on the moduli space of curves discussed in the previous section:
\begin{equation} \label{eq:GUE:Nor}
    N^{\textup{Nor}}_{g,n}(b_1,\ldots,b_n)
    =
    \left\langle \prod_{i=1}^n \frac{1}{b_{i}} \NTr{M^{b_i}} \right\rangle_{\!\!g,\cc}^{\!\!\textup{GUE}} .
\end{equation}
In this case, the eigenvalues distribution is supported on $[-2,2]$, hence $\gamma = 1$ and $\delta = 0$ in the Joukowsky map \labelcref{eq:Joukowsky}.

This correspondence admits a diagrammatic interpretation, first articulated in \cite{Gop04a,Gop04b,Gop05}. In the Gaussian matrix model, the observables on the right-hand side of \cref{eq:GUE:Nor} can be computed via free-field Wick contractions. Rephrased diagrammatically, one computes the correlators by summing over all (topologically nonequivalent) Feynman diagrams with only external vertices. As explained above, pruning corresponds to removing petals, i.e. planar Wick contractions between adjacent edges attached to the same vertex. The valence of an external vertex equals the power of the corresponding trace insertion in the expectation value. Fixing the power of $\ms{N}$, the size of the matrix, selects the genus $g$ of the diagram.

However, these Feynman diagrams cannot be directly identified with integer Strebel graphs for two reasons. First, Strebel graphs have as many \emph{faces} as boundaries, whereas our correlators generate Feynman diagrams with as many \emph{vertices} as boundaries, corresponding to the number of single-trace operators. This suggests they are graph dual to one another. However, Strebel graphs are required to have vertices of valency at least three, while matrix model Feynman diagrams can have two-sided faces. Second, Strebel graphs are \emph{metrized} ribbon graphs, whereas matrix model Feynman diagrams do not naturally carry a notion of edge-length.

To resolve these mismatches, we construct the associated Strebel graph starting from the pruned Feynman diagram in two steps, illustrated in the left panel of \cref{fig:FDtoMgn}:
\begin{enumerate}[label=\roman*)]
    \item First, assign length~$1$ to each edge of the Feynman diagram. Then identify homotopic edges, namely those that bound two-sided faces, and collapse all such homotopic edges into a single effective edge carrying a length equal to the number of collapsed edges. The resulting diagram is called the \emph{skeleton graph} of the original Feynman diagram; it has no petals and no two-sided faces, but edges carry integer edge-lengths.
    
    \item Second, take the \emph{graph dual} of the skeleton graph. This exchanges vertices and faces: since the skeleton graph has no petals and no two-sided faces, its dual automatically has vertices of valency three or higher, as required for Strebel graphs. The duality map preserves edge adjacencies, and the integer edge lengths carry over to the dual graph. The resulting graph is the sought integer Strebel graph.
\end{enumerate}

\begin{figure}[t]
    \centering
    \tikzset{every picture/.style=thick}
    \begin{tikzpicture}[x=1pt,y=1pt,scale=.7]
        \draw[line width=5pt] (56, 440) .. controls (93.3333, 477.3333) and (133.3333, 477.3333) .. (176, 440);
        \draw[white, line width=3.5pt] (56, 440) .. controls (93.3333, 477.3333) and (133.3333, 477.3333) .. (176, 440);
        \draw[line width=5pt] (56, 364) rectangle (176, 440);
        \draw[white, line width=3.5pt] (56, 364) rectangle (176, 440);
        \filldraw[fill=white](176, 364) circle[radius=5.6569];
        \filldraw[fill=white](56, 364) circle[radius=5.6569];
        \filldraw[fill=white](56, 440) circle[radius=5.6569];
        \filldraw[fill=white](176, 440) circle[radius=5.6569];
        \draw[BurntOrange,line width=5pt](116, 400) .. controls (76, 394.6667) and (42.6667, 402.6667) .. (16, 424);
        \draw[white, line width=3.5pt](116, 400) .. controls (76, 394.6667) and (42.6667, 402.6667) .. (16, 424);
        \draw[BurntOrange,line width=5pt](116, 400) .. controls (124, 456) and (64, 516) .. (16, 424);
        \draw[white, line width=3.5pt](116, 400) .. controls (124, 456) and (64, 516) .. (16, 424);
        \draw[BurntOrange,line width=5pt](116, 400) .. controls (224, 400) and (212, 440) .. (188.6667, 464) .. controls (165.3333, 488) and (130.6667, 496) .. (91.3333, 491) .. controls (52, 486) and (8, 468) .. (16, 424);
        \draw[white, line width=3.5pt](116, 400) .. controls (224, 400) and (212, 440) .. (188.6667, 464) .. controls (165.3333, 488) and (130.6667, 496) .. (91.3333, 491) .. controls (52, 486) and (8, 468) .. (16, 424);
        \draw[BurntOrange,line width=5pt](16, 424) .. controls (4, 336) and (100, 312) .. (116, 400);
        \draw[white, line width=3.5pt](16, 424) .. controls (4, 336) and (100, 312) .. (116, 400);
        \filldraw[BurntOrange,fill=white] (16, 424) circle[radius=5.6569];
        \filldraw[BurntOrange,fill=white] (116, 400) circle[radius=5.6569];
        \draw(440, 360) .. controls (472, 381.3333) and (506.6667, 381.3333) .. (544, 360) .. controls (533.3333, 392) and (546.6667, 402.6667) .. (584, 392) .. controls (562.6667, 418.6667) and (565.3333, 440) .. (592, 456) .. controls (560, 445.3333) and (530.6667, 461.3333) .. (504, 504) .. controls (472, 456) and (445.3333, 440) .. (424, 456) .. controls (466.6667, 418.6667) and (472, 386.6667) .. (440, 360) -- cycle;
        \node[BurntOrange] at (472, 384) {\scriptsize$\bullet$};
        \node[BurntOrange] at (472, 416) {\scriptsize$\bullet$};
        \node[BurntOrange] at (452, 444) {\scriptsize$\bullet$};
        \node[BurntOrange] at (492, 464) {\scriptsize$\bullet$};
        \node[BurntOrange] at (500, 444) {\scriptsize$\bullet$};
        \node[BurntOrange] at (528, 436) {\scriptsize$\bullet$};
        \node[BurntOrange] at (524, 460) {\scriptsize$\bullet$};
        \node[BurntOrange] at (560, 428) {\scriptsize$\bullet$};
        \node[BurntOrange] at (552, 416) {\scriptsize$\bullet$};
        \node[BurntOrange] at (532, 400) {\scriptsize$\bullet$};
        \node[BurntOrange] at (520, 380) {\scriptsize$\bullet$};
        \node[BurntOrange] at (516, 416) {\scriptsize$\bullet$};
        \node[BurntOrange] at (500, 404) {\scriptsize$\bullet$};
        \node at (105, 454) [right] {\scriptsize$2$};
        \node at (148, 488) [above] {\scriptsize$1$};
        \node at (64, 344) [below] {\scriptsize$1$};
        \node at (80, 400) [above] {\scriptsize$1$};
        \draw[<-] (80, 333).. controls (108.6243, 329.4701) and (110.908, 339.7071) .. (127.7074, 340.0633);
        \node at (156, 345) {\scriptsize{integral}};
        \node at (156, 330) {\scriptsize{Strebel}};
        \node at (156, 315) {\scriptsize{coordinates}};
        \node at (310, 478) {\small{Feynman diagram}};
        \node at (310, 420) {\small{integer Strebel graph}};
        \node at (310, 362) {\small{lattice point on moduli space}};
        \draw[<->](310, 436) -- (310, 462);
        \draw[<->](310, 378) -- (310, 404);
        \node at (570, 360) {$\mathcal{M}_{g,n}$};
    \end{tikzpicture}
    \caption{
    \textbf{Matrix correlators as lattice point counts on $\M_{g,n}$.}
    Expanding matrix model correlators in terms of Feynman diagrams allows one to reinterpret them as a count of discrete lattice points on the moduli space. Collapsing homotopic edges of a Feynman diagram yields a skeleton graph with integer edge lengths. By taking its graph dual, one obtains an integer Strebel graph parameterizing the corresponding point on moduli space.
    }
    \label{fig:FDtoMgn}
\end{figure}
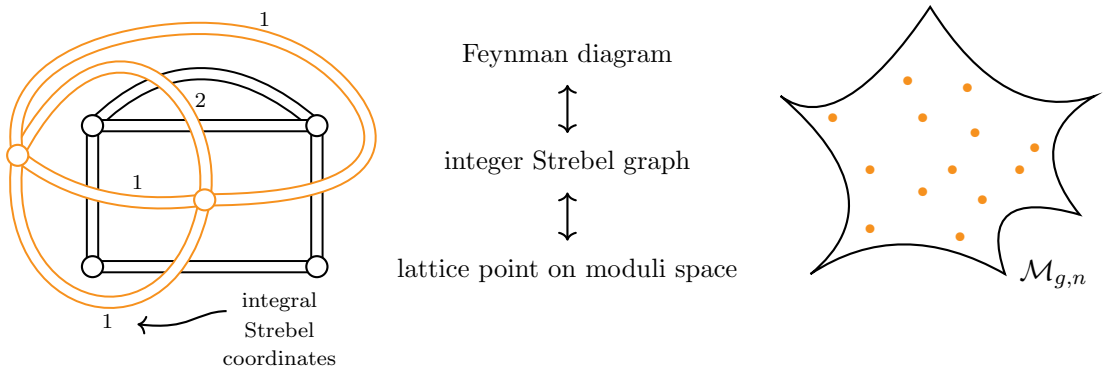

This construction establishes a one-to-one correspondence between each set of Wick contractions (equivalently, each Feynman diagram) and a point in the combinatorial moduli space. The edge lengths of the resulting Strebel graph serve as coordinates on this space; their integrality produces a discrete subset of points. Since the combinatorial moduli space is isomorphic to $\M_{g,n}$ by Strebel's theorem, this discrete subset corresponds precisely to the lattice points of $\M_{g,n}$ (see the right panel of \cref{fig:FDtoMgn}).

It is worth noting that these Riemann surfaces are special: by a theorem of G.~V.~Belyĭ, they correspond to the arithmetic points on $\M_{g,n}$. Arithmetic surfaces are defined as the zero-locus of complex polynomials with coefficients in the algebraic numbers $\overline{\mathbb{Q}}$. They play a central role in Grothendieck's theory of dessins d'enfants (``children's drawings'') and exhibit deep number-theoretic properties.

We now illustrate this construction for the simple GUE-observable $\langle \frac{1}{6} \NTr{M^6} \rangle_{g=1,\cc}$. There are two topologically nonequivalent pruned diagrams: a first diagram without homotopic edges, and a second one with two homotopic edges. In collapsing the two homotopic edges in the second diagram, we obtain one edge of length $2$. The skeleton graphs are then dual to integer Strebel graphs, drawn in orange.
{\allowdisplaybreaks
\begin{align}
    \notag
    \left\langle \frac{1}{6} \NTr{M^6} \right\rangle_{\!\!g=1,\cc}^{\!\!\textup{GUE}}
    &\!\!\!\overset{\hphantom{\text{skeleton}}}{=}\!\!\!
    \begin{tikzpicture}[x=1pt,y=1pt,scale=.4,baseline={(0, 10cm)}]
        \draw[line width=5pt](128, 704) .. controls (96, 736) and (112, 760) .. (133.3333, 758.6667) .. controls (154.6667, 757.3333) and (181.3333, 730.6667) .. (182.6667, 709.3333) .. controls (184, 688) and (160, 672) .. (128, 704);
        \draw[white, line width=3.5pt](128, 704) .. controls (96, 736) and (112, 760) .. (133.3333, 758.6667) .. controls (154.6667, 757.3333) and (181.3333, 730.6667) .. (182.6667, 709.3333) .. controls (184, 688) and (160, 672) .. (128, 704);
        \draw[line width=5pt](128, 704) .. controls (192, 704) and (192, 736) .. (181.3333, 757.3333) .. controls (170.6667, 778.6667) and (149.3333, 789.3333) .. (128, 789.3333) .. controls (106.6667, 789.3333) and (85.3333, 778.6667) .. (74.6667, 757.3333) .. controls (64, 736) and (64, 704) .. (128, 704);
        \draw[white, line width=3.5pt](128, 704) .. controls (192, 704) and (192, 736) .. (181.3333, 757.3333) .. controls (170.6667, 778.6667) and (149.3333, 789.3333) .. (128, 789.3333) .. controls (106.6667, 789.3333) and (85.3333, 778.6667) .. (74.6667, 757.3333) .. controls (64, 736) and (64, 704) .. (128, 704);
        \draw[line width=5pt](128, 704) .. controls (176, 752) and (144, 768) .. (117.3333, 765.3333) .. controls (90.6667, 762.6667) and (69.3333, 741.3333) .. (66.6667, 714.6667) .. controls (64, 688) and (80, 656) .. (128, 704);
        \draw[white, line width=3.5pt](128, 704) .. controls (176, 752) and (144, 768) .. (117.3333, 765.3333) .. controls (90.6667, 762.6667) and (69.3333, 741.3333) .. (66.6667, 714.6667) .. controls (64, 688) and (80, 656) .. (128, 704);
        \node at (128, 792) [above] {\tiny$1$};
        \node at (64, 696) [left]{\tiny$1$};
        \node at (184, 696) [right]{\tiny$1$};
        \filldraw[thick, fill=white](128, 704) circle[radius=8];
    \end{tikzpicture}
    +
    \begin{tikzpicture}[x=1pt,y=1pt,scale=.4,baseline={(0, 10cm)}]
        \draw[line width=5pt](128, 704) .. controls (96, 752) and (80, 728) .. (80, 704) .. controls (80, 680) and (96, 656) .. (128, 704);
        \draw[white, line width=3.5pt](128, 704) .. controls (96, 752) and (80, 728) .. (80, 704) .. controls (80, 680) and (96, 656) .. (128, 704);
        \draw[line width=5pt](128, 704) .. controls (176, 736) and (120, 752) .. (86.6667, 749.3333) .. controls (53.3333, 746.6667) and (42.6667, 725.3333) .. (42.6667, 704) .. controls (42.6667, 682.6667) and (53.3333, 661.3333) .. (86.6667, 658.6667) .. controls (120, 656) and (176, 672) .. (128, 704);
        \draw[white, line width=3.5pt](128, 704) .. controls (176, 736) and (120, 752) .. (86.6667, 749.3333) .. controls (53.3333, 746.6667) and (42.6667, 725.3333) .. (42.6667, 704) .. controls (42.6667, 682.6667) and (53.3333, 661.3333) .. (86.6667, 658.6667) .. controls (120, 656) and (176, 672) .. (128, 704);
        \draw[line width=5pt](128, 704) .. controls (192, 704) and (184, 744) .. (164, 764) .. controls (144, 784) and (112, 784) .. (92, 764) .. controls (72, 744) and (64, 704) .. (128, 704);
        \draw[white, line width=3.5pt](128, 704) .. controls (192, 704) and (184, 744) .. (164, 764) .. controls (144, 784) and (112, 784) .. (92, 764) .. controls (72, 744) and (64, 704) .. (128, 704);
        \node at (128, 784) [above] {\tiny$1$};
        \node at (82, 690) [left] {\tiny$1$};
        \node at (50, 740) [left] {\tiny$1$};
        \filldraw[thick,fill=white](128, 704) circle[radius=8];
    \end{tikzpicture}
    \\[-1ex]
    &\!\!\!\overset{\text{skeleton}}{=}\!\!\!
    \begin{tikzpicture}[x=1pt,y=1pt,scale=.4,baseline={(0, 10cm)}]
        \draw[line width=5pt](128, 704) .. controls (96, 736) and (112, 760) .. (133.3333, 758.6667) .. controls (154.6667, 757.3333) and (181.3333, 730.6667) .. (182.6667, 709.3333) .. controls (184, 688) and (160, 672) .. (128, 704);
        \draw[white, line width=3.5pt](128, 704) .. controls (96, 736) and (112, 760) .. (133.3333, 758.6667) .. controls (154.6667, 757.3333) and (181.3333, 730.6667) .. (182.6667, 709.3333) .. controls (184, 688) and (160, 672) .. (128, 704);
        \draw[line width=5pt](128, 704) .. controls (192, 704) and (192, 736) .. (181.3333, 757.3333) .. controls (170.6667, 778.6667) and (149.3333, 789.3333) .. (128, 789.3333) .. controls (106.6667, 789.3333) and (85.3333, 778.6667) .. (74.6667, 757.3333) .. controls (64, 736) and (64, 704) .. (128, 704);
        \draw[white, line width=3.5pt](128, 704) .. controls (192, 704) and (192, 736) .. (181.3333, 757.3333) .. controls (170.6667, 778.6667) and (149.3333, 789.3333) .. (128, 789.3333) .. controls (106.6667, 789.3333) and (85.3333, 778.6667) .. (74.6667, 757.3333) .. controls (64, 736) and (64, 704) .. (128, 704);
        \draw[line width=5pt](128, 704) .. controls (176, 752) and (144, 768) .. (117.3333, 765.3333) .. controls (90.6667, 762.6667) and (69.3333, 741.3333) .. (66.6667, 714.6667) .. controls (64, 688) and (80, 656) .. (128, 704);
        \draw[white, line width=3.5pt](128, 704) .. controls (176, 752) and (144, 768) .. (117.3333, 765.3333) .. controls (90.6667, 762.6667) and (69.3333, 741.3333) .. (66.6667, 714.6667) .. controls (64, 688) and (80, 656) .. (128, 704);
        \node at (128, 792) [above] {\tiny$1$};
        \node at (64, 696) [left]{\tiny$1$};
        \node at (184, 696) [right]{\tiny$1$};
        \filldraw[thick, fill=white](128, 704) circle[radius=8];
    \end{tikzpicture}
    +
    \begin{tikzpicture}[x=1pt,y=1pt,scale=.4,baseline={(0, 10cm)}]
        \draw[line width=5pt](128, 704) .. controls (96, 768) and (64, 736) .. (64, 704) .. controls (64, 672) and (96, 640) .. (128, 704);
        \draw[white, line width=3.5pt](128, 704) .. controls (96, 768) and (64, 736) .. (64, 704) .. controls (64, 672) and (96, 640) .. (128, 704);
        \draw[line width=5pt](128, 704) .. controls (192, 704) and (184, 744) .. (164, 764) .. controls (144, 784) and (112, 784) .. (92, 764) .. controls (72, 744) and (64, 704) .. (128, 704);
        \draw[white, line width=3.5pt](128, 704) .. controls (192, 704) and (184, 744) .. (164, 764) .. controls (144, 784) and (112, 784) .. (92, 764) .. controls (72, 744) and (64, 704) .. (128, 704);
        \node at (128, 782) [above] {\tiny$1$};
        \node at (62, 690) [left] {\tiny$2$};
        \filldraw[thick,fill=white](128, 704) circle[radius=8];
    \end{tikzpicture}
    \\[-2ex]
    \notag
    &\!\!\!\overunderset{\text{dual}}{\hphantom{\text{skeleton}}}{=}\!\!\!
    \begin{tikzpicture}[x=1pt,y=1pt,scale=.4,baseline={(0, 10cm)}]
        \draw[BurntOrange,line width=5pt](128, 704) -- (128, 752);
        \draw[white, line width=3.5pt](128, 704) -- (128, 752);  
        \draw[BurntOrange,line width=5pt](128, 752) .. controls (128, 784) and (104, 784) .. (92, 765.3333) .. controls (80, 746.6667) and (80, 709.3333) .. (92, 690.6667) .. controls (104, 672) and (128, 672) .. (128, 704);
        \draw[white, line width=3.5pt](128, 752) .. controls (128, 784) and (104, 784) .. (92, 765.3333) .. controls (80, 746.6667) and (80, 709.3333) .. (92, 690.6667) .. controls (104, 672) and (128, 672) .. (128, 704);
        \draw[BurntOrange,line width=5pt](128, 752) .. controls (160, 752) and (160, 776) .. (152, 790.6667) .. controls (144, 805.3333) and (128, 810.6667) .. (112, 808) .. controls (96, 805.3333) and (80, 794.6667) .. (72, 769.3333) .. controls (64, 744) and (64, 704) .. (128, 704);
        \draw[white, line width=3.5pt](128, 752) .. controls (160, 752) and (160, 776) .. (152, 790.6667) .. controls (144, 805.3333) and (128, 810.6667) .. (112, 808) .. controls (96, 805.3333) and (80, 794.6667) .. (72, 769.3333) .. controls (64, 744) and (64, 704) .. (128, 704);
        \filldraw[thick,BurntOrange,fill=white](128, 752) circle[radius=8];
        \filldraw[thick,BurntOrange,fill=white](128, 704) circle[radius=8];
        \node at (126, 728) [right] {\tiny$1$};
        \node at (122, 790) {\tiny$1$};
        \node at (76, 784) [left] {\tiny$1$};
    \end{tikzpicture}
    +
    \begin{tikzpicture}[x=1pt,y=1pt,scale=.4,baseline={(0, 10cm)}]
        \draw[BurntOrange,line width=5pt](128, 720) .. controls (128, 768) and (96, 768) .. (80, 752) .. controls (64, 736) and (64, 704) .. (80, 688) .. controls (96, 672) and (128, 672) .. (128, 720);
        \draw[white,line width=3.5pt](128, 720) .. controls (128, 768) and (96, 768) .. (80, 752) .. controls (64, 736) and (64, 704) .. (80, 688) .. controls (96, 672) and (128, 672) .. (128, 720);
        \draw[BurntOrange,line width=5pt](128, 720) .. controls (176, 720) and (176, 752) .. (160, 768) .. controls (144, 784) and (112, 784) .. (96, 768) .. controls (80, 752) and (80, 720) .. (128, 720);
        \draw[white,line width=3.5pt](128, 720) .. controls (176, 720) and (176, 752) .. (160, 768) .. controls (144, 784) and (112, 784) .. (96, 768) .. controls (80, 752) and (80, 720) .. (128, 720);
        \filldraw[thick,BurntOrange,fill=white](128, 720) circle[radius=8];
        \node at (67, 700) [left] {\tiny$2$};
        \node at (129, 784) [above] {\tiny$1$};
    \end{tikzpicture}
    =
    N_{1,1}^{\textup{Nor}}(6).
\end{align}}

Each such graph is weighted by the inverse of the order of its automorphism group, i.e. we divide by its symmetry factor, giving
\begin{equation}
    \left\langle \frac{1}{6} \NTr{M^6} \right\rangle_{\!\!g=1,\cc}^{\!\!\textup{GUE}}
    =
    N_{1,1}^{\textup{Nor}}(6)
    =
    \frac{1}{6} \cdot 1 + \frac{1}{2} \cdot 1 = \frac{2}{3}.
\end{equation}
This example illustrates how the count of integer Strebel graphs is, by construction, manifestly equal to the original GUE correlator, thereby showing \labelcref{eq:GUE:Nor}.

\subsection{Perturbative discreteness beyond GUE}
So far, our discussion of matrix correlators as discrete volumes of moduli space has been restricted to the Gaussian case. We now wish to understand in what sense this picture continues to hold once interactions are turned on. The punchline will be that the picture of discrete points on $\M_{g,n}$ persists at each order in perturbation theory in the 't~Hooft coupling(s), although these points are no longer necessarily labeled by integer Strebel graphs. Consider the following quartic deformation of the Gaussian model:
\begin{equation}
    Z_{\ms{N}}(t_4)
    =
    \int_{\mathcal{H}_{\ms{N} \times \ms{N}}}
        dM \ e^{-\ms{N} \Tr \left( \frac{1}{2} M^2 - t_{4} \, 2T_{4}(M/2)\right) } ,
\end{equation}
where $t_4$ plays the role of the 't~Hooft coupling and is kept fixed in the large $\ms{N}$ limit. The unusual form of the perturbation stems from the relation between pruned and standard traces in \cref{eq:prune:Cheby}, namely $\frac{1}{4}\NTr{M^4} \; \leftrightarrow \; \Tr 2T_{4}(M/2)$. In \cref{sec:DSSYK:rec}, we will see that this structure persists in the potential of the DSSYK matrix integral, suggesting a geometric origin for the appearance of Chebyshev polynomials first identified in \cite{JKMS23}. 

In the quartic case, perturbation theory in $t_{4}$ expresses the connected matrix correlators in terms of those of the free theory:
\begin{equation}
\begin{split}
     \left\langle \prod_{i=1}^n \frac{1}{b_{i}} \NTr{M^{b_{i}}} \right\rangle_{\!\!g,\cc}^{\!\!\textup{quartic}}
     &=
     \sum_{m=0}^{\infty}
        \frac{t_{4}^{m}}{m!}
        \left\langle
            \prod_{i=1}^n \frac{1}{b_{i}} \NTr{M^{b_{i}}}
            \left( \frac{1}{4} \NTr{M^4} \right)^m
        \right\rangle_{\!\!g,\cc}^{\!\!\textup{GUE}} \\
     &=
     \sum_{m=0}^{\infty} \frac{t_{4}^{m}}{m!} \,
        N^{\textup{Nor}}_{g,n+m}(b_1,\ldots,b_{n},\underbrace{4,\ldots,4}_{m\text{ times}}) \, .
\end{split}
\end{equation}
At the level of Feynman diagrams, $m$ denotes the number of internal vertices, each of valence four. The pruning procedure disallows any petals on these internal vertices as well. This simple perturbative expansion therefore rewrites the interacting correlators as a weighted sum of the $N^{\textup{Nor}}_{g,n+m}$ computed in the GUE. In that sense, all matrix model correlators remain trivially related to the lattice point counts of the moduli space. However, this expression involves a sequence of moduli spaces $\M_{g,n+m}$ and does not yet establish discreteness directly on $\M_{g,n}$. We need something sharper.

\begin{figure}[t]
    \centering
    \tikzset{every picture/.style=thick}
    \begin{tikzpicture}[x=1pt,y=1pt,scale=.7]
        \draw(88, 392) .. controls (120, 413.3333) and (154.6667, 413.3333) .. (192, 392) .. controls (181.3333, 424) and (194.6667, 434.6667) .. (232, 424) .. controls (210.6667, 450.6667) and (213.3333, 472) .. (240, 488) .. controls (208, 477.3333) and (178.6667, 493.3333) .. (152, 536) .. controls (120, 488) and (93.3333, 472) .. (72, 488) .. controls (114.6667, 450.6667) and (120, 418.6667) .. (88, 392) -- cycle;
        \draw(344, 392) .. controls (376, 413.3333) and (410.6667, 413.3333) .. (448, 392) .. controls (437.3333, 424) and (450.6667, 434.6667) .. (488, 424) .. controls (466.6667, 450.6667) and (469.3333, 472) .. (496, 488) .. controls (464, 477.3333) and (434.6667, 493.3333) .. (408, 536) .. controls (376, 488) and (349.3333, 472) .. (328, 488) .. controls (370.6667, 450.6667) and (376, 418.6667) .. (344, 392) -- cycle;
        \node at (226, 384) {$\mathcal{M}_{g,n+m}$};
        \node at (474, 384) {$\mathcal{M}_{g,n}$};
        \node[BurntOrange] at (120, 416) {\scriptsize$\bullet$};
        \node[BurntOrange] at (120, 448) {\scriptsize$\bullet$};
        \node[BurntOrange] at (100, 476) {\scriptsize$\bullet$};
        \node[BurntOrange] at (140, 496) {\scriptsize$\bullet$};
        \node[BurntOrange] at (148, 476) {\scriptsize$\bullet$};
        \node[BurntOrange] at (176, 468) {\scriptsize$\bullet$};
        \node[BurntOrange] at (172, 492) {\scriptsize$\bullet$};
        \node[BurntOrange] at (208, 460) {\scriptsize$\bullet$};
        \node[BurntOrange] at (200, 448) {\scriptsize$\bullet$};
        \node[BurntOrange] at (180, 432) {\scriptsize$\bullet$};
        \node[BurntOrange] at (168, 412) {\scriptsize$\bullet$};
        \node[BurntOrange] at (164, 448) {\scriptsize$\bullet$};
        \node[BurntOrange] at (148, 436) {\scriptsize$\bullet$};
        \draw[->](264, 448) -- (312, 448);
        \node[BurntOrange] at (384, 424) {\scriptsize$\bullet$};
        \node[BurntOrange] at (376, 456) {\scriptsize$\bullet$};
        \node[BurntOrange] at (376, 476) {\scriptsize$\bullet$};
        \node[BurntOrange] at (404, 512) {\scriptsize$\bullet$};
        \node[BurntOrange] at (400, 484) {\scriptsize$\bullet$};
        \node[BurntOrange] at (424, 464) {\scriptsize$\bullet$};
        \node[BurntOrange] at (424, 484) {\scriptsize$\bullet$};
        \node[BurntOrange] at (460, 476) {\scriptsize$\bullet$};
        \node[BurntOrange] at (452, 456) {\scriptsize$\bullet$};
        \node[BurntOrange] at (432, 444) {\scriptsize$\bullet$};
        \node[BurntOrange] at (436, 416) {\scriptsize$\bullet$};
        \node[BurntOrange] at (408, 460) {\scriptsize$\bullet$};
        \node[BurntOrange] at (404, 436) {\scriptsize$\bullet$};
        \node at (288, 464) {forgetful};
        \node at (288, 432) {map};
        
        \node at (50, 458) {\scriptsize{integer}};
        \node at (50, 438) {\scriptsize{Strebel graph}};
        
        \node at (550, 468) {\scriptsize{its image;}};
        \node at (550, 448) {\scriptsize{not necessarily integer}};
        \node at (550, 428) {\scriptsize{Strebel graph}};
        
        \draw[->](44, 476) .. controls (48, 516) and (76, 514) .. (95, 507) .. controls (114, 500) and (124, 488) .. (140, 480);
        \draw[->](528, 486) .. controls (528, 512) and (444, 536) .. (408, 492);
    \end{tikzpicture}
    \caption{
    \textbf{Perturbative Discreteness.} At order $m$ in perturbation theory in the 't~Hooft coupling, the $N_{g,n}$ of the interacting matrix model can be computed from a finite number of Feynman diagrams with $m$ internal vertices. These diagrams map to a discrete set of points on $\M_{g,n+m}$, labeled by integer Strebel graphs (left). Under the forgetful map, these points project to another discrete set on $\M_{g,n}$, whose images are generally not parametrized by integer Strebel points.
    }
    \label{fig:forget}
\end{figure}
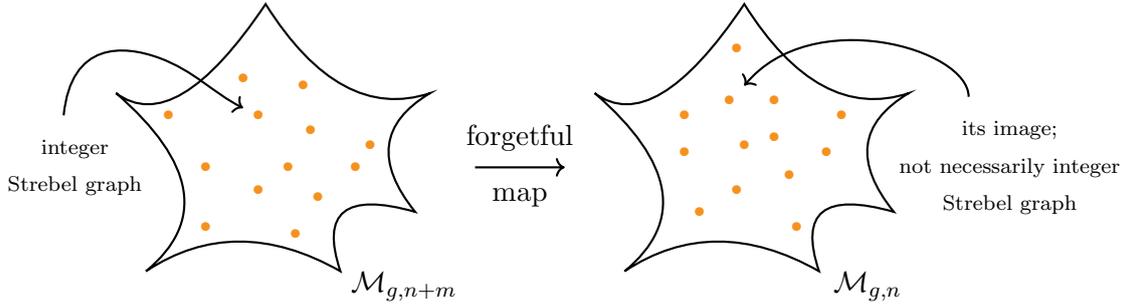

A clue\footnote{This argument was suggested to us by R.~Gopakumar.} comes from what is known in the mathematical literature as the \emph{forgetful map}, $p_m \colon \M_{g,n+m} \to \M_{g,n}$, which describes what happens when one forgets the last $m$ marked points \cite{GL26}. Via our construction in \cref{subsec:GUE:correl}, each Feynman diagram contributing to a term of order $m$ in perturbation theory can be mapped to a point on $\M_{g,n+m}$. This point is labeled by an integer Strebel graph. We can now follow the action of repeatedly applying the forgetful map to the discrete points populating $\M_{g,n+m}$, all the way down to $\M_{g,n}$. We do not yet fully understand how the forgetful map acts on integer lattice points, nor do we have a clear picture at the level of the combinatorial moduli space. However, each integer point on $\M_{g,n+m}$ is mapped to a unique point on $\M_{g,n}$, which generally will not correspond to an integer Strebel graph. Since only finitely many Feynman diagrams contribute at any order in perturbation theory, their image under the forgetful map yields a discrete set of points on $\M_{g,n}$, cf. \cref{fig:forget}.

Although many mathematical details remain to be worked out, this construction offers a compelling picture of a \emph{perturbative discreteness} persisting in interacting matrix models. Even beyond perturbation theory, a trace of this discreteness survives: the parameters $b_i$, which encode the boundary lengths of the dual Riemann surfaces, remain integer-valued, reflecting their origin as matrix powers.

\section{A discrete Mirzakhani recursion for matrix correlators}
\label{sec:discrete:rec}
In this section, we provide an alternative argument for interpreting pruned correlators as discrete volumes of moduli spaces by proving that they satisfy a discrete Mirzakhani recursion (\cref{thm:A}). The dependence on the potential enters only through the specific kernels and the initial data. Our proof is derived from the Eynard--Orantin topological recursion, which computes the standard traces by recasting the Schwinger--Dyson equations.

\subsection{Spectral curve for matrix correlators}
We begin by recalling how the genus-$g$, $n$-point differentials $\omega_{g,n}$ of a large $\ms{N}$ Hermitian matrix model in the one-cut phase (cf.~\cref{eq:MM}) are obtained via the Eynard--Orantin topological recursion on the spectral curve determined by the potential $V$. For a comprehensive reference, see~\cite{Eyn16,EKR15}. In this setting, the genus-zero resolvent satisfies the standard loop equation
\begin{equation}
    W_{0,1}(x)^2 = V'(x) W_{0,1}(x) - P(x),
\end{equation}
where
\begin{equation}\label{eq:resolvent}
    W_{0,1}(x) \coloneqq \frac{1}{\ms{N}}
    \Braket{ \Tr \frac{1}{x - M} }_{\!\!g=0},
    \qquad
    P(x) \coloneqq \frac{1}{\ms{N}} 
    \Braket{ \Tr \frac{V'(x) - V'(M)}{x - M} }_{\!\!g=0}.
\end{equation}
Geometrically, this defines the spectral curve of the matrix model:
\begin{equation}
    y^2 = \frac{1}{4} V'(x)^2 - P(x),
\end{equation}
with $y = \tfrac{1}{2} V'(x) - W_{0,1}(x)$. In the one-cut regime, this curve is a genus-zero Riemann surface with a square-root branch cut, where $y$ can be written as
\begin{equation}
    y = \frac{1}{2} \sqrt{(x-a_+)(x-a_-)} \, Q(x),
\end{equation}
for $Q(x)$ an analytic function near the cut.
The branch points of $y$ are located at $x = a_{\pm}$, corresponding to the endpoints of the eigenvalue support. As mentioned in the previous section, a convenient uniformization is obtained by introducing the Joukowsky variable $z$:
\begin{equation}\label{eq:SC}
    \mathcal{S} =
    \begin{cases}
        x(z) = \gamma \left(z + \dfrac{1}{z}\right) + \delta, \\[2ex]
        y(z) = \dfrac{\gamma}{2} \left( z - \dfrac{1}{z} \right) Q(z),
    \end{cases}
    \qquad
    \gamma = \dfrac{a_+ - a_-}{4} \quad\text{and}\quad \delta = \dfrac{a_+ + a_-}{2}.
\end{equation}
By abuse of notation, we denote $Q(z) \coloneqq Q(x(z))$.
Under the involution $z \mapsto z^{-1}$, the two branches of the square root $\sqrt{(x - a_+)(x - a_-)} = \gamma(z - z^{-1})$ are exchanged, while $Q(z^{-1}) = Q(z)$.
In particular, the two sheets correspond to the exterior and the interior of the unit circle: the exterior is often called the \emph{physical sheet}, as it contains the point $x = \infty$, which is the natural expansion point of the correlators; in contrast, the interior of the unit circle is called the \emph{non-physical sheet}.
The points $z = \pm 1$ map to the branch points $x = a_{\pm}$.

The main result of \cite{Eyn04,CEO06,EO07} states that the correlation differentials \labelcref{eq:corr:TR} are computed by a topological recursion formula involving residues at the ramification points $z = \pm 1$ of $x$. Given the above setup, the Eynard--Orantin topological recursion formula computes $\omega_{g,n}$ recursively via the following residue calculus:
\begin{equation}\label{eq:EO}
\begin{split}
	\omega_{g,n}(z_1,\dots,z_n)
	=
    \Res_{z = \pm 1} & \frac{K(z_1,z)}{\gamma^2 Q(z)}
	\Bigg(
		\omega_{g-1,n+1}(z,z,z_2,\dots,z_n) \\
		& +
		\sum_{\substack{h + h' = g \\ J \sqcup J' = \{2,\dots,n\}}}^{\textup{no }(0,1)}
			\omega_{h,1+|J|}(z,z_J) \,
			\omega_{h',1+|J'|}(z,z_{J'})
	\Bigg),
\end{split}
\end{equation}
where $K(z_1,z)$ is the Eynard--Orantin kernel for the GUE spectral curve:
\begin{equation}\label{eq:kernel}
	K(z_1,z)
	\coloneqq
	\frac{1}{2} \left(
		\frac{1}{z_1 - z} - \frac{1}{z_1 - z^{-1}}
	\right)
	\frac{z^3 dz_1}{(1 - z^2)^2 dz}.
\end{equation}
The superscript ``no $(0,1)$'' indicates that $(h,1+|J|)$ and $(h',1+|J'|)$ never contain terms of disc topology $(0,1)$---though, unlike \cref{eq:disc:rec}, the unstable cylinder amplitude $(g,n) = (0,2)$ is included. The above formula is a recursion in the negative Euler characteristic $2g - 2 + n$, hence the name topological recursion. See \cite{Bou26} for a modern and more detailed account of topological recursion. The above formula follows from the standard Eynard--Orantin topological recursion through the anti-symmetry property under the global involution satisfied by the correlation differentials: $\omega_{g,n}(z,\ldots) = -\omega_{g,n}(1/z,\ldots)$, cf. \cite[equation~(4.2.17)]{EKR15}. For the special case $(g,n) = (1,1)$ one must modify the expression because the term $\omega_{0,2}(z,z)$ is ill-defined (it is singular on the diagonal); we omit these technical adjustments, which are straightforward to reconstruct.

From now onward, we assume\footnote{
    In \cref{sec:DSSYK:rec}, in the context of the DSSYK matrix model, we will consider a case where $Q$ is the Jacobi theta function, and as such it has infinitely many poles; the details will be discussed there.
} the following analytic behavior of $Q$.

\begin{assumption}\label{assumpt}
    From now onward, we assume that $Q$ is a meromorphic function on $\P^1$, with zeros away from the unit circle $|z| = 1$, holomorphic and non-vanishing at the origin, and satisfying the symmetry relations $Q(z^{-1}) = Q(z)$.
\end{assumption}

The special case $\gamma = 1$ and $Q(z) = 1$ reproduces the GUE spectral curve.  
In this sense, $Q$ encodes, at the level of the spectral curve, the effect of the interactions present in the matrix potential.

\subsection{The ABCD of pruned traces}
\label{subsec:ABCD}
We can now complement the statement of \cref{thm:A} under the \cref{assumpt} by providing the explicit expressions for the recursion kernels $B$ and $C$, as well as for the initial data $N_{0,3} \coloneqq A$ and $N_{1,1} \coloneqq D$:
\begin{equation}\label{eq:ABCD}
    \begin{aligned}
    	A(b_1,b_2,b_3)
    	& \coloneqq
        \sigma_+
        +
        (-1)^{b_1+b_2+b_3} \sigma_- , \\
    	B(b,b',\beta)
    	& \coloneqq
    	\frac{1}{2b} \Bigl(
    		H(b + b' - \beta)
    		-
    		H(-b - b' - \beta) \\
    		& \qquad\qquad\qquad
    		+
    		H(b - b' - \beta)
    		-
    		H(-b + b' - \beta)
    	\Bigr) , \\
    	C(b,\beta,\beta')
    	& \coloneqq
    	\frac{1}{b} \Bigl(
    		H(b - \beta - \beta')
    		-
    		H(-b - \beta - \beta')
    	\Bigr) , \\
    	D(b)
    	& \coloneqq
        \frac{b^2 - 4}{48} \bigl( \sigma_+ + (-1)^b \sigma_- \bigr)
        +
        \frac{1}{16} \bigl( \tau_+ + (-1)^b \tau_- \bigr).
    \end{aligned}
\end{equation}
All quantities involved are expressed in terms of the matrix model spectral curve \labelcref{eq:SC} as
\begin{equation}\label{eq:cnstnts}
    \sigma_{\pm} \coloneqq \frac{1}{2 \gamma^2 Q(\pm 1)},
    \qquad
    \tau_{\pm} \coloneqq \frac{d^2}{dz^2} \Bigl( \frac{1}{2 \gamma^2 Q(z)} \Bigr)\Big|_{z=\pm 1},
\end{equation}
while $H \colon \Z \to \C$ is defined by a contour integral in a counter-clockwise direction encircling $z = 0$ and all zeros of $Q$ lying inside the unit circle, while avoiding $z = \pm 1$:
\begin{equation}\label{eq:building:block}
	H(\ell)
	\coloneqq
	\frac{1}{2\pi\iu} \oint_{\Gamma} \frac{2 z^{1 - \ell}}{(1-z^2)^2 \gamma^2 Q(z)} dz.
    \qquad\qquad
    \begin{tikzpicture}[baseline={(0, 0cm)}]
        \draw[thick,decoration={markings,mark=at position 0.33 with {\arrow{>}}},postaction={decorate}] (1,0) arc (0:180:1cm);
        \draw[thick,decoration={markings,mark=at position 0.33 with {\arrow{>}}},postaction={decorate}] (-1,0) arc (-180:0:1cm);
        \fill[white] (1,0) circle (.2cm);
        \fill[white] (-1,0) circle (.2cm);
        \draw[->] (-1.5,0) -- (1.5,0);
        \draw[->] (0,-1.3) -- (0,1.3);
        \node at (1,0) {\tiny$\bullet$};
        \node at (1,0) [below right] {\tiny$+1$};
        \node at (-1,0) {\tiny$\bullet$};
        \node at (-1,0) [below left] {\tiny$-1$};
        \node at (1,1) {$\Gamma$};
        \clip (0,0) circle (1cm);
        \draw[thick] (1,0) circle (.2cm);
        \draw[thick] (-1,0) circle (.2cm);
    \end{tikzpicture}
\end{equation}
Notice that the integrand is essentially the reciprocal disk amplitude,
\begin{equation}\label{eq:W01}
    \omega_{0,1}(z) = y(z) dx(z) = \frac{(1-z^2)^2}{2z^3} \gamma^2 Q(z) dz,
\end{equation}
while the contour of integration is a contour around the cut in the non-physical sheet. In terms of the resolvent $W_{0,1}$, the function $H$ reads as in \cref{eq:H}.

The ABCD terminology was first introduced in \cite{ABCO24}. It originates from the reformulation of topological recursion by Kontsevich--Soibelman~\cite{KS18} in terms of quantum Airy structures, a generalization of Virasoro constraints.

The remaining part of this section is devoted to the proof of \cref{thm:A}. Before proceeding, two comments are in order.
\begin{enumerate}
    \item A practical remark. The expressions are linear in $1/Q$: if $1/Q = \sum_{k} 1/Q_k$, the contribution of each $Q_k$ can be computed separately and then summed to obtain the final result. This provides a powerful computational tool: for a given $Q$, the strategy is to expand it into partial fractions and compute the contribution to $H$ from each individual term. This approach is illustrated in \cref{app:ex}, where we compute $H$ for the partial fraction components appearing in the DSSYK model.

    \item The contour integral is simply the sum of the residues within its interior, namely at $z = 0$ and at the zeros of $Q$ lying inside the unit circle. This leads to the natural decomposition
    \begin{equation}\label{eq:splitting}
        H(\ell) = F(\ell) + G(\ell),
    \end{equation}
    where
    \begin{equation}\label{eq:F:G}
        F(\ell)
        \coloneqq
        \Res_{z = 0} \frac{2z^{1 - \ell}}{(1-z^2)^2 \gamma^2 Q(z)}\,dz,
        \quad
        G(\ell)
        \coloneqq
        \sum_{|\alpha|<1}
            \Res_{z = \alpha} \frac{2 z^{1 - \ell}}{(1-z^2)^2 \gamma^2 Q(z)}\,dz.
    \end{equation}
    This decomposition is precisely how the function $H$ arises in the proof of \cref{thm:A}.
\end{enumerate}

\subsection{Proof of the discrete recursion}
\label{ssec:disc:rec:proof}
To establish \cref{thm:A}, recall that the pruned correlators $N_{g,n}$ are defined in terms of the genus-$g$, $n$-point function:
\begin{equation}
    \omega_{g,n}(z_1,\dots,z_n)
    =
    \sum_{b_1,\dots,b_n >0}
        N_{g,n}(b_1,\dots,b_n) \prod_{i=1}^n b_i z_i^{b_i-1} dz_i.
\end{equation}
We derive our discrete recursion formula for $N_{g,n}$, \cref{eq:disc:rec}, from the Eynard--Orantin recursion for $\omega_{g,n}$, \cref{eq:EO}, in four main steps:
\begin{enumerate}[label=\Roman*)]
	\item Separate the contributions that contain the cylinder amplitude (the $B$-terms) from those that do not (the $C$-terms).
	
	\item Move the contour from $z=\pm 1$ to the other poles of the integrand, namely the points $z=z_i^{\pm 1}$ and the zeros of $Q$. This is possible because the spectral curve is the Riemann sphere $\mathbb{P}^1$ in our one-cut uniformization.
	
	\item Compute the residues at $z=z_i^{\pm 1}$, which produce the $F$-contributions, and those at the zeros of $Q$, which produce the $G$-contributions. Altogether this recovers the function $H$ appearing in \cref{eq:splitting}.
	
	\item Compute the initial data corresponding to the pair of pants and the one-holed torus, namely $A \coloneqq N_{0,3}$ and $D \coloneqq N_{1,1}$.
\end{enumerate}
We now analyze each of these steps in more detail, relegating the more technical computations to \cref{app:proofsRecRel}. The proof follows Norbury's computations for GUE \cite{Nor13}, although the presence of the interaction term renders several steps considerably more involved.

{\bfseries I) The $B$- and $C$-terms.} In the sum over the splittings of the genus and the boundary components, we factor out the terms containing the cylinder amplitudes. As a result, the right-hand side of the residue formula \labelcref{eq:EO} naturally splits into two types of contributions: we refer to them as the $B_m$-terms (for $m = 2,\ldots,n$) and the $C$-term, defined by
\begin{equation}
\begin{aligned}
	\omega_{B_m}(z)
	&\coloneqq
	\omega_{g,n-1}(z,z_2,\ldots,\widehat{z_m},\ldots,z_n) , \\
	\omega_{C}(z,z)
	&\coloneqq
	\omega_{g-1,n+1}(z,z,z_2,\dots,z_n)
	+
	\sum_{\substack{ h + h' = g \\ J \sqcup J' = \{2,\dots,n\} }}^{\textup{stable}}
		\omega_{h,1+|J|}(z,z_{J}) \omega_{h',1+|J'|}(z,z_{J'}),
\end{aligned}
\end{equation}
respectively. Notice that the sum over the splittings of the genus and the boundary components now runs only over \emph{stable} topologies, i.e. both the disk and cylinder amplitudes are excluded. We omit the dependence on the remaining variables, as they act as spectators.
Most of the subsequent computations will treat these two terms separately. With this notation, \cref{eq:EO} is written as
\begin{multline}
	\omega_{g,n}(z_1,\dots,z_n)
	=
	\sum_{m = 2}^n
	\Res_{z = \pm 1}
		\frac{K(z_1,z) \bigl( \frac{1}{(z-z_m)^2} + \frac{1}{(1-z z_m)^2} \bigr)}{\gamma^2 Q(z)}
		\omega_{B_m}(z)
		\\
	+
	\Res_{z = \pm 1}
		\frac{K(z_1,z)}{\gamma^2 Q(z)}
		\omega_{C}(z,z) .
\end{multline}

{\bfseries II) Moving the contour.} Next, we move the contour from around $z = \pm 1$ to encircle all other poles of the integrand, using the residue theorem.
Recall that the only poles of the correlation differentials are located at the ramification points, i.e. $z = \pm 1$.

For the $B_m$-term, the other poles are located at $z = z_1^{\pm 1}$ (due to the presence of the kernel $K$), at $z = z_m^{\pm 1}$ (from the factors originating from $\omega_{0,2}$), and at the zeros of $Q$ (from $1/Q$).
Similarly, for the $C$-term the other poles are located at $z = z_1^{\pm 1}$ and at the zeros of $Q$. This gives
\begin{multline}
	\omega_{g,n}(z_1,\dots,z_n)
	= \\
	- \sum_{m = 2}^n
	\bigg(
		\Res_{z = z_1^{\pm 1}} + \Res_{z = z_m^{\pm 1}} + \sum_{\alpha} \Res_{z = \alpha}
	\bigg)
		\frac{K(z_1,z) \bigl( \frac{1}{(z-z_m)^2} + \frac{1}{(1-z z_m)^2} \bigr)}{\gamma^2 Q(z)} 
		\omega_{B_m}(z)
		\\
	-
	\bigg(
		\Res_{z = z_1^{\pm 1}} + \sum_{\alpha} \Res_{z = \alpha}
	\bigg)
	   \frac{K(z_1,z)}{\gamma^2 Q(z)} \omega_{C}(z,z),
\end{multline}
where $\alpha$ runs over \emph{all} zeros of $Q$, and the overall minus sign reflects the opposite orientation of the original contour when it is deformed to encircle other poles.

{\bfseries III) Computing the residues.} Next, we handle separately the residues at $z_i^{\pm 1}$ and those at the zeros of $Q$. This splitting gives rise to the decomposition of the building-block function $H$ into the $F$-term and the $G$-term, respectively, in \cref{eq:splitting}. In both cases, we must separately consider the $B$-terms and the $C$-term.

{\bfseries III.1) Residues at $z_i^{\pm 1}$ as $F$-contributions.} For the $B_m$-terms, a direct computation shows that the residues at $z = z_1^{\pm 1}$ contribute equally as
\begin{equation}
\begin{split}
	\Res_{z = z_1^{\pm 1}} (B_m\text{-term})
	& =
	-
	\left(
		\frac{1}{(z_1 - z_m)^2} + \frac{1}{(1 - z_1z_m)^2}
	\right)
	\frac{z_1^3}{(1-z_1^2)^2 \gamma^2 Q(z_1)}
	\omega_{B_m}(z_1) \, dz_m \\
	& =
	-
	\de_{z_m} \left[
		\left(
			\frac{1}{z_1 - z_m} - \frac{1}{z_1 - z_m^{-1}}
		\right)
		\frac{z_1^3}{(1-z_1^2)^2 \gamma^2 Q(z_1) dz_1}
		\omega_{B_m}(z_1)
	\right] dz_1 dz_m.
\end{split}
\end{equation}
The rewriting as a total $z_m$-derivative is only for later convenience. On the other hand, due to the presence of a double pole, the residues at $z = z_m^{\pm 1}$ directly evaluate as a total $z_m$-derivative:
\begin{equation}
	\Res_{z = z_m^{\pm 1}} (B_m\text{-term})
	=
	\de_{z_m} \left[
		\left(
			  \frac{1}{z_1 - z_m} - \frac{1}{z_1 - z_m^{-1}}
		\right)
		\frac{z_m^3}{(1 - z_m^2)^2 \gamma^2 Q(z_m) dz_m}
		\omega_{B_m}(z_m)
	\right] dz_1 dz_m.
\end{equation}
As for the $C$-term, a direct computation shows that the residues at $z = z_1^{\pm}$ give
\begin{equation}
	\Res_{z = z_1^{\pm 1}}
		(C\text{-term})
	=
	- \frac{z_1^3}{(1-z_1^2)^2 \gamma^2 Q(z_1) dz_1} \omega_C(z_1,z_1) .
\end{equation}
In both cases, the relevant residues involve the correlation differentials divided by $Q$. Since we are interested in the discrete Laplace transform of such expressions, it is natural to expect that this transform is obtained as the convolution of the kernels corresponding to $Q = 1$ (i.e. GUE) with the discrete Laplace transform of $1/(\gamma^2 Q)$. This is precisely the content of \cref{lem:GUE,lem:conv} in the appendix. The final result reads
\begin{multline}
	\bigg(
		\Res_{z = z_1^{\pm 1}}
		+
		\Res_{z = z_m^{\pm 1}}
	\bigg) (B_m\text{-term})
	=
	-
	\sum_{b_1,\ldots,b_n > 0}
		\Bigg[
			\sum_{\beta > 0}
				\beta 
				\frac{1}{2b_1} \Bigl(
					F(b_1 + b_m - \beta)
					-
					F(-b_1 - b_m - \beta) \\
					+
					F(b_1 - b_m - \beta)
					-
					F(-b_1 + b_m - \beta)
				\Bigr) N_{B_m}(\beta)
		\Bigg] \prod_{i=1}^n b_i z_i^{b_i -1} dz_i
\end{multline}
for the $B_m$-term, and
\begin{multline}
	\Res_{z = z_1^{\pm 1}} (C\text{-term})
	=
	-
	\frac{1}{2}
	\sum_{b_1,\ldots,b_n > 0}
		\Bigg[
			\sum_{\beta, \beta' > 0}
				\beta \beta'
				\frac{1}{b_1} \Bigl(
					F(b_1 - \beta - \beta') \\
					-
					F(-b_1 - \beta - \beta')
				\Bigr) N_C(\beta,\beta')
		\Bigg] \prod_{i=1}^n b_i z_i^{b_i -1} dz_i
\end{multline}
for the $C$-term, where $F$ is given in \cref{eq:F:G}. It is essentially the discrete convolution of the GUE building block, the ramp function, and the Taylor coefficients of $1/(\gamma^2 Q)$, as explained in the appendix. This defines the first term, $F$, appearing in the building-block function $H$ from \cref{eq:splitting}. Here $N_{B_m}$ and $N_C$ are the discrete Laplace transforms of $\omega_{B_m}$ and $\omega_C$, cf. \cref{eq:Laplace:WB:WC}.

{\bfseries III.2) Residues at the zeros of $Q$ as $G$-contributions.} First, notice that in both the $B$- and $C$-cases, the residues at $\alpha$ contribute equally to those at $\alpha^{-1}$. In formulae, $\Res_{z = \alpha^{\pm 1}} = 2 \Res_{z = \alpha}$. Thus, we can restrict our attention to the zeros of $Q$ lying inside the unit circle. This computation is carried out in \cref{lem:zeros} and reads
\begin{multline}
	2 \sum_{|\alpha| < 1} \Res_{z = \alpha} \ (B_m\text{-term})
	=
	-
	\sum_{b_1,\ldots,b_n > 0}
		\Bigg[
			\sum_{\beta > 0}
				\beta 
				\frac{1}{2b_1} \Bigl(
					G(b_1 + b_m - \beta)
					-
					G(-b_1 - b_m - \beta) \\
					+
					G(b_1 - b_m - \beta)
					-
					G(-b_1 + b_m - \beta)
				\Bigr) N_{B_m}(\beta)
		\Bigg] \prod_{i=1}^n b_i z_i^{b_i -1} dz_i
\end{multline}
for the $B_m$ term, and
\begin{multline}
	2 \sum_{|\alpha| < 1} \Res_{z = \alpha} \ (C\text{-term})
	=
	-
	\frac{1}{2}
	\sum_{b_1,\ldots,b_n > 0}
		\Bigg[
			\sum_{\beta, \beta' > 0}
				\beta \beta'
				\frac{1}{b_1} \Big(
					G(b_1 - \beta - \beta') \\
					-
					G(-b_1 - \beta - \beta')
				\Big) N_C(\beta,\beta')
		\Bigg] \prod_{i=1}^n b_i z_i^{b_i -1} dz_i
\end{multline}
for the $C$-term, where $G$ is given by a residue over the zeros of $Q$ inside the unit circle as in \cref{eq:F:G}. This defines the second term, $G$, appearing in the building-block function $H$ from \cref{eq:splitting}. Altogether, this yields the desired recursive formula \labelcref{eq:disc:rec} from \cref{thm:A}.

{\bfseries IV) The initial conditions.} To complete the proof, it remains only to compute the initial data. This is obtained by a straightforward direct calculation, which we omit here.

This completes the proof of \cref{thm:A}.

\section{The BMN-like limit}
\label{sec:BMN}

In this section, we study a universal subsector of one-cut matrix models obtained by sending the powers $b_i$ of the matrices appearing in the pruned traces uniformly to infinity (see \cref{thm:B} in the introduction). In this regime, analogous to the BMN limit\footnote{
    Strictly speaking, the BMN limit considers powers of the matrices that scale with $\ms{N}$. In some sense, we are studying a simpler limit, where we first take $\ms{N} \rightarrow \infty$, and then take the powers of the matrices to be large, at each order in the $1/\ms{N}$ expansion.
} in AdS/CFT \cite{BMN02}, the pruned correlators converge to the Kontsevich volumes that govern one of the fundamental building blocks of intersection theory on the moduli space of Riemann surfaces.

To rigorously identify the limit, a word of caution: the pruned correlators depend on the matrix powers not only through polynomial growth but also through their parity. Thus, the naive large $b_i$ limit need not exist. Once one separates even and odd sectors, however, each sector does admit a limit: the pruned correlators converge (up to a scalar factor that depends on the parity) to the Kontsevich volumes, as stated in \cref{thm:B} with scaling constants
\begin{equation}
    \sigma_{\pm} \coloneqq \frac{1}{2 \gamma^2 Q(\pm 1)}
\end{equation}
as in \cref{eq:cnstnts}.

Our proof relies on the discrete recursion from \cref{thm:A}, which, in the limit, converges in a Riemann-sum-to-integral fashion to the continuous recursion satisfied by the Kontsevich volumes, \cref{thm:Kon}. Up to scaling constants, the building-block function $H(\ell)$ universally asymptotes to the ramp function $\rho(\ell) \coloneqq \ell \theta(\ell)$, which serves as the building-block function for the Kontsevich volumes.

This limit admits an equivalent interpretation as the familiar \emph{edge of the spectrum} (or Airy) zoom in random matrix theory \cite{TW94,Eyn16,BrezinHikamiAiry,BrezinZee}, see \cref{fig:edge}. Near the spectral endpoints, the local behavior of any one-cut model is universally governed by the Airy curve, whose topological recursion computes the Kontsevich volumes. This heuristic also explains the two scaling factors as the contributions from the two edges of the spectrum. Although making this correspondence entirely rigorous beyond genus zero is delicate, our proof proceeds directly from the discrete recursion: we show that, term by term, it converges to the Kontsevich recursion, and this convergence propagates by induction on the Euler characteristic.

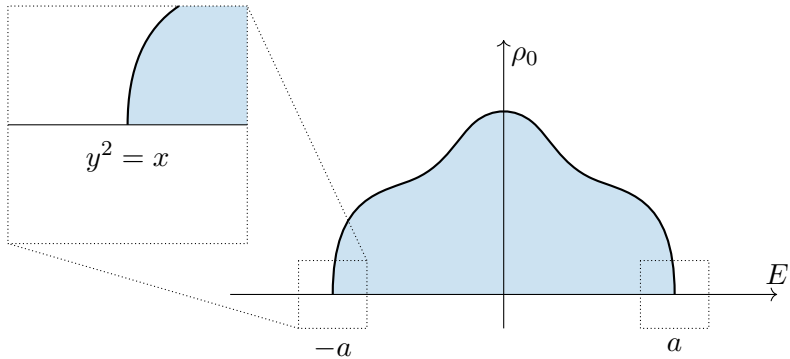
\begin{figure}
    \centering
    \begin{tikzpicture}[x=1pt,y=1pt,scale=.8]
        \fill[RoyalBlue,opacity=.2] (280, 640) -- (440, 640) .. controls (440, 688) and (416, 688) .. (401.3333, 694.6667) .. controls (386.6667, 701.3333) and (381.3333, 714.6667) .. (373.3333, 721.3333) .. controls (365.3333, 728) and (354.6667, 728) .. (346.6667, 721.3333) .. controls (338.6667, 714.6667) and (333.3333, 701.3333) .. (318.6667, 694.6667) .. controls (304, 688) and (280, 688) .. (280, 640) -- cycle;
        \draw[thick] (280, 640) .. controls (280, 688) and (304, 688) .. (318.6667, 694.6667) .. controls (333.3333, 701.3333) and (338.6667, 714.6667) .. (346.6667, 721.3333) .. controls (354.6667, 728) and (365.3333, 728) .. (373.3333, 721.3333) .. controls (381.3333, 714.6667) and (386.6667, 701.3333) .. (401.3333, 694.6667) .. controls (416, 688) and (440, 688) .. (440, 640);
        \draw[->] (360, 624) -- (360, 760);
        \draw[->] (232, 640) -- (488, 640);
        \draw[densely dotted] (424, 656) rectangle (456, 624);
        \draw[densely dotted] (264, 656) rectangle (296, 624);
        \fill[RoyalBlue,opacity=.2] (184, 720) -- (240, 720) -- (240, 776) -- (208, 776) .. controls (192, 765.3333) and (184, 746.6667) .. (184, 720) -- cycle;
        \draw (128, 720) -- (240, 720);
        \draw[thick] (184, 720) .. controls (184, 746.6667) and (192, 765.3333) .. (208, 776);
        \draw[densely dotted] (296, 656) -- (240, 776);
        \draw[densely dotted] (264, 624) -- (128, 664);
        \draw[densely dotted] (128, 776) rectangle (240, 664);
        \node at (280, 624) [below] {$a_-$};
        \node at (440, 624) [below] {$a_+$};
        \node at (370, 752) {$\varrho_0$};
        \node at (488, 640) [above] {$E$};
        \node at (184, 705) {$y^2 = x$};
    \end{tikzpicture}
    \caption{\textbf{Edge of the Spectrum:} Heuristically, in the BMN-like limit, correlators probe eigenvalues close the edge of the spectrum, governed by the Airy universality class $y^2 = x$.}
    \label{fig:edge}
\end{figure}

A more diagrammatic intuition can also be given. In the large $b_i$ regime, most Wick contractions contributing to a pruned correlator occur between edges attached to external vertices, rather than through internal ones. As the external valences grow, the dominant combinatorial patterns are those where external legs contract among themselves, effectively filling the diagram and washing out the detailed structure of the potential. This explains the universality of the Airy limit: the microscopic details of the interaction potential become irrelevant. Making this argument fully precise is challenging, since it involves summing over arbitrarily many internal vertices; moreover, the potential itself generates vertices of unbounded valency. Nevertheless, our recursion-based approach provides a clean derivation of this universal limit, bypassing these combinatorial complications.

We conclude by noting that, for even potentials, the two scaling factors coincide: $\sigma_+ = \sigma_-$. Thus, we find twice the same contribution when $t^{-1}(L_1 + \cdots + L_n)$ is even, and zero in the odd case. This agrees with the fact that, for even potentials, the correlators vanish unless the powers of the matrices sum to an even integer.

We now proceed with the rigorous proof of \cref{thm:B} by induction on $2g-2+n$.

\paragraph{Base case.} The induction step is easily deduced from \cref{eq:ABCD}: the $(0,3)$ case is straightforward, while the $(1,1)$ case reads
\begin{equation}
\begin{split}
    t^{2} \,
    N_{1,1}\Bigl( \frac{L_1}{t} \Bigr)
    &=
    \frac{L_1^2 - 4t^2}{48} \bigl( \sigma_+ + (-1)^{\frac{L_1}{t}} \sigma_- \bigr)
        +
        \frac{t^2}{16} \bigl( \tau_+ - (-1)^{\frac{L_1}{t}} \tau_- \bigr)
        \\
    &\sim
    \frac{L_1^2}{48} \bigl( \sigma_+ + (-1)^{\frac{L_1}{t}} \sigma_- \bigr).
\end{split}
\end{equation}
Here and throughout this section, we use $\sim$ to denote asymptotic equivalence: for two functions $f$ and $g$, we write $f(t) \sim g(t)$ if and only if $f(t)/g(t) \to 1$ as $t \to 0^+$.

\paragraph{Induction step.} 
Now suppose, by induction, that \cref{thm:B} holds for all $2g' - 2 + n' < 2g - 2 + n$. 
For ease of notation, set 
$N_{g,n}^{t}(L_1,\ldots,L_n) \coloneqq t^{2(3g-3+n)} N_{g,n}(L_1/t,\ldots,L_n/t)$ 
for the rescaled discrete volumes, and
\begin{equation}
    B^{t}(L_1,L_m,\ell)
    \coloneqq
    B\Bigl( \frac{L_1}{t},\frac{L_m}{t},\frac{\ell}{t} \Bigr),
    \qquad
    C^{t}(L_1,\ell,\ell')
    \coloneqq
    C\Bigl( \frac{L_1}{t},\frac{\ell}{t},\frac{\ell'}{t} \Bigr),
\end{equation}
for the rescaled kernels. The recursion for the rescaled discrete volumes then reads
\begin{equation}\label{eq:rescaled:rec}
\begin{split}
    N_{g,n}^{t}(L_1,\ldots,L_n)
    &=
    t
    \sum_{m=2}^n \sum_{\ell \in t\Z_+}
        \ell \, B^{t}(L_1,L_m,\ell)
        N_{g,n-1}^{t}(\ell,L_2,\dots, \widehat{L_m},\dots, L_n) \\
    &\qquad
    + \frac{t^2}{2} \sum_{\ell,\ell' \in t\Z_+}
        \ell \ell' \, C^{t}(L_1,\ell,\ell') \Bigg(
            N_{g-1,n+1}^{t}(\ell,\ell',L_2,\dots,L_n) \\
    &\qquad\qquad\qquad
            +
            \sum_{\substack{ h + h' = g \\ J \sqcup J' = \{2,\dots,n\} }}^{\textup{stable}}
                N_{h,1+|J|}^{t}(\ell,L_{J})
                N_{h',1+|J'|}^{t}(\ell',L_{J'})
    \Bigg).
\end{split}
\end{equation}
Notice that the internal sums over $\ell$ and $\ell'$ have been rescaled as well, hence the prefactors $t$ and $t^{2}$ multiplying the $B$- and $C$-terms, respectively, and the sums running over the rescaled positive integers $t \Z_+$.

Since the rescaled discrete volumes on the right-hand side asymptote to known quantities by the induction hypothesis, we only need to analyze the limit of the recursion kernels. 
As these are constructed from the building-block function $H$, we turn our attention to the latter. 
We claim that
\begin{equation}\label{eq:H:limit}
    \lim_{\substack{t \to 0^+ \\ t^{-1}\ell \\ \textup{even/odd}}} t \, H\Bigl( \frac{\ell}{t} \Bigr)
    =
    (\sigma_+ \pm \sigma_-) \rho(\ell),
\end{equation}
where, as before, $\rho(\ell) = \ell \theta(\ell)$ denotes the ramp function. 
To prove the claim, we analyze separately the cases where $\ell$ is negative or positive. 
Before proceeding, notice that in the definition of $H$ (\cref{eq:building:block}), we can deform the contour $\Gamma$ to a circle $|z| = r$ with $r < 1$, chosen large enough to contain all zeros of $Q$ lying inside the unit circle.

For $\ell < 0$, by changing variables $z = 1/w$ in the contour integral, we find
\begin{equation}
    t H\Bigl( \frac{\ell}{t} \Bigr)
    =
    - \frac{t}{2\pi \iu} \oint_{|w| = 1/r} \frac{2w^{\frac{\ell}{t}+1}}{(1-w^2)^2 \gamma^2 Q(w)} \, dw .
\end{equation}
The $\ell$-independent part of the integrand is bounded along the contour by a constant $C > 0$, independent of $\ell$, so that $|t H(\ell/t)| \le t C / r^{\ell/t}$, which tends to zero as $t \to 0^+$ since $\ell$ is negative and $r < 1$.

As for the case $\ell \ge 0$, we can rewrite the contour integral defining $H$ as
\begin{equation}
    t H\Bigl( \frac{\ell}{t} \Bigr)
    =
    -
    t \Res_{z = \pm 1} \frac{2z^{1-\frac{\ell}{t}}}{(1-z^2)^2 \gamma^2 Q(z)} \, dz
    -
    \frac{t}{2\pi \iu} \oint_{|z| = 1/r} \frac{2z^{1-\frac{\ell}{t}}}{(1-z^2)^2 \gamma^2 Q(z)} \, dz .
\end{equation}
Here we first flipped the orientation of the contour $|z| = r$, introducing the minus sign, then deformed it to $|z| = 1/r$ and picked up the residues at $z = \pm 1$ along the way. By the same argument as above, the contour integral tends to zero as $t \to 0^+$. Hence, we are left with evaluating the residues at $z = \pm 1$. Since these are double poles, a direct computation yields
\begin{equation}
    t \Res_{z = \pm 1} \frac{2z^{1-\frac{\ell}{t}}}{(1-z^2)^2 \gamma^2 Q(z)} \, dz
    =
    \frac{\ell}{2\gamma^2} \left( \frac{1}{Q(+1)} + (-1)^{\ell/t} \frac{1}{Q(-1)} \right)
    +
    \mathrm{O}(t).
\end{equation}
This proves the claimed limits \labelcref{eq:H:limit}.

An immediate consequence is then the behavior of the rescaled kernels in the limit:
\begin{equation}
\begin{aligned}
    \lim_{\substack{t \to 0^+ \\ t^{-1}(L_1+L_m+\ell) \\ \textup{even/odd}}} & B^{t}(L_1,L_m,\ell)
    =
    (\sigma_+ \pm \sigma_-) \, B^{\textup{comb}}(L_1,L_m,\ell), \\
    \lim_{\substack{t \to 0^+ \\ t^{-1}(L_1+\ell+\ell') \\ \textup{even/odd}}} & C^{t}(L_1,\ell,\ell')
    =
    (\sigma_+ \pm \sigma_-) \, C^{\textup{comb}}(L_1,\ell,\ell') .
\end{aligned}
\end{equation}
We are now ready to complete the proof by induction. We only consider the even case in \cref{thm:B}; the odd case is entirely analogous. Suppose then that $L_1 + \cdots + L_n$ is an even multiple of $t$. Each sum over ``internal boundary lengths'' in \cref{eq:rescaled:rec} now splits into two or four possible sums as follows:
\begin{itemize}
    \item In the $B$-term, split the sum according to whether $t^{-1}(\ell + L_2 + \cdots + \widehat{L_m} + \cdots + L_n)$ is even or odd. Since, by hypothesis, $t^{-1}(L_1 + \cdots + L_n)$ is even, we obtain that $t^{-1}(L_1 + L_m + \ell)$ is even, respectively odd.
    
    \item In the connected $C$-term, split the sum according to whether $t^{-1}(\ell + \ell' + L_2 + \cdots + L_n)$ is even or odd. We then obtain that $t^{-1}(L_1 + \ell + \ell')$ must be even, respectively odd.
    
    \item In the disconnected $C$-term, fix a splitting $J \sqcup J' = \{2,\ldots,n\}$ of the boundary components, and split the sum into four, according to the parity of $t^{-1}(\ell + \sum_{j \in J} L_j)$ and $t^{-1}(\ell' + \sum_{j' \in J'} L_{j'})$. Then the parity of $t^{-1}(L_1 + \ell + \ell')$ is the sum of the parities of $t^{-1}(\ell + L_J)$ and $t^{-1}(\ell' + L_{J'})$ modulo two.
\end{itemize}
Now consider the $B_m$-term in \cref{eq:rescaled:rec}. Splitting the sum as described above and applying the induction hypothesis, we find
\begin{equation}
\begin{split}
    t
    \sum_{\ell \in t\Z_+} &
        \ell \, B^{t}(L_1,L_m,\ell)
        N_{g,n-1}^{t}(\ell,L_2,\dots, \widehat{L_m},\dots, L_n)
         \\
    &=
    t
    \sum_{\substack{\ell \in t\Z_+ \\ t^{-1}(L_1 + L_m + \ell) \\ \textup{even}}}
        \ell \, B^{t}(L_1,L_m,\ell)
        N_{g,n-1}^{t}(\ell,L_2,\dots, \widehat{L_m},\dots, L_n) \\
    & \quad
    +
    t
    \sum_{\substack{\ell \in t\Z_+ \\ t^{-1}(L_1 + L_m + \ell) \\ \textup{odd}}}
        \ell \, B^{t}(L_1,L_m,\ell)
        N_{g,n-1}^{t}(\ell,L_2,\dots, \widehat{L_m},\dots, L_n) .
\end{split}
\end{equation}
The summands in the last two lines tend to a multiple of the corresponding term from the continuum recursion for Kontsevich's volumes, i.e.\ $B^{\textup{comb}} \, V_{g,n-1}^{\textup{Kon}}$, with multiplicative constants $(\sigma_+ \pm \sigma_-) \,(\sigma_+^{2g-2+n-1} \pm \sigma_-^{2g-2+n-1})$ in the even/odd case, respectively. As for the sums, note that both run over only \emph{half} of the rescaled lattice due to the even/odd conditions. Consequently, the Riemann sums converge to half of the corresponding integrals:
\begin{equation}
    t \ \ \sum_{\mathclap{\substack{\ell \in t\Z_+ \\ t^{-1}(L_1+L_m+\ell) \\ \textup{even/odd}}}} \ \ (\cdot)
    \quad\longrightarrow\quad
    \frac{1}{2}
    \int_{0}^{+\infty}
        d\ell \ (\cdot)
    \qquad
    \text{as } t \to 0^+.
\end{equation}
The $B_m$-term therefore tends to the desired quantity, as
\begin{multline}
    \frac{1}{2}
    \left[
        (\sigma_+ + \sigma_-)
        (\sigma_+^{2g-2+n-1} + \sigma_-^{2g-2+n-1})
        +
        (\sigma_+ - \sigma_-)
        (\sigma_+^{2g-2+n-1} - \sigma_-^{2g-2+n-1})
    \right] \\
    =
    \sigma_+^{2g-2+n} + \sigma_-^{2g-2+n} .
\end{multline}
The exact same argument applies to the connected $C$-term. As for the disconnected one, four splittings are involved. However, in that case the sums are over $\ell + \sum_{j \in J} L_j$ and $\ell' + \sum_{j' \in J'} L_{j'}$ with definite parities. Thus, both sums over $\ell$ and $\ell'$ run over only half of the rescaled lattice, yielding an overall factor of a \emph{quarter}. The remaining algebraic manipulation is rather tedious but straightforward, and we omit it.

Altogether, we obtain the desired limit: in the even case,
\begin{equation}
\begin{split}
    N_{g,n}^{t}&(L_1,\ldots,L_n)
    \sim \\
    &
    (\sigma_+^{2g-2+n} + \sigma_-^{2g-2+n})
    \Bigg[
    \int_{0}^{+\infty}
        d\ell \,
          \ell \, B^{\textup{comb}}(L_1,L_m,\ell)
          V_{g,n-1}^{\textup{Kon}}(\ell,L_2,\dots, \widehat{L_m},\dots, L_n) \\
    & \qquad\qquad\quad
    + \frac{1}{2} \int_{0}^{+\infty} \int_{0}^{+\infty}
        d\ell \, d\ell' \,
        \ell \ell' \, C^{\textup{comb}}(L_1,\ell,\ell') \Bigg(
            V_{g-1,n+1}^{\textup{Kon}}(\ell,\ell',L_2,\dots,L_n) \\
    & \qquad\qquad\qquad\qquad\qquad\qquad\quad
            +
            \sum_{\substack{ h + h' = g \\ J \sqcup J' = \{2,\dots,n\} }}^{\textup{stable}}
                V_{h,1+|J|}^{\textup{Kon}}(\ell,L_{J})
                V_{h',1+|J'|}^{\textup{Kon}}(\ell',L_{J'})
    \Bigg)
    \Bigg],
\end{split}
\end{equation}
which in turn equals $(\sigma_+^{2g-2+n} + \sigma_-^{2g-2+n})V_{g,n}^{\textup{Kon}}(L_1,\ldots,L_n)$, multiple of the Kontsevich volume, thanks to \cref{thm:Kon}.

\section{Discrete \texorpdfstring{$q$}{q}-Weil--Petersson volumes from the DSSYK matrix integral}
\label{sec:DSSYK:rec}

The Sachdev--Ye--Kitaev (SYK) model \cite{Sachdev:1992fk, Kitaev_talk, MaldacenaStanford} is a quantum mechanical system in $(0+1)$ dimensions consisting of $\ms{M}$ Majorana fermions with all-to-all $p$-body interactions. Its dynamics is governed by the Hamiltonian
\begin{equation}
    H
    \coloneqq
    \mathrm{i}^{p/2} \sum_{1 \leq i_1 < \cdots < i_p \leq \ms{M}} 
    J_{i_1 \cdots i_p} \, \psi_{i_1} \cdots \psi_{i_p}\, ,
\end{equation}
where the couplings $J_{i_1 \cdots i_p}$ are drawn from a Gaussian ensemble with zero mean and variance equal to the inverse binomial coefficient:
\begin{equation}
    \langle J_{i_1 \cdots i_p} \rangle_J = 0,
    \qquad
    \langle J_{i_1 \cdots i_p}^2 \rangle_J = \binom{\ms{M}}{p}^{-1}.
\end{equation}
Here $\langle \cdot \rangle_J$ denotes the ensemble average over random couplings. In the planar limit, the model becomes exactly solvable upon taking the double-scaling limit $\ms{M},p\to \infty$ with $\lambda \coloneqq 2p^2/\ms{M}$ fixed. This regime is referred to as the double-scaled SYK (DSSYK) model \cite{Cotler:2016fpe, Berkooz:2018qkz, Berkooz:2018jqr}; for a recent review, see \cite{BerkoozReview}. Using transfer-matrix techniques, the expectation value of the partition function $\langle \Tr e^{-\beta H} \rangle_J$ of DSSYK can be computed explicitly \cite{Erd_s_2014, Berkooz:2018qkz} as
\begin{equation}\label{eq:disk:amplitude}
    \Braket{ \Tr{e^{-\beta H}} }_J
    =
    \int_0^\pi \frac{d\theta}{2\pi} \,
        (q;q)_\infty
        (e^{2i\theta}; q )_\infty
        (e^{-2i\theta}; q )_\infty
        \,
    e^{-\beta E(\theta)} ,
\end{equation}
where $q = e^{-\lambda}$, $E(\theta) = -2\cos\theta/\sqrt{1-q}$ and $(x;q)_\infty = \prod_{l=0}^{\infty}(1-xq^l)$ denotes the $q$-Pochhammer symbol. 
The authors of \cite{JKMS23} observed that the expectation value \labelcref{eq:disk:amplitude} can equivalently be expressed as the genus-zero, one-point function $\langle \mathrm{Tr}\, e^{\beta M}\rangle_{g=0}$ of the matrix model with even potential 
\begin{equation} \label{eq:dssykPot}
    V_{q}(M)
    \coloneqq
    \sum_{k=1}^{\infty} 
        \frac{(-1)^{k+1}}{k}
        q^{k(k+1)/2} ( 1+q^{-k} ) \,
        T_{2k}\biggl( \frac{\sqrt{1-q}}{2} M \biggr) .
\end{equation}
More explicitly, the disorder-averaged amplitude \labelcref{eq:disk:amplitude} can be recast as an expectation value supported by the large $\ms{N}$ eigenvalue distribution
\begin{equation}
    \varrho_0(x)
    =
    \frac{1}{2\pi \sqrt{a^2 - x^2}} \,
    (q;q)_\infty
    (e^{2i\theta};q)_\infty
    (e^{-2i\theta};q)_\infty ,
    \qquad
    a \coloneqq \frac{2}{\sqrt{1-q}}
\end{equation}
for $x = -E(\theta)$, as
\begin{equation}
    \Braket{ \Tr{e^{-\beta H}} }_J
    =
    \int_{-a}^{a}
    dx\, \varrho_0(x)\, e^{\beta x}.
\end{equation}
The spectral curve of the corresponding matrix model thus reads \cite{Oku23}:
\begin{equation}\label{eq:DSSYK}
    \mc{S}^{\textup{DSSYK}}
    =
	\begin{cases}
		x(z) = \dfrac{a}{2}\left( z + \dfrac{1}{z} \right), \\[2ex]
		y(z) = \dfrac{1}{a} \left( z - \dfrac{1}{z} \right)
		\prod_{k \ge 1} (1-q^k)(1-z^2q^k)(1-z^{-2}q^k).
	\end{cases}
\end{equation}
Note that, as in the relation between the SSS matrix model and standard SYK \cite{SSSJTmatrix}, only the disk one-point function matches between the DSSYK matrix integral and DSSYK itself.

In this section, we analyze this DSSYK spectral curve and explicitly compute the associated ABCD of \cref{subsec:ABCD}. This provides a recursion à la Mirzakhani for the pruned correlators of the double-scaled SYK matrix model:
\begin{equation}
    N_{g,n}^{\textup{DSSYK}}(b_1,\ldots,b_n;q)
    =
    \left\langle \prod_{i=1}^n \frac{1}{b_i} \NTr{M^{b_i}} \right\rangle_{\!\!g,\cc}^{\!\!\textup{DSSYK}}.
\end{equation}
We have emphasized the $q$-dependence of the pruned correlators through the DSSYK matrix model potential \labelcref{eq:dssykPot}. Setting $q=0$ recovers Norbury's discrete volumes, as the potential reduces to that of the GUE: $\lim_{q \to 0} V_{q}(M) = M^2/2$. In this sense, $N_{g,n}^{\textup{DSSYK}}$ can be seen as a $q$-deformation of the lattice point counting on the moduli space of curves.

As a powerful application of our recursion, we analyze a more complicated limit, tuning both $q \to 1$ in the matrix model potential while simultaneously rescaling the powers of the traces in the correlators. We will show that the pruned correlators converge to the Weil--Petersson volumes, confirming a conjecture by Okuyama \cite[equation~(6.5)]{Oku23}.

Throughout this section, $q$ is assumed to lie in the interval $[0,1)$, and we use the shorthand notation $(q)_\infty \coloneqq (q;q)_\infty = \prod_{k \ge 1} (1-q^k)$, also known as Euler's function.

\subsection{A discrete \texorpdfstring{$q$}{q}-Mirzakhani recursion}
Following \cref{sec:discrete:rec}, the spectral curve \labelcref{eq:DSSYK} corresponds to a deformation of the GUE curve provided by
\begin{equation}
    \gamma
    = \frac{1}{\sqrt{1-q}},
    \qquad
	Q_q(z)
    =
    (1-q) \prod_{k \ge 1} (1-q^k)(1-z^2q^k)(1-z^{-2}q^k) .
\end{equation}
The partial fraction decomposition of $1/Q_q$ is in fact known, and given for instance in \cite[page~136]{TM02}:
\begin{equation}\label{eq:Qq}
	\frac{1}{\gamma^2 Q_q(z)}
	=
	\frac{1}{(q)_{\infty}^{3}}
	-
	\frac{1}{(q)_{\infty}^{3}} \left( z - \frac{1}{z} \right)^2
	\sum_{k \ge 1}
		(-1)^k
		q^{\frac{k(k+1)}{2}} (1 + q^k) 
		\frac{1}{( 1 - z^2 q^k ) ( 1 - z^{-2} q^k )}.
\end{equation}
This corresponds to the example analyzed in \cref{{app:ex}}, with constants $2\sigma = (q)_{\infty}^{-3}$, $\alpha_k = q^{k/2}$ and $A_k = (q)_{\infty}^{-3} (-1)^k q^{\frac{k(k+1)}{2}} (1 + q^k)$. After some algebraic manipulation explained in \cref{app:cancellation}, we find that the building-block function reads\footnote{
    In this case, another justification is due, since $1/Q_{q}(z)$ from \labelcref{eq:Qq} has infinitely many poles. The main idea is that its partial fraction decomposition can be well approximated by a sequence of rational functions, in the same way that $\frac{\pi z}{\sin(\pi z)} \approx 1 -2z^2\sum_{k=1}^N \frac{(-1)^k}{k^2 - z^2}$. More precisely, for fixed $|q| < 1$, the quantity
    \begin{equation*}
	   H_k
	   =
	   (-1)^{k+1} q^{\frac{k(k+1)}{2}}
		\frac{q^{-\frac{k\ell}{2}}}{1 - q^k}
    \end{equation*}
    decays super-exponentially in $k$: there exists $0 < \rho < 1$ and $C>0$ such that $|H_k| \le C \, \rho^{k^2}$. Therefore the contour integral computing $H$ can be exchanged with the series in $k$ thanks to absolute convergence.
}
\begin{equation}\label{eq:H:qWP}
	H_q(\ell)
	=
	\frac{2}{(q)_{\infty}^{3}}
    \sum_{k \ge 1}
		(-1)^{k+1} q^{\frac{k(k+1)}{2}}
		\frac{q^{-\frac{k\ell}{2}}}{1 - q^k} .
\end{equation}
This gives the following $q$-deformations of the Mirzakhani kernels $B$ and $C$ and the initial data $A$ and $D$:
\begin{equation}\label{eq:ABCD:qWP}
\begin{aligned}
	A_q(b_1,b_2,b_3)
	& = \frac{1+(-1)^{b_1+b_2+b_3}}{2(q)_{\infty}^{3}}, \\
	B_q(b,b',\beta)
	& =
	\frac{1}{2b} \Bigl(
		H_q(b + b' - \beta)
		-
		H_q(-b - b' - \beta) \\
	& \qquad\qquad\qquad\qquad
		+
		H_q(b - b' - \beta)
		-
		H_q(-b + b' - \beta)
	\Bigr) , \\
	C_q(b,\beta,\beta')
	& =
	\frac{1}{b} \Bigl(
		H_q(b - \beta - \beta')
		-
		H_q(-b - \beta - \beta')
	\Bigr) , \\
	D_q(b)
	& =
	\frac{1+(-1)^{b}}{2 (q)_{\infty}^{3}} \biggl( \frac{b^2 - 4}{48} + \frac{\zeta_q(2)}{2} \biggr) .
\end{aligned}
\end{equation}
The $D$-term follows from the fact that, for even potential, the contributions from the first derivative in \cref{eq:ABCD} vanish while the contributions from the second derivative equal
\begin{equation}
	\frac{d^2}{dz^2}\frac{(q)_{\infty}^{3}}{\gamma^2 Q_q(z)} \bigg|_{z = \pm 1}
	=
	8
	\sum_{k \ge 1} (-1)^{k+1} q^{\frac{k(k+1)}{2}} \frac{1 + q^k}{(1 - q^k)^2}
	=
	8 \zeta_q(2).
\end{equation}
The first equality is a direct evaluation of the partial fraction decomposition. The last equality is shown in \cite[corollary~1.1]{HH15}. Here $\zeta_q(s)$ is the $q$-analog of the Riemann zeta function:
\begin{equation}
	\zeta_q(s)
    \coloneqq
    \sum_{k \ge 1} \frac{q^{\frac{ks}{2}}}{(1 - q^k)^s} .
\end{equation}
In the second installment of this paper, we will study more generally the structural dependence of the pruned DSSYK correlators on even values of the $q$-zeta function.

To sum-up, we have the following discrete $q$-analog of Mirzakhani's recursion.

\begin{proposition}\label{prop:qMir}
    For $2g-2+n > 1$ and $b_1 + \cdots + b_n$ even, the pruned DSSYK correlators satisfy the recursion relation
	\begin{equation}
    \begin{split}
    	N^{\textup{DSSYK}}_{g,n}(b_1,\ldots,b_n;q)
    	&=
    	\sum_{\beta > 0}
    		  \beta \, B_q(b_1,b_m,\beta)
    		  N^{\textup{DSSYK}}_{g,n-1}(\beta,b_2,\dots, \widehat{b_m},\dots, b_n;q) \\
    	& \quad
    	+ \frac{1}{2} \sum_{\beta,\beta' > 0}
    		\beta \beta' \, C_q(b_1,\beta,\beta') \Bigg(
    			N^{\textup{DSSYK}}_{g-1,n+1}(\beta,\beta',b_2,\dots,b_n;q) \\
    	& \qquad\qquad
    			+
    			\sum_{\substack{ h + h' = g \\ J \sqcup J' = \{2,\dots,n\} }}^{\textup{stable}}
    				N^{\textup{DSSYK}}_{h,1+|J|}(\beta,b_{J};q)
    				N^{\textup{DSSYK}}_{h',1+|J'|}(\beta',b_{J'};q)
    	\Bigg),
    \end{split}
    \end{equation}
    with $B_q$ and $C_q$ as in \cref{eq:ABCD:qWP}. Together with the initial data $A_q = N_{0,3}^{\textup{DSSYK}}$ and $D_q = N_{1,1}^{\textup{DSSYK}}$, this recursion uniquely determines all correlators.
\end{proposition}

Although this follows as a straightforward consequence of the general \cref{thm:A}, we believe it is of independent interest to both physicists and mathematicians, especially in light of the considerations outlined in the discussion section.

\subsection{Proof of Okuyama's conjecture}
Recall the notation $q = e^{-\lambda}$ for the double-scaling parameter of the underlying DSSYK model. The goal of this section is to show that, as we send $b_i \rightarrow \infty,\lambda \rightarrow 0$ keeping $\lambda \, b_{i}=L_{i}$ fixed, the discrete volumes asymptote to the Weil--Petersson volumes:
\begin{equation}\label{eq:Oku:conj}
	\lim_{\lambda \to 0^+} \
        (2 (q)_{\infty}^3 )^{2g-2+n}
		\lambda^{2(3g-3+n)}
		N_{g,n}^{\textup{DSSYK}}\Bigl(
            \frac{L_1}{\lambda},\ldots,\frac{L_n}{\lambda};q=e^{-\lambda}
        \Bigr)
	=
	2 \cdot V^{\textup{WP}}_{g,n}(L_1,\ldots,L_{n})
\end{equation}
whenever the sum of the $L_i/\lambda \in \mathbb{Z}_{+}$ is an even integer (otherwise the correlator vanishes since the potential is even). To prove the limit, we proceed similarly to the BMN-like limit by induction on $2g-2+n$. 

\paragraph{Base case.}
We first show that the initial conditions of the recursion satisfy Okuyama's conjecture. The $(0,3)$ case in \cref{eq:ABCD:qWP} is straightforward. Using $\zeta_{q}(2) \sim \lambda^{-2} \zeta(2)$ as $\lambda \to 0^+$, the $(1,1)$ case also flows to the Weil--Petersson result: 
\begin{equation}
	2(q)_{\infty}^3 \, \lambda^2 \,
    N_{1,1}^{\textup{DSSYK}}\Bigl(\frac{L_1}{\lambda};q=e^{-\lambda}\Bigr)
	\sim
	2 \left( \frac{L_1^2}{48} + \frac{\zeta(2)}{2} \right)
	=
	2 \cdot V_{1,1}^{\textup{WP}}(L_1) .
\end{equation}

\paragraph{Induction step.}
Our recursion relation from \cref{prop:qMir} provides the induction step. By induction hypothesis, all \smash{$N_{g',n'}^{\textup{DSSYK}}$} for fixed $2g'-2+n'<2g-2+n$ satisfy Okuyama's conjecture. We thus only need to consider how the recursion kernels $B_q$ and $C_q$ behave in the combined limit to show all higher correlators \smash{$N_{g,n}^{\textup{DSSYK}}$} also flow to the continuum Weil--Petersson volumes. Since these are built out of the basic building-block function we computed in \labelcref{eq:H:qWP}, the entire computation reduces to studying the limiting behavior of $H_q$. The crucial point is to note that the explicit $H_q$ is a $q$-analog of $2 \log(1 + e^{\ell/2})$, the building-block function appearing in the Weil--Petersson volumes\footnote{
    As a side note, it is curious to find here a complete Fermi--Dirac integral: $2 \log(1 + e^{\ell/2}) = \int_{0}^{\infty} \frac{dt}{1 + e^{(t-\ell)/2}}$. Thus, the building-block in the DSSYK case is a $q$-analog of a complete Fermi--Dirac integral. We do not understand its meaning, if any.}, cf. \cref{thm:Mir}:
\begin{equation}
	\lim_{\lambda \to 0^+} \
        (q)_{\infty}^3 \, \lambda \,
        H_q\Bigl( \frac{\ell}{\lambda} \Bigr)
	=
	2 \sum_{k \ge 1} \frac{(-1)^{k+1}}{k} e^{\ell k/2}
	=
	2 \log(1 + e^{\ell/2}).
\end{equation}
This implies that the $q$-kernels behave as  
\begin{equation}
\begin{aligned}
	(q)_{\infty}^3 \, B_q\Bigl(
        \frac{L_1}{\lambda},\frac{L_m}{\lambda},\frac{\ell}{\lambda}
    \Bigr)
	&\sim
	B^{\textup{hyp}}(L_1,L_m,\ell), \\
	(q)_{\infty}^3 \, C_q\Bigl(
        \frac{L_1}{\lambda},\frac{\ell}{\lambda},\frac{\ell'}{\lambda}
    \Bigr)
	&\sim
	C^{\textup{hyp}}(L_1,\ell,\ell').
\end{aligned}
\end{equation}
The remaining details closely parallel the proof of the BMN-like limit, see \cref{sec:BMN}. This concludes the proof of \cref{thm:C}.

\section{Discussion \& outlook}

\textbf{CohFT perspective.} In the second installment of this paper, we provide an in\-ter\-sec\-tion-theoretic expression of the pruned correlators $N_{g,n}$ by deriving an operator dictionary between matrix model traces and cohomology classes on $\Mbar_{g,n}$. This generalizes the relation between matrix model observables and cohomology classes on $\Mbar_{g,n}$, extending beyond the classical Kontsevich and Weil--Petersson cases, as well as the usual double-scaling regime. Using the language of cohomological field theory, we will see precisely how the matrix model potential becomes encoded in the integrand on moduli space.

\medskip

\textbf{Relation to sine-dilaton gravity.} As mentioned in the introduction, the ETH matrix model for DSSYK discussed in \cref{sec:DSSYK:rec} plays to DSSYK the same role that the SSS matrix integral plays to standard SYK. The SSS model admits a dual description in terms of JT gravity \cite{SSSJTmatrix}. This naturally raises the question: what is the gravity dual of the DSSYK matrix integral studied here? In a series of recent works, Blommaert and collaborators have proposed \emph{sine-dilaton gravity} as the natural counterpart. In their description, the discreteness of the boundary lengths $b_i$ arises from a quantization condition (see, e.g. \cite[equation~(2.16)]{BLMPP25}), discussed primarily at genus zero. Their equation~(5.1), which relates the insertion of geodesics boundary of length $b$ to matrix model quantities, agrees with our use of pruned traces, given the relation to Chebyshev polynomials in \cref{eq:prune:Cheby}.

\medskip

\textbf{Relation to recent works.} There has been a flurry of recent developments connecting double-scaled matrix integrals to new low-dimensional string theories. In particular, the amplitudes of the Virasoro Minimal String (VMS) \cite{CollierVMS,CliffVMS, CliffSuperVMS, RangamaniSuperVMS,IoannisVMS} and Complex Liouville String (CLS) \cite{Collier:CLS,Collier:CLS:world,Collier:CLS:matrix,Collier:CLS:nonper,Collier:path} also define certain continuous deformations of the Weil--Petersson volumes or variations thereof. For example, in the VMS, the deformation is characterized by the $\mathsf{b}$ parameter of the underlying Liouville theories. Those volumes do not agree with any considered here, in part because the dual matrix descriptions are all double-scaled: their boundary lengths are continuous. 

A related yet distinct line of development has been pursued by Do and Norbury, who have recently introduced a $q$-analog of the Weil--Petersson volumes defined via a continuous recursion \cite{DN}. Their construction is closely related to our pruned correlators, but follow from a particular limit of ours. More precisely, their volumes arise as a \emph{top-degree limit} of those considered in \cref{sec:DSSYK:rec}, obtained by assigning $\deg b_i = 1$ and $\deg \zeta_q(d) = d$ and keeping only the leading terms. For example, in the genus-$1$, $1$-point case, their volume is $\frac{L^2}{48} + \zeta_q(2)$, while our corresponding pruned correlator is $\frac{L^2 - 4}{48} + \zeta_q(2)$. Importantly, their $q$-volumes are labeled by continuous boundary lengths, they cannot reproduce the lattice count on moduli space obtained in the $q\to 0$ limit, and do not agree with the correlators of the DSSYK matrix model considered by Okuyama \cite{Oku23}. 

As for non-double-scaled theories, beyond the discrete Mirzakhani-type recursion derived in this paper, there must also exist discrete analogues of the string and dilaton equations. In fact, Okuyama \cite{OkuyamaCap} has recently applied the discrete Laplace transform of \cref{sec:discrete:rec} to the genus-zero, one-point function to introduce what he calls the cap amplitude. He then shows that pruned correlators obey a discrete dilaton equation, thus illustrating another discrete facet of topological recursion. It would be equally interesting to investigate whether a discrete analogue of the string equation can be formulated within this framework.

\medskip

\textbf{Beyond perturbative discreteness.} Although we have referred to pruned matrix model correlators as \emph{discrete volumes}, the picture of a weighted count of isolated points on moduli space breaks down non-perturbatively in the interaction couplings (at each order in $1/\ms{N}$). It would be very interesting to make the resulting picture precise. We expect the integrand on moduli space to be sharply peaked around these points at weak coupling, with a characteristic width set by the coupling.

\medskip

\textbf{What do the discrete $q$-WP volumes count?} While we have shown that the pruned correlators in the DSSYK matrix model converge to the standard Weil--Petersson volumes in the $q \to 1$ limit, we do not yet have an independent geometric definition of these discrete analogs from the point of view of $\M_{g,n}$. In particular, can we assign a genuine counting problem to the $q$-parameter? In the actual DSSYK model, the power of $q$ enumerates intersections of chord diagrams. How is this combinatorial interpretation reflected in the ETH matrix integral description? This question should prove mathematically very rich.

\medskip

\textbf{A unification of moduli space volumes from DSSYK.} Finally, we note that the discrete $q$-analogs of the Weil--Petersson volumes derived from the DSSYK matrix model unify three major notions of volumes of the moduli space of Riemann surfaces that have shaped the field of algebraic geometry over the past three decades. They can all be recovered in appropriate limits. The structure of these limits can be summarized schematically as follows.
\begin{equation}
\begin{tikzcd}[column sep=6em, row sep=6em]
	{N^{\textup{DSSYK}}_{g,n}(b_1,\ldots,b_n;q)} & {V^{\textup{WP}}_{g,n}(L_1,\ldots,L_n)} \\
	{N^{\textup{Nor}}_{g,n}(b_1,\ldots,b_n)} & {V^{\textup{Kon}}_{g,n}(\ell_1,\ldots,\ell_n)}
	\arrow["{\substack{q \to 1, \; b_i \to \infty \\ \log(q^{-1})b_i = L_i}}", from=1-1, to=1-2]
	\arrow["{q \to 0}"', from=1-1, to=2-1]
	\arrow["\text{BMN}"{marking, allow upside down,description}, dotted, from=1-1, to=2-2]
	\arrow["{\substack{s \to 0, \; L_i \to \infty \\ s L_i = \ell_i}}", from=1-2, to=2-2]
	\arrow["{\substack{t \to 0, \; b_i \to \infty \\ tb_i = \ell_i}}"', from=2-1, to=2-2]
\end{tikzcd}
\end{equation}
As a simple illustration, consider the case of genus-$1$, $1$-point.
\begin{equation}
\begin{tikzcd}
	{\frac{1+(-1)^{b}}{2} \Bigl( \frac{b^2 - 4}{48} + \frac{\zeta_q(2)}{2} \Bigr)} & {\frac{L^2}{48} + \frac{\zeta(2)}{2}} \\
	{\frac{1+(-1)^{b}}{2} \frac{b^2 - 4}{48}} & {\frac{\ell^2}{48}}
	\arrow[from=1-1, to=1-2]
	\arrow[from=1-1, to=2-1]
	\arrow[dotted, from=1-1, to=2-2]
	\arrow[from=1-2, to=2-2]
	\arrow[from=2-1, to=2-2]
\end{tikzcd}
\end{equation}
It is remarkable that the DSSYK matrix model appears to encode so much of the geometry of the moduli space of Riemann surfaces.

\section*{Acknowledgments}
A.G.~is supported by an ETH~Fellowship (22-2~FEL-003) and a Hermann–Weyl Instructorship from the Forschungsinstitut für Mathematik at ETH~Zürich. He thanks Gaëtan Borot, Séverin Charbonnier, Danilo Lewański, and Campbell Wheeler for several discussions on volumes of moduli spaces in connection with Laplace transform and topological recursion. P.M is supported by the Swiss National Science Foundation (SNSF). E.A.M.~is supported by a SwissMAP Research Fellowship. He would like to thank Scott Collier for patient explanations of the VMS and CLS amplitudes, Lorenz Eberhardt for conversations on stringy origins of these recursion relations (in particular for teaching him about higher equations of motion in Liouville theory),  David Gross, Clifford Johnson, Rishabh Kaushik, Shota Komatsu, Beat Nairz, Debmalya Sarkar and Steve Shenker for discussions relating to perturbative discreteness, Adam Levine on connections to sine-dilaton gravity, Ahmed Almheiri, Daniel Jafferis and Julian Sonner on ETH matrix models. He would also particularly like to thank Matthias Gaberdiel, Rajesh Gopakumar and Wei Li for related collaborations and important conversations on stringy realizations of discrete volumes. This work was performed in part at the Aspen Center for Physics, which is supported by National Science Foundation grant PHY-2210452. 

\newpage
\appendix
\section{An example towards DSSYK} 
\label{app:ex}
In view of the DSSYK spectral curve discussed in \cref{sec:DSSYK:rec}, in this appendix we consider the simple case in which the function $Q$, determining the matrix model spectral curve, takes the form
\begin{equation} 
	\frac{1}{\gamma^2 Q(z)}
	=
	2\sigma
	-
	\left( z - \frac{1}{z} \right)^2 \frac{A}{(1-z^2\alpha^2)(1-z^{-2}\alpha^2)} ,
\end{equation}
where $\sigma$, $A$ and $\alpha$ are arbitrary scalars with $|\alpha| < 1$. Notice that $\sigma = \frac{1}{2\gamma^2 Q(\pm 1)}$. Moreover, $Q$ is even, corresponding to an even potential in the underlying matrix model.

The goal of this appendix is to compute the associated building-block function $H$, following the recipe illustrated in \cref{subsec:ABCD}. The computation of the contour integral defining $H$ is split as in \cref{eq:F:G} into the contribution from the residue at $z = 0$ (the $F$-term), and that from the zeros of $Q$ lying inside the unit circle (the $G$-term).

The function $F$ is given by computing the Taylor expansion coefficients of $1/Q$ at $z = 0$, and reads
\begin{equation}\label{eq:F:toy} 
	F(\ell)
	=
    \frac{1 + (-1)^\ell}{2}
    \left(
	   2\sigma\rho(\ell)
	   +
	   \frac{2A}{\alpha^2}
	   \frac{\alpha^\ell - \alpha^{-\ell}}{\alpha^2 - \alpha^{-2}}
        \theta(\ell)
	\right)
	.
\end{equation}
Recall that $\theta(\ell)$ is the Heaviside theta function, and $\rho(\ell) = \ell \theta(\ell)$ is the ramp function. The function $G$ is given instead from the residues at $z = \pm \alpha$, which are explicitly evaluated as
\begin{equation}\label{eq:G:toy} 
	G(\ell)
	=
	\frac{1+(-1)^\ell}{2}
	\frac{2A}{\alpha^2}
	\frac{\alpha^{-\ell}}{\alpha^2 - \alpha^{-2}} .
\end{equation}
Since $Q$ is even, only even values of $\ell$ enter in the recursion formula. Thus, we can drop the parity indicator in $F$ and $G$. Summing them up, we find that the building-block function $H$ is given as
\begin{equation}
	H(\ell)
	=
	2\sigma\rho(\ell)
	+
	\frac{2A}{\alpha^2}
		\frac{\alpha^{|\ell|}}{\alpha^2 - \alpha^{-2}}
	.
\end{equation}

Notice that a similar equation holds for $Q$ with more zeros: for
\begin{equation}\label{eq:M:toy}
	\frac{1}{\gamma^2 Q(z)}
	=
	2\sigma
	-
	\left( z - \frac{1}{z} \right)^2
	\sum_{k=1}^K \frac{A_k}{(1-z^2\alpha_k^2)(1-z^{-2}\alpha_k^2)} ,
\end{equation}
with $\sigma = \frac{1}{2\gamma^2 Q(\pm 1)}$ and $\alpha_k$ lies inside the unit circle, then
\begin{equation}\label{eq:H:toy}
	H(\ell)
	=
	2\sigma\rho(\ell)
	+
	\sum_{k=1}^K
		\frac{2A_k}{\alpha_k^2}
			\frac{\alpha_k^{|\ell|}}{\alpha_k^2 - \alpha_k^{-2}}
	.
\end{equation}

\section{Discrete recursion: proofs} \label{app:proofsRecRel}
In what follows, set
\begin{equation}\label{eq:Laplace:WB:WC}
\begin{aligned}
	\omega_B(z_1)
	&\coloneqq
	\sum_{b_1 > 0} b_1 \, N_B(b_1) z_1^{b_1-1} dz_1, \\
	\omega_C(z_1,z_2)
	&\coloneqq
	\sum_{b_1,b_2 > 0} b_1 b_2 \, N_C(b_1,b_2) z_1^{b_1-1} z_2^{b_2-1} dz_1 dz_2,
\end{aligned}
\end{equation}
for $\omega_B$ and $\omega_C$ symmetric meromorphic differentials on $\P^1$ with poles at $\pm 1$ only and satisfying $\omega_B(z_1^{-1}) = - \omega_B(z_1)$ and $\omega_C(z_1^{-1},z_2) = - \omega_C(z_1,z_2)$.
In other words, $\omega_B$ and $\omega_C$ are the discrete Laplace transforms of $N_B$ and $N_C$, respectively.
More generally, set
\begin{equation}
	\mathfrak{L}\bigl[ N \bigr](z_1,\ldots,z_n)
	\coloneqq
	\sum_{b_1,\ldots,b_n > 0} N(b_1,\ldots,b_n) \prod_{i=1}^n b_i z_i^{b_1 -1} dz_i
\end{equation}
for the discrete Laplace transform of a quasi-polynomial function $N$. Recall the kernel $B^{\textup{comb}}$ and $C^{\textup{comb}}$ from the recursion for the Kontsevich volumes, \cref{thm:Kon} defined in terms of the ramp function $\rho(\ell) = \ell \theta(\ell)$.

\begin{lemma}\label{lem:GUE}
	The following equations hold:
	\begin{align}
		\notag
		\de_{z_2}
		\Bigg[
			\bigg(
				\frac{z_1^3}{(1 - z_1^2)^2 dz_1} \omega_B(z_1)
				-
				\frac{z_2^3}{(1 - z_2^2)^2 dz_2} & \omega_B(z_2)
			\bigg)
			\bigg(
				\frac{1}{z_1 - z_2}
				-
				\frac{1}{z_1 - z_2^{-1}}
			\bigg)
		\Bigg] dz_1 dz_2
		= \\
		\label{eq:B0}
		&
		\mathfrak{L}\Bigg[
			\sum_{\substack{\beta > 0 \\ 2 \mid b_1 + b_2 - \beta}} \beta \, B^{\textup{comb}}(b_1,b_2,\beta') N_B(\beta) 
		\Bigg](z_1,z_2),
		\\
		\label{eq:C0}
		\frac{z_1^3}{(1-z_1^2)^2 dz_1} \omega_C(z_1,z_1)
		=
		\frac{1}{2}
		&
		\mathfrak{L}\Bigg[
			\sum_{\substack{\beta,\beta' > 0 \\ 2 \mid \beta + \beta' - b_1}} \beta \beta' \, C^{\textup{comb}}(b_1,\beta,\beta') N_C(\beta,\beta') 
		\Bigg](z_1).
	\end{align}
	In the sums on the right-hand side, the condition $2 \mid \ell$ indicates that $\ell$ is an even integer.
\end{lemma}

\begin{proof}
	This is essentially contained in \cite[lemma~1]{Nor13}.
	We repeat the computation to illustrate the idea, starting with the right-hand side of \cref{eq:C0}.
	The basic strategy is simple: exchange the sum over $\beta$ and $\beta'$ with the sum over $b_1$ coming from the discrete Laplace transform:
	\begin{equation}
	\begin{split}
		\frac{1}{2}
		\mathfrak{L}\Bigg[
			\sum_{\substack{\beta,\beta' > 0 \\ 2 \mid \beta + \beta' - b_1}} \beta \beta' \, C^{\textup{comb}}(b_1,\beta,\beta') &N_C(\beta,\beta') 
		\Bigg](z_1)
		= \\
		&=
		\frac{1}{2}
			\sum_{b_1 > 0} z_1^{b_1-1}
				\sum_{\substack{\beta,\beta' > 0 \\ \beta + \beta' \le b_1 \\ 2 \mid \beta + \beta' - b_1}} \beta \beta' \, (b_1 - \beta - \beta') N_C(\beta,\beta') dz_1 \\
		&=
		\frac{1}{2}
		\sum_{\beta,\beta' > 0}
			\beta \beta' N_C(\beta,\beta')
			\sum_{\substack{b_1 \ge \beta+\beta' \\ 2 \mid \beta + \beta' - b_1}}
				(b_1 - \beta - \beta') z_1^{b_1-1} dz_1.
	\end{split}
	\end{equation}
	Summing up the geometric series and using the parity condition, we find that the innermost sum equals $2 \frac{z_1^3}{(1-z_1^2)^2} z_1^{\beta+\beta'-2} dz_1$. We conclude that
	\begin{equation}
	\begin{split}
		\frac{1}{2}
		\mathfrak{L}\Bigg[
			\sum_{\substack{\beta,\beta' > 0 \\ 2 \mid \beta + \beta' - b_1}} \beta \beta' \, C^{\textup{comb}}(b_1,\beta,\beta') N_C(\beta,\beta') 
		\Bigg](z_1)
		&=
		\frac{z_1^3}{(1-z_1^2)^2}
		\sum_{\beta,\beta' > 0}
			\beta \beta' N_C(\beta,\beta')
			z_1^{\beta+\beta'-2}  dz_1 \\
		&=
		\frac{z_1^3}{(1-z_1^2)^2 dz_1} \omega_C(z_1,z_1) .
	\end{split}
	\end{equation}
	A similar argument holds for \cref{eq:B0}, after splitting the sum into the different terms appearing in $B^{\textup{comb}}$.
\end{proof}

Notice that the parity conditions appearing in the lemma above can be neatly reabsorbed by taking the function $\frac{1+(-1)^{b}}{2} \rho(b)$ instead of simply $\rho(b)$ in the definition of the combinatorial kernels. We now consider a function $Q$ as in \cref{sec:discrete:rec} and a non-zero scalar $\gamma$. Define the expansion coefficients of $1/(\gamma^2 Q)$ around $z = 0$ as $1/\gamma^2 Q(z) \eqqcolon \sum_{b \ge 0} \mu(b) z^b$. Define the new kernels
\begin{equation}
\begin{split}
	\hat{B}(b_1,b_2,\beta)
	&\coloneqq
    \frac{1}{2b_1}
	\bigl(
		F(b_1 - b_2 - \beta) - F(-b_1 + b_2 - \beta) \\
		&\qquad\qquad\qquad
		+ F(b_1 + b_2 - \beta) - F(-b_1 -b_2 - \beta)
	\bigr), \\
	\hat{C}(b_1,\beta,\beta')
	&\coloneqq
    \frac{1}{b_1}
	\bigl(
		F(b_1 - \beta - \beta') - F(-b_1 - \beta - \beta')
	\bigr),
\end{split}
\end{equation}
where $F$ is the discrete convolution of $\frac{1+(-1)^{b}}{2} \rho(b)$ and $\mu(b)$, which can be easily re-written as a residue at $z = 0$, in accordance with \cref{eq:F:G}:
\begin{equation}
	F(\ell)
	\coloneqq
    \theta(\ell)
	\sum_{b = 0}^{\ell} \frac{1 + (-1)^{\ell - b}}{2} (\ell - b) \, \mu(b)
    =
    \Res_{z = 0} \frac{z^{1-\ell}}{(1 - z^2)^2 \gamma^2 Q(z)} dz.
\end{equation}
The following result is a simple consequence of \cref{lem:GUE} and the convolution-product property of the discrete Laplace transform, analogous to its continuous counterpart.

\begin{lemma}\label{lem:conv}
	The following equations hold:
	\begin{align}
		\notag
		\de_{z_2}
		\Bigg[
			\bigg(
				\frac{z_1^3}{(1 - z_1^2)^2 \gamma^2 Q(z_1) dz_1} \omega_B(z_1)
				&-
				\frac{z_2^3}{(1 - z_2^2)^2 \gamma^2 Q(z_2) dz_2} \omega_B(z_2)
			\bigg) \times
            \\
			\bigg(
				\frac{1}{z_1 - z_2}
				-
				\frac{1}{z_1 - z_2^{-1}}
			\bigg)
		\Bigg] dz_1 dz_2
		&=
		\mathfrak{L}\Bigg[
			\sum_{\substack{\beta > 0 \\ 2 \mid b_1 + b_2 - \beta}} \beta \, \hat{B}(b_1,b_2,\beta') N_B(\beta) 
		\Bigg](z_1,z_2),
		\\
		\frac{z_1^3}{(1-z_1^2)^2 \gamma^2 Q(z_1) dz_1} \omega_C(z_1,z_1)
		&=
		\frac{1}{2} \,
		\mathfrak{L}\Bigg[
			\sum_{\substack{\beta,\beta' > 0 \\ 2 \mid \beta + \beta' - b_1}} \beta \beta' \, \hat{C}(b_1,\beta,\beta') N_C(\beta,\beta') 
		\Bigg](z_1).
	\end{align}
\end{lemma}

We conclude the appendix with a third result, computing the $G$-contribution from \cref{eq:building:block}. To this end, we introduce the kernels
\begin{equation}
\begin{split}
	\check{B}_{\alpha}(b_1,b_2,\beta)
	&\coloneqq
	\frac{1}{2b_1}
	\bigl(
		G_{\alpha}(b_1 - b_2 - \beta) - G_{\alpha}(-b_1 + b_2 - \beta) \\
		&\qquad\qquad
		+ G_{\alpha}(b_1 + b_2 - \beta) - G_{\alpha}(-b_1 - b_2 - \beta)
	\bigr), \\
	\check{C}_{\alpha}(b_1,\beta,\beta')
	&\coloneqq
	\frac{1}{b_1}
	\bigl(
		G_{\alpha}(b_1 - \beta - \beta') - G_{\alpha}(-b_1 - \beta - \beta')
	\bigr),
\end{split}
\end{equation}
where $G_{\alpha}$ is given by a residue at a zero $\alpha$ of $Q$ inside the unit circle, as in \cref{eq:F:G}:
\begin{equation}
	G_{\alpha}(\ell)
	\coloneqq
	\Res_{z = \alpha} \frac{2 \, z^{1-\ell}}{(1-z^2)^2 \gamma^2 Q(z)} dz .
\end{equation}

\begin{lemma}\label{lem:zeros}
	The following equations hold:
	\begin{align}
		\notag
		\label{eq:Balpha}
		- 2\Res_{z = \alpha}
			\frac{K(z_1,z)}{\gamma^2 Q(z)}
			\left( \frac{dz dz_2}{(z-z_2)^2} + \frac{dz dz_2}{(1-z z_2)^2} \right)
			& \omega_B(z)
		= \\
		& \mathfrak{L}\Bigg[
			\sum_{\beta > 0} \beta \, \check{B}_{\alpha}(b_1,b_2,\beta) N_B(\beta) 
		\Bigg](z_1,z_2),
		\\
		\label{eq:Calpha}
		- 2\Res_{z = \alpha}
			\frac{K(z_1,z)}{\gamma^2 Q(z)} \omega_C(z,z)
		=
		\frac{1}{2}
		&
		\mathfrak{L}\Bigg[
			\sum_{\beta,\beta' > 0} \beta \beta' \, \check{C}_{\alpha}(b_1,\beta,\beta') N_C(\beta,\beta') 
		\Bigg](z_1).
	\end{align}
\end{lemma}

\begin{proof}
	We prove \cref{eq:Calpha}, with \cref{eq:Balpha} following by a similar strategy.
	Expanding $\omega_C$ on the left-hand side, we obtain
	\begin{equation}
		\sum_{\beta,\beta' > 0}
			\beta \beta'
			\left(
				\Res_{z = \alpha}
				\left(
					\frac{1}{z_1 - z^{-1}} - \frac{1}{z_1 - z}
				\right)
				\frac{z^{1 + \beta + \beta'} dz}{(1 - z^2)^2 \gamma^2 Q(z)}
			\right)
			N_{C}(\beta,\beta') dz_1.
	\end{equation}
	Since $|z_1| < |z|$, as $z_1$ is near $0$ while $z$ is around $\alpha$, the correct expansion of the geometric series is
	\begin{equation}
		\frac{1}{z_1 - z^{-1}} - \frac{1}{z_1 - z}
		=
		\sum_{b_1 > 0} ( z^{-b_1} - z^{b_1} ) z_1^{b_1-1}.
	\end{equation}
	Substituting this expansion into the above formula, we find
	\begin{multline}
		\frac{1}{2}
		\sum_{b_1>0}
		\left[
			\sum_{\beta,\beta' > 0}
				\beta \beta'
				\left(
					\Res_{z = \alpha}
					\frac{1}{b_1} ( z^{-b_1} - z^{b_1} )
					\frac{2 \, z^{1 + \beta + \beta'} dz}{(1 - z^2)^2 \gamma^2 Q(z)}
				\right)
			N_{C}(\beta,\beta')
		\right]
		b_1 z_1^{b_1 - 1} dz_1
		= \\
		\frac{1}{2}
		\mathfrak{L}\Bigg[
			\sum_{\beta,\beta' > 0} \beta \beta' \, \check{C}_{\alpha}(b_1,\beta,\beta') N_C(\beta,\beta') 
		\Bigg](z_1) .
	\end{multline}
\end{proof}

\section{Cancellations in the DSSYK kernel}
\label{app:cancellation}
In this appendix, we prove the cancellation appearing in the $F$-term of the DSSYK kernel. More precisely, consider the formula for the $F$-term in \labelcref{eq:F:toy} with  with $k \ge 1$ and constants $2\sigma = (q)_{\infty}^{-3}$, $\alpha_k = q^{k/2}$ and $A_k = (q)_{\infty}^{-3} (-1)^k q^{\frac{k(k+1)}{2}} (1 + q^k)$. Ignoring the parity indicator (which plays no role for even potential like in DSSYK matrix model), we find
\begin{equation}
	F_q(\ell)
	=
	\frac{\rho(\ell)}{(q)_{\infty}^3}
	+
	\frac{2\theta(\ell)}{(q)_{\infty}^3}
    \sum_{k \ge 1}
		(-1)^k
		q^{\binom{k}{2}} (1 + q^k) \frac{q^{\ell k/2} - q^{-\ell k/2}}{q^k - q^{-k}} .
\end{equation}
The goal of this section is to show that $F_q(\ell) = 0$ for all even values of $\ell$. This cancellation is the $q$-analog of
\begin{equation}
    \rho(\ell) + \Bigl( 2\log(1+e^{-\ell/2}) - 2\log(1+e^{\ell/2}) \Bigr) \theta(\ell) = 0,
\end{equation}
which appears in computations for the Weil--Petersson case. Hence, the building-block function $H_q$ for the DSSYK correlators coincide with the $G_q$ function, as given in \cref{eq:H:qWP}. The above cancellation follows from the following $q$-series identity.

\begin{lemma}
	The following holds
	\begin{equation}
		\label{eq:vanishing:even}
		\sum_{k \ge 1}
		(-1)^{k}
		q^{\binom{k}{2}} (1 + q^k)
			\frac{q^{mk} - q^{-mk}}{q^k - q^{-k}}
		=
		- m.
	\end{equation}
\end{lemma}

\begin{proof}
    First, rewrite the series as
	\begin{equation}
		\sum_{k \ge 1}
		(-1)^{k}
		q^{\binom{k}{2}} (1 + q^k)
			\frac{q^{mk} - q^{-mk}}{q^k - q^{-k}}
		=
		\sum_{k \ge 1}
			  (-1)^{k}
			q^{\frac{k(k+1)}{2} - mk}
			  \frac{1 - q^{2mk}}{1 - q^k} .
	\end{equation}
	Writing $\frac{1 - q^{2mk}}{1 - q^k} = \sum_{p = 0}^{2m-1} q^{pk}$, we deduce that
	\begin{equation}
		\sum_{k \ge 1}
		(-1)^{k}
		q^{\binom{k}{2}} (1 + q^k)
			\frac{q^{mk} - q^{-mk}}{q^k - q^{-k}}
		=
		\sum_{p = 0}^{2m-1} A_{m,p},
		\quad
		A_{m,p}
		\coloneqq
		\sum_{k \ge 1}
			  (-1)^{k}
			q^{\frac{k(k+1)}{2} - mk + pk}.
	\end{equation}
	We claim that $A_{m,p} + A_{m,2m-1-p} + 1 = 0$, which implies the result. Indeed, by relabeling the index of summation in the second sum as $k \mapsto -k$, we find
	\begin{equation}
	\begin{split}
		A_{m,p} + A_{m,2m-1-p} + 1
		&=
		\sum_{k \ge 1}
			  (-1)^{k}
			q^{\frac{k(k+1)}{2} - mk + pk}
		+
		\sum_{k \le -1}
			  (-1)^{k}
			q^{\frac{k(k+1)}{2} + mk + pk}
		+ 1 \\
		&=
		\sum_{k \in \Z}
			  (-1)^{k}
			q^{\frac{k(k+1)}{2} - mk + pk} \\
		&=
		(q;q)_{\infty} (q^{m-p};q)_{\infty} (q^{p-m+1};q)_{\infty},
	\end{split}
	\end{equation}
	where in the last line we have recognized Jacobi's triple product. Since $0 \le p \le 2m-1$, one of the last two Pochhammer symbols vanishes. This concludes the proof.
\end{proof}

\bibliographystyle{JHEP}
\bibliography{Bibliography.bib}

\end{document}